%% file: main.tex
\pdfoutput=1
\PassOptionsToPackage{dvipsnames}{xcolor}
\documentclass[acmsmall,10pt,screen]{acmart}
\input{packages}
\input{defs}

\input{pseudocode_macros}

\input{figures_macros}

\acmJournal{PACMPL}
\acmVolume{1}
\acmNumber{POPL} 
\acmArticle{1}
\acmYear{2024}
\acmMonth{1}
\acmDOI{} 
\startPage{1}

\setcopyright{none}

\author[S. Chakraborty]{Soham Chakraborty}
\orcid{0000-0002-4454-2050}
\affiliation{
  \institution{TU Delft}            
  \country{Netherlands}                    
}
\email{s.s.chakraborty@tudelft.nl} 

\author[S. Krishna]{Shankaranarayanan Krishna}
\orcid{0000-0003-0925-398X}
\affiliation{
  \institution{IIT Bombay}            
  \country{India}                    
}
\email{krishnas@cse.iitb.ac.in}

\author[U. Mathur]{Umang Mathur}
\orcid{0000-0002-7610-0660} 
\affiliation{
  \institution{National University of Singapore}            
  \country{Singapore}                    
}
\affiliation{
  \institution{IIT Bombay}            
  \country{India}                    
}
\email{umathur@comp.nus.edu.sg}          

\author[A. Pavlogiannis]{Andreas Pavlogiannis}
\orcid{0000-0002-8943-0722}
\affiliation{
  \institution{Aarhus University}            
  \country{Denmark}                    
}
\email{pavlogiannis@cs.au.dk}          

\usepackage{booktabs}   
\usepackage{subcaption} 

\begin{document}

\title{How Hard is Weak-Memory Testing?}



\sloppy
\input{abstract}
\maketitle

\input{intro}
\input{preliminaries}

\input{lower_threads_locations_rlx}
\input{lower_threads_locations_wra_ra_sra}

\input{other_models}
\input{bounded_values}

\input{conclusion}


\clearpage

\begin{acks}
Andreas Pavlogiannis was partially supported by a research grant (VIL42117) from VILLUM FONDEN.
S. Krishna was partially supported by the SERB MATRICS grant MTR/2019/000095. 
Umang Mathur was partially supported by a Singapore Ministry of Education (MoE) Academic Research Fund (AcRF) Tier 1 grant.
\end{acks}

\bibliography{references}

\newpage

\appendix

\input{app_lower_ra}
\input{app_lower_other_models}

\end{document}

%% file: packages.tex
\usepackage{tikz}

\usepackage{xspace}

\usepackage{booktabs}   
\usepackage{subcaption} 

\usepackage{enumitem} 
\usepackage{makecell} 
\usepackage{mathtools}
\usepackage{wrapfig}

\usepackage{amsmath}
\usepackage{pifont}
\usepackage{thm-restate}
\usetikzlibrary{arrows,automata,shapes,decorations,decorations.markings,calc, matrix,decorations.pathmorphing, patterns,backgrounds,shapes.misc,arrows.meta,positioning}

\usepackage{xspace}
\usepackage[ruled,vlined,resetcount,linesnumbered,noend]{algorithm2e}
\usepackage{parskip}
\usepackage{comment}

\usepackage{arydshln}
 \usepackage{multicol}
 \usepackage{multirow}

\usepackage{hyperref}
\usepackage[capitalise]{cleveref}

\usepackage{varwidth}

\usepackage{array}
\newcolumntype{H}{>{\setbox0=\hbox\bgroup}c<{\egroup}@{}}

\usepackage{ifthen}
\usepackage{stackengine}

\usepackage[normalem]{ulem} 

%% file: defs.tex
\Crefname{claim}{Claim}{Claims}

\newcounter{sarrow}

\def\XX{\kern-3pt/}

\newcommand{\Paragraph}[1]{\smallskip \noindent{\bf {#1}}}
\newcommand{\SubParagraph}[1]{\smallskip \noindent{\emph{#1}}}

\newcommand{\OrderedBefore}{\preceq}
\newcommand{\StrictOrderedBefore}{\prec}

\colorlet{colorPO}{darkgray!80!black}
\colorlet{colorRF}{blue}
\colorlet{colorCO}{red!80!black}
\colorlet{colorMO}{red!80!black}
\colorlet{colorOB}{orange}
\colorlet{colorFR}{orange}
\colorlet{colorECO}{orange}
\colorlet{colorCOM}{magenta}
\colorlet{colorSW}{teal}
\colorlet{colorHB}{green!40!black}
\colorlet{colorPPO}{magenta}
\colorlet{colorRSEQ}{green!40!black}
\colorlet{colorSC}{violet}
\colorlet{colorGW}{brown}
\colorlet{colorPSC}{violet}
\colorlet{colorREL}{olive}
\colorlet{colorWB}{olive}
\colorlet{colorDOB}{violet}
\colorlet{colorRMW}{brown}

\newcommand{\set}[1]{\{#1\}}
\newcommand{\setpred}[2]{\{#1 \,|\, #2\}}
\newcommand{\tup}[1]{\langle#1\rangle}

\renewcommand{\emptyset}{\varnothing}

\newcommand{\R}{\mathsf{R}}
\newcommand{\W}{\mathsf{W}}
\newcommand{\E}{\mathsf{E}}

\newcommand{\Locations}{\mathcal{V}}

\newcommand{\scmm}{\mathsf{SC}}
\newcommand{\wramm}{\mathsf{WRA}}
\newcommand{\ramm}{\mathsf{RA}}
\newcommand{\rlxmm}{\mathsf{Relaxed}}
\newcommand{\arlxmm}{\mathsf{Relaxed\text{-}Acyclic}}
\newcommand{\parlxmm}{\mathsf{Relaxed(\text{-}Acyclic)}}
\newcommand{\sramm}{\mathsf{SRA}}
\newcommand{\cmmm}{\mathsf{CM}}
\newcommand{\ccmm}{\mathsf{CC}}
\newcommand{\cvmm}{\mathsf{CCv}}

\newcommand{\tsomm}{\mathsf{TSO}}
\newcommand{\psomm}{\mathsf{PSO}}
\newcommand{\powermm}{\mathsf{POWER}}
\newcommand{\mmorder}{\preccurlyeq}
\newcommand{\po}{\mathsf{\color{colorPO}po}}

\newcommand{\rf}{\mathsf{\color{colorRF}rf}}

\newcommand{\mo}{\mathsf{\color{colorMO}mo}}

\newcommand{\mopartial}{\overline{\mo}}

\newcommand{\fr}{\mathsf{\color{colorFR}fr}}

\newcommand{\hb}{\mathsf{\color{colorHB}hb}}

\newcommand{\ob}[1]{\mathsf{{\color{colorOB}ob}_{#1}}}
\newcommand{\obOne}[1]{\mathsf{{\color{colorOB}ob}^1_{#1}}}

\def\WCoh{write-coherence}
\def\SWCoh{strong-write-coherence}
\def\RCoh{read-coherence}
\def\WRCoh{weak-read-coherence}
\def\RlxWCoh{relaxed-write-coherence}
\def\RlxRCoh{relaxed-read-coherence}
\def\PORF{porf-acyclicity}
\def\OB{ob-acyclicity}

\newcommand{\bfalse}{\bottom}
\newcommand{\btrue}{\top}

\newcommand{\irr}{\mathsf{irr}}
\newcommand{\acy}{\mathsf{acy}}

\newcommand{\op}{\mathsf{op}}
\newcommand{\tid}{\mathsf{tid}}

\newcommand{\event}{e}
\newcommand{\id}{\mathsf{id}}

\newcommand{\llab}{\mathsf{lab}}
\newcommand{\lloc}{\mathsf{loc}}
\newcommand{\ord}{\mathsf{ord}}
\newcommand{\Threads}{\mathcal{T}}
\newcommand{\Val}{\mathsf{Val}}
\newcommand{\bottom}{\perp}
\newcommand{\rd}{\mathtt{r}}
\newcommand{\wt}{\mathtt{w}}

\newcommand{\ex}{\mathsf{\color{black}X}}

\newcommand{\expartial}{\overline{\ex}}

\newcommand{\ov}[1]{\overline{#1}}

\newcommand{\NP}{\textsc{NP}}
\newcommand{\PComplexity}{\textbf{\textsc{P}}}
\newcommand{\NPComplete}{\textbf{\textsc{NP}-complete}}

\newcommand{\NumEvents}{n}
\newcommand{\NumLocations}{d}
\newcommand{\NumThreads}{\kappa}
\newcommand{\Phase}{\operatorname{phase}}
\newcommand{\Step}{\operatorname{step}}

\newcommand{\Memory}{\operatorname{M}}
\newcommand{\Executed}{\operatorname{P}}
\newcommand{\Buffers}{\operatorname{B}}
\newcommand{\LTS}{\mathcal{L}}

\newcommand{\Enabled}{\operatorname{enbl}}

\newcounter{ssarrow}
\newcommand\xrsquigarrow[1]{%
\stepcounter{ssarrow}%
\mathrel{\begin{tikzpicture}[baseline= {( $ (current bounding box.south) + (0,-0.5ex) $ )}]
\node[inner sep=.5ex] (\thessarrow) {$\scriptstyle #1$};
\path[draw,<-,decorate,
  decoration={zigzag,amplitude=0.7pt,segment length=1.2mm,pre=lineto,pre length=4pt}] 
    (\thessarrow.south east) -- (\thessarrow.south west);
\end{tikzpicture}}%
}

\newcommand{\LTo}[1]{\xrightarrow{#1}}
\newcommand{\LPath}[1]{\xrsquigarrow{#1}}

\newcommand{\Clause}{\mathcal{C}}
\newcommand{\CopyGadget}[2]{\operatorname{Copy}^{#1}_{#2}} 
\newcommand{\CopyGadgetDown}[2]{\operatorname{\ov{Copy}}^{#1}_{#2}} 
\newcommand{\AtMostOneGadget}[4]{\operatorname{AMOne}[#1]^{#2}_{#3,#4}}
\newcommand{\AtLeastOneGadget}[4]{\operatorname{ALOne}^{#1}_{#2,#3,#4}}

\newcommand{\InactiveWrites}{\operatorname{InactiveWr}}

\newcommand{\MemModel}{\mathcal{M}}

\newlist{compactenum}{enumerate}{3} 
\setlist[compactenum]{label=(\arabic*), nosep,leftmargin=*}
\crefname{compactenumi}{Item}{Items}

\newlist{compactitem}{itemize}{3} 
\setlist[compactitem]{label=\textbullet, nosep,leftmargin=*}
\crefname{compactitemi}{Item}{Items}

\newlist{compactdesc}{description}{3} 
\setlist[compactdesc]{nosep,leftmargin=1em}
\crefname{compactdesci}{Item}{Items}

\makeatletter
\newcommand{\labitem}[2]{%
\def\@itemlabel{\textbf{#1}}
\item
\def\@currentlabel{#1}\label{#2}}
\makeatother

%% file: pseudocode_macros.tex
\SetKwProg{myfun}{function}{}{}
\SetKwProg{myhandler}{handler}{}{}
\SetKwFunction{init}{Initialization}
\SetKwFunction{getLW}{getLastWriteBeforeOffline}
\SetKwFunction{poprop}{poPropagate}
\SetKwFunction{joinchangedindex}{pairwiseMaxIndexes}
\SetKwFunction{getLWB}{getLastWriteBefore}
\SetKwFunction{getLRB}{getLastReadBefore}
\SetKwFunction{getLRWB}{getLastReadWriteBefore}
\SetKwFunction{initWtLst}{initWtLst}
\SetKwFunction{initRdLst}{initRdLst}
\SetKwFunction{initWtRdLst}{initWtRdLst}
\SetKwFunction{get}{get}
\SetKwFunction{getHBTS}{getHBTimestamps}
\SetKwFunction{getTCPC}{getTopAndPositionInRFChain}
\SetKwFunction{checkRace}{checkRace}
\SetKwFunction{rdhandler}{readHandler}
\SetKwFunction{wthandler}{writeHandler}
\SetKwFunction{acqhandler}{acquire}
\SetKwFunction{relhandler}{release}
\SetKwInput{Input}{Input}
\SetKwInOut{Output}{Output}
\SetKw{Let}{let}
\SetKw{Break}{break}
\SetKw{NOT}{not}
\SetKw{declare}{declare}
\SetKw{exit}{exit}
\SetKw{Continue}{continue}
\SetKwFor{Foreach}{for each}{}{}%
\DontPrintSemicolon

\SetCommentSty{mycommfont}
\SetNoFillComment

\makeatletter
\newcommand*\bigcdot{\mathpalette\bigcdot@{.5}}
\newcommand*\bigcdot@[2]{\mathbin{\vcenter{\hbox{\scalebox{#2}{$\m@th#1\bullet$}}}}}
\makeatother

\newcommand{\true}{\top\!\!\!\!\top}

%% file: figures_macros.tex
\tikzstyle{annotateArrow} = [black]
\tikzstyle{event} = [rectangle, fill=white, inner sep=0, minimum height=4.5mm, minimum width=11mm, rounded corners]

\pgfdeclarelayer{bg}    
\pgfdeclarelayer{fg}    
\pgfsetlayers{bg,main,fg}  

\tikzset{
   every path/.style={>=stealth},
   po/.style={->,color=colorPO, thick},
   ppo/.style={->,color=colorPPO, thick},
   sw/.style={->,color=colorSW, thick},
   rf/.style={->,color=colorRF,dashed, thick},
   mrf/.style={->,color=colorRF,dashed, thick},   
   urf/.style={->,color=colorRF, thick},
   fr/.style={->,color=colorFR,dashed, thick},
   hb/.style={->,color=colorHB,thick, thick},
   mo/.style={->,color=colorMO,dotted, very thick},
   dob/.style={->,color=colorDOB, very thick},
   ob/.style={->,color=colorOB, very thick},
   rmw/.style={->,color=colorRMW,thick, thick},
   rseq/.style={->,color=colorRSEQ,thick,dotted, thick},
   com/.style={->,color=colorCOM,thick, thick},
}

\colorlet{ColorOne}{black}
\colorlet{ColorTwo}{gray!70!black!70}
\definecolor{ColorThree}{HTML}{559310}
\definecolor{ColorFour}{HTML}{7f42a9}
\definecolor{ColorFive}{HTML}{ef972d}
\definecolor{ColorSix}{HTML}{1e90ff}
\definecolor{ColorSeven}{HTML}{4c319e}
\definecolor{ColorEight}{HTML}{ec102f}
\definecolor{ColorNine}{HTML}{86dc3d}
\definecolor{ColorTen}{HTML}{208eb7}
\definecolor{ColorEleven}{HTML}{f6248f}
\definecolor{ColorTwelve}{HTML}{0000ff}

\newcommand{\VarXOne}{\textcolor{ColorOne}{x_1}}
\newcommand{\VarXTwo}{\textcolor{ColorTwo}{x_2}}

\newcommand{\VarYOne}{\textcolor{ColorThree}{y_1}}
\newcommand{\VarYTwo}{\textcolor{ColorFour}{y_2}}
\newcommand{\VarYThree}{\textcolor{ColorFive}{y_3}}
\newcommand{\VarYFour}{\textcolor{ColorSix}{y_4}}
\newcommand{\VarYFive}{\textcolor{ColorSeven}{y_5}}
\newcommand{\VarYSix}{\textcolor{ColorEight}{y_6}}
\newcommand{\VarYSeven}{\textcolor{ColorNine}{y_7}}
\newcommand{\VarYEight}{\textcolor{ColorTen}{y_8}}
\newcommand{\VarYEll}{\textcolor{ColorThree}{y_{\ell}}}

\newcommand{\VarZOne}{\textcolor{ColorThree}{z_1}}
\newcommand{\VarZTwo}{\textcolor{ColorFour}{z_2}}
\newcommand{\VarZThree}{\textcolor{ColorFive}{z_3}}
\newcommand{\VarZFour}{\textcolor{ColorSix}{z_4}}
\newcommand{\VarZFive}{\textcolor{ColorSeven}{z_5}}
\newcommand{\VarZSix}{\textcolor{ColorEight}{z_6}}
\newcommand{\VarZSeven}{\textcolor{ColorNine}{z_7}}
\newcommand{\VarZEight}{\textcolor{ColorTen}{z_8}}
\newcommand{\VarZEll}{\textcolor{ColorThree}{z_{\ell}}}

\newcommand{\VarA}[1]{\textcolor{ColorEleven}{a_{#1}}}
\newcommand{\VarB}{\textcolor{ColorTwelve}{b}}

\newcommand{\ReadXOne}[1]{\textcolor{ColorOne}{\rd(x_1, #1)}}
\newcommand{\WriteXOne}[1]{\textcolor{ColorOne}{\wt(x_1, #1)}}
\newcommand{\ReadXTwo}[1]{\textcolor{ColorTwo}{\rd(x_2, #1)}}
\newcommand{\WriteXTwo}[1]{\textcolor{ColorTwo}{\wt(x_2, #1)}}

\newcommand{\ReadYOne}[1]{\textcolor{ColorThree}{\rd(y_1, #1)}}
\newcommand{\WriteYOne}[1]{\textcolor{ColorThree}{\wt(y_1, #1)}}
\newcommand{\ReadYTwo}[1]{\textcolor{ColorFour}{\rd(y_2, #1)}}
\newcommand{\WriteYTwo}[1]{\textcolor{ColorFour}{\wt(y_2, #1)}}
\newcommand{\ReadYThree}[1]{\textcolor{ColorFive}{\rd(y_3, #1)}}
\newcommand{\WriteYThree}[1]{\textcolor{ColorFive}{\wt(y_3, #1)}}
\newcommand{\ReadYFour}[1]{\textcolor{ColorSix}{\rd(y_4, #1)}}
\newcommand{\WriteYFour}[1]{\textcolor{ColorSix}{\wt(y_4, #1)}}
\newcommand{\ReadYFive}[1]{\textcolor{ColorSeven}{\rd(y_5, #1)}}
\newcommand{\WriteYFive}[1]{\textcolor{ColorSeven}{\wt(y_5, #1)}}
\newcommand{\ReadYSix}[1]{\textcolor{ColorEight}{\rd(y_6, #1)}}
\newcommand{\WriteYSix}[1]{\textcolor{ColorEight}{\wt(y_6, #1)}}
\newcommand{\ReadYSeven}[1]{\textcolor{ColorNine}{\rd(y_7, #1)}}
\newcommand{\WriteYSeven}[1]{\textcolor{ColorNine}{\wt(y_7, #1)}}
\newcommand{\ReadYEight}[1]{\textcolor{ColorTen}{\rd(y_8, #1)}}
\newcommand{\WriteYEight}[1]{\textcolor{ColorTen}{\wt(y_8, #1)}}

\newcommand{\ReadZOne}[1]{\textcolor{ColorThree}{\rd(z_1, #1)}}
\newcommand{\WriteZOne}[1]{\textcolor{ColorThree}{\wt(z_1, #1)}}
\newcommand{\ReadZTwo}[1]{\textcolor{ColorFour}{\rd(z_2, #1)}}
\newcommand{\WriteZTwo}[1]{\textcolor{ColorFour}{\wt(z_2, #1)}}
\newcommand{\ReadZThree}[1]{\textcolor{ColorFive}{\rd(z_3, #1)}}
\newcommand{\WriteZThree}[1]{\textcolor{ColorFive}{\wt(z_3, #1)}}
\newcommand{\ReadZFour}[1]{\textcolor{ColorSix}{\rd(z_4, #1)}}
\newcommand{\WriteZFour}[1]{\textcolor{ColorSix}{\wt(z_4, #1)}}
\newcommand{\ReadZFive}[1]{\textcolor{ColorSeven}{\rd(z_5, #1)}}
\newcommand{\WriteZFive}[1]{\textcolor{ColorSeven}{\wt(z_5, #1)}}
\newcommand{\ReadZSix}[1]{\textcolor{ColorEight}{\rd(z_6, #1)}}
\newcommand{\WriteZSix}[1]{\textcolor{ColorEight}{\wt(z_6, #1)}}
\newcommand{\ReadZSeven}[1]{\textcolor{ColorNine}{\rd(z_7, #1)}}
\newcommand{\WriteZSeven}[1]{\textcolor{ColorNine}{\wt(z_7, #1)}}
\newcommand{\ReadZEight}[1]{\textcolor{ColorTen}{\rd(z_8, #1)}}
\newcommand{\WriteZEight}[1]{\textcolor{ColorTen}{\wt(z_8, #1)}}

\newcommand{\ReadA}[2]{\textcolor{ColorEleven}{\rd(a_{#2}, #1)}}
\newcommand{\WriteA}[2]{\textcolor{ColorEleven}{\wt(a_{#2}, #1)}}

\newcommand{\ReadB}[1]{\textcolor{ColorTwelve}{\rd(b, #1)}}
\newcommand{\WriteB}[1]{\textcolor{ColorTwelve}{\wt(b, #1)}}

\newcommand{\TReadXOne}[2]{\textcolor{ColorOne}{\rd(#2, x_1, #1)}}
\newcommand{\TWriteXOne}[2]{\textcolor{ColorOne}{\wt(#2, x_1, #1)}}
\newcommand{\TReadXTwo}[2]{\textcolor{ColorTwo}{\rd(#2, x_2, #1)}}
\newcommand{\TWriteXTwo}[2]{\textcolor{ColorTwo}{\wt(#2, x_2, #1)}}

\newcommand{\TReadYOne}[2]{\textcolor{ColorThree}{\rd(#2, y_1, #1)}}
\newcommand{\TWriteYOne}[2]{\textcolor{ColorThree}{\wt(#2, y_1, #1)}}
\newcommand{\TReadYTwo}[2]{\textcolor{ColorFour}{\rd(#2, y_2, #1)}}

\newcommand{\TReadYFour}[2]{\textcolor{ColorSix}{\rd(#2, y_4, #1)}}
\newcommand{\TWriteYFour}[2]{\textcolor{ColorSix}{\wt(#2, y_4, #1)}}
\newcommand{\TReadYFive}[2]{\textcolor{ColorSeven}{\rd(#2, y_5, #1)}}

\newcommand{\TReadYSeven}[2]{\textcolor{ColorNine}{\rd(#2, y_7, #1)}}
\newcommand{\TWriteYSeven}[2]{\textcolor{ColorNine}{\wt(#2, y_7, #1)}}
\newcommand{\TReadYEight}[2]{\textcolor{ColorTen}{\rd(#2, y_8, #1)}}
\newcommand{\TWriteYEight}[2]{\textcolor{ColorTen}{\wt(#2, y_8, #1)}}

\newcommand{\TReadZOne}[2]{\textcolor{ColorThree}{\rd(#2, z_1, #1)}}
\newcommand{\TWriteZOne}[2]{\textcolor{ColorThree}{\wt(#2, z_1, #1)}}
\newcommand{\TReadZTwo}[2]{\textcolor{ColorFour}{\rd(#2, z_2, #1)}}

\newcommand{\TReadZFive}[2]{\textcolor{ColorSeven}{\rd(#2, z_5, #1)}}

\newcommand{\TReadZSix}[2]{\textcolor{ColorEight}{\rd(#2, z_6, #1)}}

\newcommand{\TReadZSeven}[2]{\textcolor{ColorNine}{\rd(#2, z_7, #1)}}

\newcommand{\TReadZEight}[2]{\textcolor{ColorTen}{\rd(#2, z_8, #1)}}

\newcommand{\TReadA}[3]{\textcolor{ColorEleven}{\rd(#3, a_{#2}, #1)}}
\newcommand{\TWriteA}[3]{\textcolor{ColorEleven}{\wt(#3, a_{#2}, #1)}}

\newcommand{\TReadB}[2]{\textcolor{ColorTwelve}{\rd(#2, b, #1)}}
\newcommand{\TWriteB}[2]{\textcolor{ColorTwelve}{\wt(#2, b, #1)}}

\newcommand*\circled[1]{\tikz[baseline=(char.base)]{
            \node[shape=circle,draw=none,fill=orange!20!white,inner sep=0.7pt, solid] (char) {\textcolor{black}{\texttt{#1}}};}}
\newcommand*\circledsmall[1]{\tikz[baseline=(char.base)]{
  \node[shape=circle,draw=none,fill=orange!20!white,inner sep=0.7pt, solid] (char) {\textcolor{black}{\texttt{#1}}};}}

%% file: abstract.tex
\begin{abstract}
Weak-memory models are standard formal specifications of concurrency across hardware, programming languages, and distributed systems. 
A fundamental computational problem is \emph{consistency testing}:~is the observed execution of a concurrent program in alignment with the specification of the underlying system?
The problem has been studied extensively across Sequential Consistency (SC) and weak memory, and proven to be $\NP$-complete when some aspect of the input (e.g., number of threads/memory locations) is unbounded.
This unboundedness has left a natural question open:~are there efficient \emph{parameterized} algorithms for testing?

The main contribution of this paper is a deep hardness result for consistency testing under many popular weak-memory models:~
the problem remains $\NP$-complete even in its \emph{bounded} setting, where candidate executions contain a bounded number of threads, memory locations, and values.
This hardness spreads across several Release-Acquire variants of C11, a popular variant of its Relaxed fragment, popular Causal Consistency models, and the POWER architecture.
To our knowledge, this is the first result that fully exposes the hardness of weak-memory testing and proves that the problem \emph{admits no parameterization} under standard input parameters.
It also yields a computational separation of these models from SC, x86-TSO, PSO, and Relaxed, for which bounded consistency testing is either known (for SC), or shown here (for the rest), to be in polynomial time.
\end{abstract}

%% file: intro.tex

\section{Introduction}\label{SEC:INTRO}



Memory-consistency models play a crucial role in the design, use, and verification of concurrent systems spanning hardware, programming languages, and distributed computing. 
These models formally define the set of behaviors that the system can exhibit as a whole, accounting for the intricate communication patterns between its entities due to buffers, caching, message delays, etc.
The simplest and most widespread, general model is Sequential Consistency ($\scmm$)~\cite{Lamport1978}, which defines program behavior by thread interleaving.
Although its simplicity is a big advantage, $\scmm$ fails to capture the additional, complex behaviors that are abundant in modern concurrency.

In contrast, \emph{weak-memory models}  are richer and more faithful specifications of concurrent/distributed communication and are developed specifically for the system under consideration.
For example, the x86, POWER, and ARM architectures follow their own memory models~\cite{Owens:2010, Alglave:2014,Alglave:2021}, programming languages provide certain primitives for writing weak-memory concurrent programs~\cite{Batty:2013}, and distributed systems implement various models of causal consistency~\cite{Hutto1990,Burckhardt2014,Bouajjani2017}. 
Naturally, verification techniques are specific to the memory model at hand, so as to account for (and verify) all the possible behaviors that the system can exhibit according to the model.

One of the core computational problems associated with a memory model is that of \emph{consistency testing:~is a high-level, observed behavior of a program in alignment with the semantics of the underlying model}~\cite{Gibbons:1997}?
The observed behavior is specified in terms of an abstract execution that defines the sequence of instructions each process/thread executed, the shared memory locations it read from/wrote to, along with the respective values read/written.
Answering this question requires determining the low-level, unobserved behavior of the architecture that gave rise to the observed behavior of the program;
for example, the order in which writes were made visible to (one or more of) the threads, and the dataflow between writes and reads.

Consistency testing is a natural task in both the development and the implementation of memory models.
In particular, memory models are contracts between the designers of a system and its users~\cite{Adve:1990}.
When designing hardware architectures, memory subsystems, compiler optimizations, and distributed-communication protocols, 
consistency-testing serves to validate that the contract has been respected~\cite{Gibbons:1997,Qadeer2003,Manovit2006,Chen2009,Windsor:STVR22}.
From the opposite direction, litmus testing is a standard approach to understanding the semantics of hardware architectures~\cite{Alglave2011,Alglave:2014} so as to design faithful models around them.
Here, given a candidate memory model and the observed execution of a litmus test, consistency checking  verifies whether the execution is a counterexample to the model.
Finally, consistency testing is also used as a separability criterion between different memory models~\cite{Wickerson2017,Kokologiannakis2023}.

Consistency checks are also a widespread task in program verification and testing.
In stateless model checking, the goal of the model checker is to enumerate-and-check the absence of errors in all program executions (typically up to some bound).
To reduce the load of the verification task, an abstraction mechanism partitions the space of all behaviors into equivalence classes, each represented by an abstract execution.
Instead of enumerating concrete executions, the model checker enumerates abstract executions, which yields an exponential reduction of the search space.
Each candidate abstract execution undergoes a consistency check to ensure that the model checker does not diverge to unrealizable parts of the search space~\cite{Chalupa2017,Abdulla:2018,Chatterjee:2019,Kokologiannakis:2019,Kokologiannakis:2022,Abdulla:2019b,Agarwal2021,Bui:2021,Abdulla2023}.
In runtime testing, predictive techniques aim to infer the presence of unobserved, erroneous executions from observed executions that are bug-free. 
Such techniques operate by constructing a candidate execution that manifests the bug (using the observed execution as a guide), and then apply a consistency check (explicitly or implicitly) to verify that the  execution indeed represents valid program behavior (hence the bug report is a true positive)~\cite{Kalhauge2018,Huang14,Pavlogiannis2019,Mathur:lics2020,Luo:2021,Kini2017,Mathur2021}.

Owing to the widespread applicability, the computational complexity of
consistency testing has been studied thoroughly for a wide variety of memory models 
and a wide variety of settings.
The seminal work of \citet{Gibbons:1997} showed that the problem is $\NP$-complete for $\scmm$, even when either the number of threads or the number of memory locations is bounded (but not both).
Later, \citet{Cantin2005} proved that the problem remains $\NP$-complete even with a single memory location (but with unboundedly many threads), which also implies $\NP$-completness for all memory models that adhere to the ``$\scmm$-per-location'' property, such as $\tsomm$, $\psomm$, $\ramm$, $\sramm$ and $\rlxmm$.
\citet{Gonthmakher2003} showed a similar $\NP$-completeness for a Java memory model.
\citet{Furbach2015} proposed a unified treatment of the consistency problem on many weak-memory models which lead to similar $\NP$-completeness results,
while the $\NP$-completeness of consistency testing under various causal-consistency models was proven in~\cite{Bouajjani2017}.

A popular approach to tackle the intractability in consistency testing is via \emph{parameterization}:~an intractable problem becomes tractable when some of its input parameters, such as the number of threads, is bounded~\cite{Gibbons1994,Mathur:lics2020,Abdulla:2018}.
On the other hand, all existing results that establish the $\NP$-hardness of consistency testing in all memory models rely on some input parameters being unbounded.
For $\scmm$, this unboundedness is, in fact, a prerequisite for intractability:~the problem becomes polynomial-time when both the number of threads and memory locations are bounded, a result that has led to efficient parameterized model checking~\cite{Agarwal2021}.
For weak-memory, however, analogous results have thus far remained elusive.
In particular, are there any efficient parameterized algorithms for consistency in weak memories?
\emph{How hard, actually, is weak-memory testing?}

Here we resolve this question by establishing a deep hardness result for many popular weak-memory models:~consistency testing is $\NP$-complete even in its \emph{bounded} setting,
where executions contain a constant number of threads, memory locations and values, and the size of the input is solely determined by the (unbounded) number of events.
To our knowledge, this is the first result that fully exposes the hardness of weak-memory testing and proves that the problem admits no parameterization under standard input parameters.
In turn, this implies that practical approaches to testing have to resort to heuristics, while model checkers might be more performant when exploring finer abstractions (such as reads-from~\cite{Chalupa2017,Abdulla:2019b,Tunc2023}, or those based on executions graphs~\cite{Kokologiannakis:2018,Lahav:2019}).

\input{contributions}

%% file: contributions.tex

\Paragraph{Our contributions.}
We study the \emph{bounded consistency testing problem} for many popular weak-memory models found in software, hardware, and distributed systems.
The input is always an abstract execution $\expartial=(\E, \po)$ consisting of a set of events $\E$ and a program order $\po$ defining the order of execution of these events in each thread.
The task is to determine whether $\expartial$ is consistent in a given memory model.
The \emph{boundedness} of the problem refers to the number of threads, memory locations, and values accessed by $\expartial$ being bounded (i.e., constant).
It is easy to see that the problem is in $\NP$ in all the models we consider, and we will not be establishing this fact formally.
We write $\MemModel_1\mmorder \MemModel_2$ to denote that memory model $\MemModel_2$ is weaker than memory model $\MemModel_1$, i.e., any execution that is consistent in $\MemModel_1$ is also consistent in $\MemModel_2$.

We begin with release-acquire semantics, as popularized by C11.
We consider the Release-Acquire model ($\ramm$), as well as its Strong ($\sramm$) and Weak variants ($\wramm$)~\cite{Lahav:2022}.
In addition, we consider $\arlxmm$, the standard Relaxed semantics of C11 equipped with the common assumption of causal acyclicity (aka $(\po\cup\rf)$ acyclicity).
This assumption is often used as an additional axiom~\cite{Margalit2021,NorrisDemsky:2013} as it has been argued that $(\po\cup\rf)$-cycles do not arise in practice~\cite{Lee:PLDI:2023}.
We prove the following theorem.

\begin{restatable}{theorem}{thmlowerra}\label{thm:lower_ra}
Consistency testing for bounded inputs is $\NP$-complete for $\arlxmm$ as well as 
for any memory model $\MemModel$ such that $\sramm \mmorder \MemModel \mmorder \wramm$, even in their atomic read-modify-write (RMW)-free fragment.
\end{restatable}
Note that \cref{thm:lower_ra} establishes hardness for $\arlxmm$, $\wramm$, $\ramm$, $\sramm$ as well as the whole range of models between $\sramm$ and $\wramm$.
This result improves existing results on consistency checking, for which hardness relied on an unbounded domain of threads and/or memory locations~\cite{Gibbons1994,Cantin2005,Furbach2015}.

Next, we turn our attention to popular causal-consistency models~\cite{Lamport1978,Fidge1988}.
There have been several efforts to formalize various aspects of causal consistency,
out of which have emerged three well-accepted models, namely
Causal Consistency $\ccmm$~\cite{Hutto1990,Bouajjani2017},
Causal Convergence $\cvmm$~\cite{Burckhardt2014,Perrin2016,Bouajjani2017}, and
Causal Memory $\cmmm$~\cite{Ahamad1995,Perrin2016,Bouajjani2017}.
It was recently shown that $\ccmm$ coincides with $\wramm$ while $\cvmm$ coincides with $\sramm$~\cite{Lahav:2022}.
Thus \cref{thm:lower_ra} extends to $\ccmm$ and $\cvmm$.
We prove that the problem is also hard for $\cmmm$, thereby establishing hardness for the ranges defined by the three main models.

\input{figures/models_lattice}

\begin{restatable}{theorem}{thmlowercc}\label{thm:lower_cc}
Consistency testing for bounded inputs is $\NP$-complete for any memory model $\MemModel$ such that
(i)~$\cvmm \mmorder \MemModel \mmorder \ccmm$ or
(ii)~$\cmmm \mmorder \MemModel \mmorder \ccmm$.
\end{restatable}

Next, we turn our attention to the $\powermm$ architecture.
\citet{Lahav:2016} show that $\sramm$ captures precisely the guarantees of $\powermm$ for programs that are compiled from the release-acquire fragment of C/C++.
Thus \cref{thm:lower_ra} extends to the following corollary.

\begin{restatable}{corollary}{corpower}\label{cor:power}
Consistency testing for bounded inputs is $\NP$-complete for $\powermm$.
\end{restatable}

Continuing with hardware models, we study Total Store Order $(\tsomm)$ as employed in x86 architectures (aka x86-TSO) and its extension to Partial Store Order $(\psomm)$.
It turns out that, in the bounded setting, consistency checks become tractable in these models.

\begin{restatable}{theorem}{thmuppertsopso}\label{thm:upper_tso_pso}
Consistency testing for bounded threads and memory locations is in polynomial time for $\tsomm$ and $\psomm$.
\end{restatable}

One natural, final question concerns the vanilla $\rlxmm$ model, i.e., if we remove the acyclicity condition from $\arlxmm$.
In this case, the problem becomes polynomial-time, which is a corollary of the corresponding result for $\scmm$~\cite{Agarwal2021}.

\begin{restatable}{corollary}{corrlx}\label{cor:rlx}
Consistency testing for bounded threads is in polynomial time for $\rlxmm$.
\end{restatable}

Although \cref{cor:rlx} is technically straightforward, it is conceptually interesting under the following realization.
For all of our previous results (as well as for $\scmm$), the hardness of consistency coincides with whether the corresponding model exhibits multi-copy atomicity.
In contrast, $\rlxmm$ is non-multi-copy atomic, yet consistency testing is in polynomial time.

Following the results of this paper, \cref{fig:models_lattice} pictorially presents the full landscape of the tractability and the hardness in testing weak memories.

\Paragraph{High-level intuition.}
Our proofs exploit complex combinatorial properties that arise in weak memory.
Although it is hard to pinpoint one key insight that fully explains our hardness results, 
our proofs rely on the fact that most of the models we consider
(i)~are causally consistent, and
(ii)~allow $(\po'\cup\rf\cup\fr)$-cycles,
where $\po'$ is the standard program order restricted to instructions of the same type (read-read and write-write orderings) on different locations.
In contrast, the polynomial-time models $\scmm$ and $\tsomm$ forbid (ii), while $\psomm$ allows (ii) but also fails (i).

\Paragraph{Outline.}
The rest of the paper is organized as follows.
\begin{compactitem}
\item In \cref{SEC:PRELIMINARIES}, we define our problem setting and the memory models we consider based on C/C++ atomics.
We also develop relevant notation that will be helpful in later sections.
\item In \cref{SEC:LOWER_THREADS_LOCATIONS_RLX}, we prove \cref{thm:lower_ra} for $\arlxmm$.
For readability, we prove a weaker version of \cref{thm:lower_ra} in which the inputs use boundedly many threads and locations but manipulate unboundedly many values.
Later in \cref{SEC:BOUNDED_VALUES}, we explain how to perform simple modifications to our reduction to make it work even for bounded values.
\item In \cref{SEC:LOWER_THREADS_LOCATIONS_WRA_RA_SRA}, we prove \cref{thm:lower_ra} for all models $\sramm\mmorder \MemModel\mmorder \wramm$.
Similarly to the previous case, our reduction uses unboundedly many values, while the modifications described in \cref{SEC:BOUNDED_VALUES} also apply to this model, to arrive at the final result.
\item In \cref{SEC:OTHER_MODELS}, we establish \cref{thm:lower_cc}, \cref{cor:power}, \cref{thm:upper_tso_pso} and \cref{cor:rlx}.
\item Finally, in \cref{SEC:BOUNDED_VALUES}, we present the modifications in the reductions of \cref{SEC:LOWER_THREADS_LOCATIONS_RLX} and \cref{SEC:LOWER_THREADS_LOCATIONS_WRA_RA_SRA} that fully establish \cref{thm:lower_ra}.
\end{compactitem}

Some proofs are relegated to the Appendix.

%% file: figures/models_lattice.tex

\begin{figure}
\centering
\begin{tikzpicture}
\def\xstep{1.4}
\def\ystep{0.8}
\node (sc) at (0*\xstep,0*\ystep)  {$\scmm$};
\node (tso) at (0*\xstep,-1*\ystep)  {$\tsomm$};
\node (cm) at (-2*\xstep,-2*\ystep)  {$\cmmm$};
\node (pso) at (2*\xstep,-2*\ystep)  {$\psomm$};
\node (sra) at (0*\xstep,-2*\ystep) {$\sramm=\cvmm$};
\node (ra) at (0*\xstep,-3*\ystep) {$\ramm$};
\node (wra) at (-2*\xstep,-4*\ystep) {$\wramm=\ccmm$};
\node (arlx) at (0*\xstep, -4*\ystep) {$\arlxmm$}; 
\node (rlx) at (0*\xstep, -5*\ystep) {$\rlxmm$}; 
\node (power) at (2.5*\xstep, -4.75*\ystep) {$\powermm$};


    
%
  \draw[po] (sc) to (tso); 
  \draw[po] (tso) to (cm);
  \draw[po] (tso) to (sra);
  \draw[po] (tso) to (pso);
  \draw[po, line width=0.8mm] (sra) to (ra);
  \draw[po, line width=0.8mm] (ra) to (wra);
  \draw[po, line width=0.8mm] (cm) to (wra);
  \draw[po] (ra) to (arlx);
  \draw[po] (arlx) to (rlx);
  \draw[po] (pso) to (arlx);
  
%
%

\node[] at (1.8*\xstep, -0.8*\ystep) {\large $\PComplexity$};
\node[] at (-1.5*\xstep, -5*\ystep) {\large $\PComplexity$};
\node[] at (2.5*\xstep, -3.5*\ystep) {\large $\NPComplete$};

\draw[draw=gray, very thick, dashed, rounded corners] (-1*\xstep,0.2*\ystep) to[out=-50, in=160] (-0.3*\xstep,-1.3*\ystep) to (0.5*\xstep, -1.5*\ystep) to[out=-50, in=180] (3.5*\xstep, -2.5*\ystep);

\draw[draw=gray, very thick, dashed, rounded corners] (-3*\xstep,-4.6*\ystep) to[out=5, in=170] (0.3*\xstep,-4.7*\ystep)  to[out=-10, in=180] (3.5*\xstep, -5.3*\ystep) ;
\end{tikzpicture}   
\caption{\label{fig:models_lattice}
The complexity landscape of bounded weak-memory testing.
An arrow $\MemModel_1\to \MemModel_2$ means that $\MemModel_1$ is stronger than $\MemModel_2$.
Thick arrows represent range-hardness for all models between the endpoints.
The complexity of bounded consistency checking for all models except for $\scmm$ are established in this paper.
}
\end{figure}
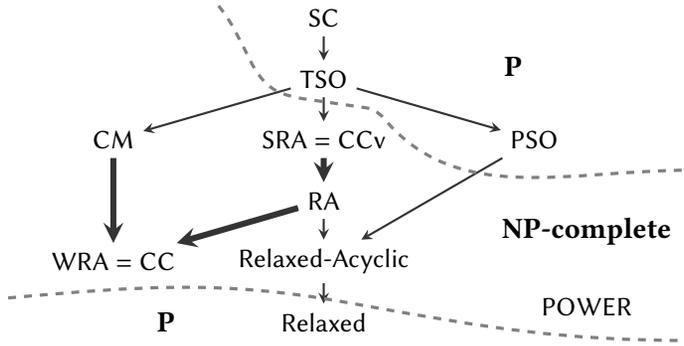

%% file: preliminaries.tex

\section{Preliminaries}\label{SEC:PRELIMINARIES}
This section defines the axiomatic semantics of the $\sramm$, $\ramm$, $\wramm$, and $\rlxmm$ memory models.
As these are standard concepts, our exposition follows recent work on the topic (e.g.,~\cite{Margalit2021,Lahav:2022,Tunc2023}).
In axiomatic semantics, program executions consist of sets of events and relations between them. 
Given an integer $i$, we let $[i]=\{1,2,\dots,i\}$.

\Paragraph{Events.} 
An event is a tuple $\tup{\id,\tid,\llab}$ where $\id$, $\tid$, $\llab$ denote a unique identifier, thread identifier, and the label respectively. 
The label is of the form $\llab=\tup{\op, \lloc, \Val, \ord}$ where $\op$, $\lloc$, $\Val$, $\ord$ respectively denote a read ($\rd$) or write ($\wt$) memory operation, accessed memory location, read or written value, and memory order respectively. 
For the $\sramm$, $\ramm$, and $\wramm$ models, reads and writes are of \emph{acquire} and \emph{release} orders respectively.
For the $\rlxmm$ model, the read and write accesses have \emph{relaxed} order. 
These memory orders are used to define the semantics of models like~C11, but we will not be using them explicitly here.
As we treat each model separately, all access orders are determined by the models and are never mixed.
Hence, we will simply write $\rd(t,x,v)$/$\wt(t,x,v)$ to denote a read/write event of thread $t$, accessing location $x$ and reading/writing value $v$.
We occasionally omit $x$ and/or $v$, when it is irrelevant or clear from the context, while we let $\tid(\event)$ denote the thread of event $\event$.
We do not introduce fences or atomic read-modify-write (RMW) events, as all our hardness results hold even with only read/write events, while our positive results can be easily extended to handle fences and RMWs.
Finally, we denote the set of read and write accesses by $\R$ and $\W$ respectively.

\Paragraph{Notation on relations.}
Let $B$ be a binary relation over a set of events $\E$.
The reflexive, transitive, reflexive-transitive closures, and inverse relations of $B$
are denoted as $B^?$, $B^+$, $B^*$, and $B^{-1}$, respectively. 
We compose two relations $B_1$ and $B_2$ as $B_1;B_2$. 
$[A]$ denotes the identity relation on a set $A$.
We write $\irr(B)$ and $\acy(B)$ to denote that $B$ is irreflexive and acyclic, respectively. 
We occasionally write that there exists a $B$-edge $\event\LTo{B}\event'$ to denote that $(\event, \event')\in B$.
We naturally extend this notation to paths, so that a $B$-path $P\colon \event\LPath{B}\event'$ is a sequence of $B$-edges
$
\event=\event_1\LTo{B}\event_2\LTo{B}\cdots \LTo{B}\event_i=\event'
$.
Finally, we write $B_x$ to restrict $B$ on events accessing location $x$.

\Paragraph{Executions and relations.}
An execution is a tuple $\ex = \tup{\E, \po, \rf, \mo}$  where $\E$ is a set of events and $\po$, $\rf$, $\mo$ are binary relations over $\E$.
In particular, the \emph{program order} ($\po \subseteq (\E \times \E)$) is a strict total order on the events of each thread.
The \emph{reads-from} relation ($\rf \subseteq (\W \times \R)$) relates a write and read event pair $(\wt, \rd)$, 
denoting that $\rd$ obtains its value from $\wt$.
Every read reads from exactly one write on the same memory location and having the same value (thus $\rf^{-1}$ is a function).   
The \emph{modification order} ($\mo \subseteq \bigcup_x (\W_x \times \W_x)$) is a strict total order over
same-location writes in an execution.
Finally, the \emph{happens-before} relation is defined as $\hb\triangleq (\po\cup\rf)^+$.
\cref{fig:axex} shows examples of executions presented as \emph{execution graphs}. 
In each execution graph the nodes represent events and the edges represent relations. 
We omit some relation-edges that are clear from the context
\footnote{It is also common to define a \emph{from reads} relation $\fr\triangleq \rf^{-1};\mo$.
However, we will not be using $\fr$ explicitly in this paper.}.

\Paragraph{Consistency Axioms.} 
Consistency axioms capture different aspects or properties of an execution, such as coherence and causality cycles, under a memory model. 
These properties are interpreted differently in different memory models.

\SubParagraph{Coherence.}
In an execution, \emph{coherence} enforces an ordering between same-location events. 
For events using the release-acquire memory orders,
\emph{write-coherence} requires that each $\mo_x$ order agrees with $\hb$. 
A stronger variant is \emph{strong-write-coherence}, which requires that $\mo$ agrees with $\hb$, transitively.
\emph{Read coherence} enforces that a read $\rd$ can read from a write $\wt$ when  
there is no \emph{intermediate} write $\wt'$ on the same-location that happens-before $\rd$.  
Depending upon how ``intermediate'' writes are treated,
two variations of read coherence are 
popular --- in standard \emph{read-coherence}, 
$\wt$ and $\wt'$ are ordered by $\mo_x$ whereas in \emph{weak-read-coherence} they are ordered by $\hb_x$. 
Finally, we also have variants of write and read coherence when all accesses are relaxed. 
Here $\hb$ is replaced with $\po$, as $\rf$-edges do not contribute to $\hb$ between different memory locations.

\begin{table}
\caption{\label{tab:axioms}
Variants of the coherence axioms.
}
\small
\begin{tabular}{lr|lr}
\multicolumn{2}{c|}{\underline{Release-Acquire}} & \multicolumn{2}{c}{\underline{Relaxed}} \\
$\irr(\mo_x;\hb)$ & (\WCoh{}) & $\irr(\mo_x;\po)$ & (\RlxWCoh{})\\
$\acy(\hb \cup \mo)$ & (\SWCoh{}) && \\
$\irr(\rf^{-1};\mo_x;\hb)$ & (\RCoh{}) & $\irr(\rf^{-1};\mo_x;\rf^?;\po)$ & (\RlxRCoh{})\\
$\irr(\hb_x;[\W];\hb_x;\rf^{-1})$ & (\WRCoh{}) &&\\
\end{tabular}
\end{table}

\SubParagraph{Causality Cycles.} 
A causality cycle arises in the presence of relaxed accesses and consists of $\po$ and $\rf$ orderings. 
A causality cycle may result in `out-of-thin-air' behavior in an execution. 
To avoid such `out-of-thin-air' behavior, many consistency models and verification tools explicitly disallow such cycles~\cite{NorrisDemsky:2013,Margalit2021,Luo:2021}.
\begin{compactitem}
 \item $\acy(\po \cup \rf)$ \hfill(\PORF{})
 \end{compactitem}

\input{figures/axiom_examples}

\cref{fig:axex} shows examples of executions forbidden by different axioms. 
The \WCoh{} axiom forbids the execution in \cref{subfig:wc} as it violates the irreflexivity of $(\mo_x;\hb)$. 
The $(\hb \cup \mo)$ cycle in \cref{subfig:swc} is forbidden by \SWCoh{}. 
The execution in \cref{subfig:rc} violates irreflexivity of $(\rf^{-1};\mo_x;\hb)$ and thus fails \RCoh{}.
In \cref{subfig:wrc}, we have $(\wt(x),\wt(x)) \in \hb_x;[\W]$, $(\wt(x),\rd(x)) \in \po \subseteq \hb$, and $(\rd(x),\wt(x)) \in \rf^{-1}$,
violating \WRCoh{}.
The execution in \cref{subfig:porf} violates \PORF{}. 
Finally, the executions in \cref{subfig:rlxwcoh} and \cref{subfig:rlxrcoh} violate \RlxWCoh{} and \RlxRCoh{}, respectively.

\Paragraph{Memory Models.}
We can now describe the main memory models we consider in this work, by listing the axioms that each execution needs to satisfy in the respective model \cref{tab:main_models}.

\SubParagraph{Release-Acquire and variants.}
The release-acquire ($\ramm$) memory model is weaker than sequential consistency and is arguably the most well-understood fragment of C11. 
Here, the reads-from relation $\rf$ induces synchronization between thread threads, which is captured in the semantics by the happens-before relation $\hb$.
Following~\cite{Lahav:2022}, we consider three variants of 
release-acquire models: Release-Acquire ($\ramm$), and its Strong ($\sramm$) and Weak ($\wramm$) variants.

$\sramm$ enforces strong-write-coherence on write accesses whereas $\ramm$ enforces write-coherence. 
On the other hand, $\wramm$ does not place any ordering between same-location writes by $\mo_x$. 
Instead, the only orderings considered between same-location writes are through the $[\W];\hb_x;[\W]$ relation.

\begin{table}
\caption{\label{tab:main_models}
The main weak-memory models based on C11 that we consider in this work.
}
\centering
\setlength\tabcolsep{10pt}
\begin{tabular}{lr|lr}
\multicolumn{2}{c|}{\underline{Release-Acquire}} & \multicolumn{2}{c}{\underline{Relaxed}} \\
$\wramm$ & \makecell{\PORF{}\\ \WRCoh{}} & $\rlxmm$ & \makecell{\RlxWCoh{}\\ \RlxRCoh{}} \\ 
\hline
$\ramm$ & \makecell{\WCoh{}\\ \RCoh{}} & \multirow{ 2}{*}{$\arlxmm$} & \multirow{ 2}{*}{ \makecell{\RlxWCoh{}\\ \RlxRCoh{} \\\PORF{} } } \\ 
\cline{1-2}
$\sramm$ & \makecell{\SWCoh{}\\ \RCoh{}} & & \\ 
\end{tabular}
\end{table}

\SubParagraph{Relaxed.}
All accesses in the $\rlxmm$ model 
satisfy the corresponding coherence axioms \RlxWCoh{} and \RlxRCoh{}, which guarantee $\scmm$-per-location.
The $\arlxmm$ model strengthens $\rlxmm$ by also requiring the acyclicity of $(\po\cup\rf)$.

Based on the set of allowed behaviors, these models can be partially ordered as $\sramm \mmorder \ramm \mmorder \set{\wramm, \set{\arlxmm\mmorder\rlxmm}}$, where models towards the right allow more behaviors.

\Paragraph{The consistency-testing problem.}
An execution $\ex$ is \emph{consistent} in a memory model $\MemModel$, written $\ex\models \MemModel$, if it satisfies the axioms of $\MemModel$. 
For example, the execution in \cref{subfig:swc} satisfies all axioms except \SWCoh{}, and hence it is consistent in $\ramm$, $\wramm$, and $\parlxmm$.

When testing the behavior of a program within a memory model, one does not have access to \emph{concrete} executions, but rather to \emph{abstract} executions.
The latter contains only information observed by the program, i.e., the events it executed and the values it read/wrote.
Formally, an \emph{abstract execution} $\expartial=\tup{\E, \po}$ is a coarser object than concrete executions, missing the $\mo$ and $\rf$ relations, and a concrete execution
$\ex = \tup{\E', \po', \rf', \mo'}$ is said to be an extension of $\expartial$ if $\E' = \E$ and $\po' = \po$.
We call $\expartial$ \emph{consistent} in $\MemModel$, written similarly as $\expartial\models \MemModel$, if there exists an $\rf$ and an $\mo$ such that the extension $\ex = \tup{\E, \po, \rf, \mo}$ is an execution with $\ex\models \MemModel$.
The problem of \emph{consistency testing} in a memory model $\MemModel$ is to determine whether $\expartial$ is consistent in $\MemModel$,
i.e., whether there is a way to resolve $\rf$ and $\mo$ in $\MemModel$ that would give rise to the observed behavior $\expartial$ on the program level.

\Paragraph{Conflicting triplets.}
In the coming sections, we use the notion of \emph{conflicting triplets}.
Given an abstract execution $\expartial=(\E, \po)$, we say that two events $\event_1,\event_2\in \E$ \emph{conflict} if they access the same location and at least one of them is a write.
Given additionally a reads-from relation $\rf$,
a \emph{conflicting triplet} (or \emph{triplet}, for short) is a tuple $(\wt, \rd, \wt')$ of pairwise conflicting events such that $(\wt, \rd)\in \rf$.

%% file: figures/axiom_examples.tex
\begin{figure}
\centering
\def\ystep{0.4}
\begin{subfigure}{0.14\textwidth}
\centering
\scalebox{0.9}{
\begin{tikzpicture}[yscale=1]
  \node (t11) at (0,0*\ystep)  {$\wt(x)$};
  \node (t12) at (0,-4*\ystep) {$\wt(y)$};
  \node (t21) at (1.4,0*\ystep) {$\rd(y)$};
  \node (t22) at (1.4,-4*\ystep) {$\wt(x)$};
%
  \draw[po] (t11) to (t12);
  \draw[po] (t21) to (t22);
%
  \draw[rf,bend left=0] (t12) to node[above,pos=0.9, sloped]{$\rf$} (t21);
%
  \draw[mo,bend left=0] (t22) to node[above,pos=0.8, sloped]{$\mo$} (t11);
\end{tikzpicture} 
}  
\caption{}
\label{subfig:wc}
\end{subfigure}
\hfill
\begin{subfigure}{0.14\textwidth}
\centering
\scalebox{0.9}{
\begin{tikzpicture}[yscale=1]
  \node (t11) at (0,0*\ystep)  {$\wt(y)$};
  \node (t12) at (0,-4*\ystep) {$\wt(x)$};
%
  \node (t21) at (1.4,0*\ystep) {$\wt(x)$};
  \node (t22) at (1.4,-4*\ystep) {$\wt(y)$};
%
  \draw[po] (t11) to (t12);
%
  \draw[po] (t21) to (t22);
%
%
  \draw[mo,bend left=0] (t12) to (t21);
  \draw[mo,bend right=0] (t22) to node[above,pos=0.8, sloped]{$\mo$} (t11);
\end{tikzpicture} 
}  
\caption{}
\label{subfig:swc}
\end{subfigure}
\hfill
\begin{subfigure}{0.14\textwidth}
\centering
\scalebox{0.9}{
\begin{tikzpicture}[yscale=1]
  \node (t11) at (0,0*\ystep)  {$\wt(x)$};
  \node (t12) at (0,-2*\ystep) {$\wt(x)$};
  \node (t13) at (0,-4*\ystep) {$\wt(y)$};
  \node (t21) at (1.4,0*\ystep) {$\rd(y)$};
  \node (t22) at (1.4,-4*\ystep) {$\rd(x)$};
%
  \draw[po] (t12) to (t13);
  \draw[po] (t21) to (t22);
 \draw[mo] (t11) to node[right]{$\mo$} (t12);
  \draw[rf,bend left=0] (t13) to node[below,pos=0.2, sloped]{$\rf$} (t21);
  \draw[rf,bend left=15] (t11.east) to node[below,pos=0.8, sloped]{$\rf$} (t22);
%
\end{tikzpicture}
}
\caption{}
\label{subfig:rc}
\end{subfigure}
\hfill
\begin{subfigure}{0.14\textwidth}
\centering
\scalebox{0.9}{
\begin{tikzpicture}[yscale=1]
  \node (t11) at (0,0*\ystep)  {$\wt(x)$};
  \node (t12) at (0,-4*\ystep) {$\wt(y)$};
  \node (t21) at (1.4,0*\ystep) {$\rd(y)$};
  \node (t22) at (1.4,-2*\ystep) {$\wt(x)$};
  \node (t23) at (1.4,-4*\ystep) {$\rd(x)$};
  \draw[po] (t11) to (t12);
  \draw[po] (t21) to (t22);
  \draw[po] (t22) to (t23);
  \draw[rf,bend left=0] (t12) to node[above,pos=0.8, sloped]{$\rf$} (t21);
  \draw[rf,bend right=0] (t11) to node[below,pos=0.8, sloped]{$\rf$} (t23);
\end{tikzpicture}   
} 
\caption{}
\label{subfig:wrc}
\end{subfigure}
\hfill
\begin{subfigure}{0.14\textwidth}
\centering
\scalebox{0.9}{
\begin{tikzpicture}[yscale=1]
  \node (t11) at (0,0*\ystep)  {$\rd(x)$};
  \node (t12) at (0,-4*\ystep) {$\wt(y)$};
  \node (t21) at (1.4,0,0*\ystep) {$\rd(y)$};
  \node (t22) at (1.4,-4*\ystep) {$\wt(x)$};
  \draw[po] (t11) to (t12);
  \draw[po] (t21) to (t22);
  \draw[rf,bend left=0] (t12) to node[above,pos=0.1, sloped]{$\rf$} (t21);
  \draw[rf,bend right=0] (t22) to node[above,pos=0.1, sloped]{$\rf$} (t11);
\end{tikzpicture}      
}
\caption{}
\label{subfig:porf}
\end{subfigure}
\begin{subfigure}{0.09\textwidth}
\centering
\scalebox{0.9}{
\begin{tikzpicture}[yscale=1]
  \node (t11) at (0,0*\ystep)  {$\wt(x)$};
  \node (t12) at (0,-4*\ystep) {$\wt(x)$};
%
%
  \draw[po] (t11) to (t12);
  \draw[mo,bend left=20] (t12) to node[left]{$\mo$} (t11);
%
\end{tikzpicture}      
}
\caption{}
\label{subfig:rlxwcoh}
\end{subfigure}
\begin{subfigure}{0.14\textwidth}
\centering
\scalebox{0.9}{
\begin{tikzpicture}[yscale=1]
  \node (t11) at (0,0*\ystep)  {$\wt(x)$};
  \node (t12) at (0,-4*\ystep) {$\wt(x)$};
  \node (t21) at (1.4,0,0*\ystep) {$\rd(x)$};
  \node (t22) at (1.4,-4*\ystep) {$\rd(x)$};
  \draw[mo, bend right=15] (t11) to node[right]{$\mo$} (t12);
  \draw[po] (t21) to (t22);
  \draw[rf,bend left=0] (t12) to node[below,pos=0.1, sloped]{$\rf$} (t21);
  \draw[rf,bend right=0] (t11) to node[above,pos=0.1, sloped]{$\rf$} (t22);
\end{tikzpicture}      
}
\caption{}
\label{subfig:rlxrcoh}
\end{subfigure}
\caption{
Executions forbidden by (\subref{subfig:wc}) \WCoh{}, (\subref{subfig:swc}) \SWCoh, (\subref{subfig:rc}) \RCoh{}, (\subref{subfig:wrc}) \WRCoh{}, (\subref{subfig:porf}) \PORF{}, (\subref{subfig:rlxwcoh}) \RlxWCoh{}, (\subref{subfig:rlxrcoh}) \RlxRCoh{}. 
}
\label{fig:axex}
\end{figure}

%% file: lower_threads_locations_rlx.tex

\section{Hardness for Relaxed-Acyclic}\label{SEC:LOWER_THREADS_LOCATIONS_RLX}
We start with the $\arlxmm$ memory model and show that consistency testing is  $\NP$-complete under bounded threads and memory locations.
This differs slightly from \cref{thm:lower_ra}, which states that the problem remains hard even with bounded values.
Since our proof is rather technical, we choose to present this intermediate result here.
We will make the final step towards \cref{thm:lower_ra} in \cref{SEC:BOUNDED_VALUES}, which consists of a simple modification of the technique presented here.
In \cref{SEC:LOWER_THREADS_LOCATIONS_RLX_REDUCTION}, we present the hardness reduction
and argue about its correctness in \cref{SEC:LOWER_THREADS_LOCATIONS_RLX_SOUNDNESS,SEC:LOWER_THREADS_LOCATIONS_RLX_COMPLETENESS}.

\input{lower_threads_locations_rlx_reduction}
\input{lower_threads_locations_rlx_soundness}
\input{lower_threads_locations_rlx_completeness}

%% file: lower_threads_locations_rlx_reduction.tex

\subsection{Reduction}
\label{SEC:LOWER_THREADS_LOCATIONS_RLX_REDUCTION}

Our reduction is from Monotone 1-in-3 SAT which is known to be $\NP$-complete~\cite{Garey:1990}. 
The input is a monotone formula $\varphi$ in conjunctive normal form, where each clause contains three literals, all of which are positive.
The task is to determine if there exists a 1-in-3 truth assignment for $\varphi$, i.e., one that 
sets exactly one literal to \emph{true} in each clause.

We remark that our reduction is combinatorially elaborate.
We found that complex interactions between threads are necessary to
expose the the nuances that make the consistency problem for the $\arlxmm$ memory model (or for that matter, other memory models we consider) hard.
Nevertheless, we assist the text with illustrations that help visualize and generalize the interaction patterns that are exploited in our reduction.
To further enhance readability, we distinguish different memory locations with different colors (in both figures and main text).

Let $\varphi=\set{\Clause_i}_{i\in[m]}$ be a monotone Boolean formula over $n$ variables $\{s_j\}_{j\in[n]}$ and $m$ clauses of the form $\Clause_i=(s_j, s_k, s_{\ell})$.
We construct an abstract execution $\expartial=(\E, \po)$ such that $\expartial\models \arlxmm$ 
iff
 $\varphi$ is satisfiable using a 1-in-3 assignment\footnote{We often use the phrase `$\varphi$ is satisfiable' to mean `$\varphi$ is satisfiable by a 1-in-3 assignment'.}.

\input{figures/reduction_scheme}

\Paragraph{High-level description.}
Our reduction constructs an abstract execution with $O(n\cdot m)$ events
accessing $\NumLocations=14$  memory locations
in $\NumThreads=23$ threads,
of which the three threads $t_1$, $t_2$ and $t_3$ form the core of the construction.
Events appear in each of these threads in $m$ \emph{phases} (one phase per clause $\Clause_i$), starting from phase $1$ and going to larger phases as we go downwards in the threads.
Each phase, in turn, consists of $n$ \emph{steps} (one step per variable $s_j$), again starting from step $1$ and going to larger steps as we go downwards.
In phase $i$ and step $j$ we have a read event $\TReadXOne{v^i_j}{t_3}$ that can read from either of two writes $\TWriteXOne{v^i_j}{t_1}$ or $\TWriteXOne{v^i_j}{t_2}$.
The former case corresponds to the assignment $s_j=\bfalse$, while the latter case corresponds to the assignment $s_j=\btrue$.
See \cref{fig:reduction_scheme} for an illustration.
Our construction guarantees that the choices for the writer of $\TReadXOne{v^i_j}{t_3}$ are consistent across all phases $i$:~either each $\TReadXOne{v^i_j}{t_3}$ reads from $\TWriteXOne{v^i_j}{t_1}$, which corresponds to setting $s_j=\bfalse$ in $\varphi$, or each $\TReadXOne{v^i_j}{t_3}$ reads from $\TWriteXOne{v^i_j}{t_2}$, which corresponds to setting $s_j=\btrue$ in $\varphi$. 
Moreover, for each clause $\Clause_i=(s_j, s_k, s_{\ell})$, our reduction guarantees that exactly one of $\TReadXOne{v^i_j}{t_3}$, $\TReadXOne{v^i_k}{t_3}$, and $\TReadXOne{v^i_{\ell}}{t_3}$  reads from thread $t_2$, which implies that the corresponding assignment on $s_j$, $s_k$ and $s_{\ell}$ satisfies the 1-in-3 property.
To achieve all these constraints, we introduce four \emph{gadgets}, which consist of events on additional threads and memory locations, that guarantee the desired properties.
In the following, we first describe each gadget separately, and then explain how to interleave them in order to obtain the abstract execution $\expartial$.

\input{figures/copy_gadget_rlx}

\Paragraph{The copy gadget $\CopyGadget{i}{j}$.}
The main gadget in our construction is the copy gadget $\CopyGadget{i}{j}$, 
defined for each $i\in[m]$ and $j\in [n]$, and shown in \cref{fig:copy_gadget_rlx}.
This gadget contains 
(i)~the three focal events  $\TWriteXOne{v^i_j}{t_1}$, $\TWriteXOne{v^i_j}{t_2}$ and $\TReadXOne{v^i_j}{t_3}$  that determine the truth value of $s_j$,
(ii)~three ``mirror'' events $\TWriteXTwo{v^i_j}{t_4}$, $\TWriteXTwo{v^i_j}{t_5}$ and $\TReadXTwo{v^i_j}{t_6}$, and
(iii)~other events on memory locations
$\VarYOne$, $\VarYTwo$, $\VarYThree$, $\VarYFour$.

The gadget couples the writers of $\TReadXOne{v^i_j}{t_3}$ and $\TReadXTwo{v^i_j}{t_6}$:~if $\TReadXOne{v^i_j}{t_3}$ reads from thread $t_1$ then $\TReadXTwo{v^i_j}{t_6}$ reads from thread $t_4$ (see \cref{subfig:copy_gadget_rlxFalse}), while
if $\TReadXOne{v^i_j}{t_3}$ reads from thread $t_2$ then $\TReadXTwo{v^i_j}{t_6}$ reads from thread $t_5$ (see \cref{subfig:copy_gadget_rlxTrue}).

\input{figures/copy_gadget_down_rlx}

\Paragraph{The copy-down gadget $\CopyGadgetDown{i}{j}$.}
We use a copy-down gadget $\CopyGadgetDown{i}{j}$, defined for $i\in[m-1]$ and $j\in[n]$, with structure identical to $\CopyGadget{i}{j}$, and shown in \cref{fig:copy_gadget_down_rlx}. This gadget contains
(i)~the three focal events $\TWriteXOne{v^{i+1}_j}{t_1}$, $\TWriteXOne{v^{i+1}_j}{t_2}$ and $\TReadXOne{v^{i+1}_j}{t_3}$,
(ii)~the three mirror events $\TWriteXTwo{v^i_j}{t_4}$, $\TWriteXTwo{v^i_j}{t_5}$ and $\TReadXTwo{v^i_j}{t_6}$, and
(iii)~other events on memory locations $\VarZOne$, $\VarZTwo$, $\VarZThree$, $\VarZFour$.

While $\CopyGadget{i}{j}$ couples the writers of the focal and mirror read events --- $\TReadXOne{v^i_j}{t_3}$ and $\TReadXTwo{v^i_j}{t_6}$ --- belonging to the same phase, 
$\CopyGadgetDown{i}{j}$ couples the writers of the focal and mirror read events --- $\TReadXOne{v^{i+1}_j}{t_3}$ and $\TReadXTwo{v^i_j}{t_6}$ --- of consecutive phases.
If $\TReadXOne{v^{i+1}_j}{t_3}$ reads from thread $t_1$, then $\TReadXTwo{v^i_j}{t_6}$ reads from thread $t_4$, while
if $\TReadXOne{v^{i+1}_j}{t_3}$ reads from thread $t_2$, then $\TReadXTwo{v^i_j}{t_6}$ reads from thread $t_5$.
Together, $\CopyGadget{i}{j}$ and $\CopyGadgetDown{i}{j}$ couple the writers of the focal reads $\TReadXOne{v^i_j}{t_3}$ and $\TReadXOne{v^{i+1}_j}{t_3}$ across consecutive phases --- either both read from thread $t_1$ or both read from $t_2$.

\input{figures/at_most_gadget_rlx}

\Paragraph{The at-most-one-true gadgets $\AtMostOneGadget{c}{i}{j}{k}$.}
Consider a clause $\Clause_i$, and for each $c\in[3]$, let $(s_j, s_k)$ be the $c$-th pair of literals that appear in $\Clause_i$ (according to some arbitrary but fixed total ordering on pairs of propositional variables).
The at-most-one-true gadget $\AtMostOneGadget{c}{i}{j}{k}$ is shown in \cref{subfig:at_most_one_gadget_rlx} and contains
(i)~the six focal events $\TWriteXOne{v^i_j}{t_1}$, $\TWriteXOne{v^i_j}{t_2}$ and $\TReadXOne{v^i_j}{t_3}$; and $\TWriteXOne{v^i_k}{t_1}$, $\TWriteXOne{v^i_k}{t_2}$ and $\TReadXOne{v^i_k}{t_3}$, and
(ii)~other events on memory location $\VarA{c}$.
The gadget guarantees that at most one of $\TReadXOne{v^i_j}{t_3}$ and $\TReadXOne{v^i_k}{t_3}$ reads from thread $t_2$, which corresponds to assigning $\btrue$ to at most one of $s_j$ and $s_k$ (\cref{subfig:at_most_one_gadget_rlx_true_false,subfig:at_most_one_gadget_rlx_false_true,subfig:at_most_one_gadget_rlx_false_false}).

\Paragraph{The at-least-one-true gadget $\AtLeastOneGadget{i}{j}{k}{\ell}$.}
Consider a clause $\Clause_i=(s_j, s_k, s_{\ell})$.
The at-least-one-true gadget $\AtLeastOneGadget{i}{j}{k}{\ell}$ is shown in \cref{subfig:at_least_one_gadget_rlx} and contains 
(i)~the nine focal events $\TWriteXOne{v^i_j}{t_1}$, $\TWriteXOne{v^i_j}{t_2}$ and $\TReadXOne{v^i_j}{t_3}$; $\TWriteXOne{v^i_k}{t_1}$, $\TWriteXOne{v^i_k}{t_2}$ and $\TReadXOne{v^i_k}{t_3}$; and $\TWriteXOne{v^i_{\ell}}{t_1}$, $\TWriteXOne{v^i_{\ell}}{t_2}$ and $\TReadXOne{v^i_{\ell}}{t_3}$, and
(ii)~other events on memory location $\VarB$.
The gadget guarantees that at least one of 
$\TReadXOne{v^i_j}{t_3}$, $\TReadXOne{v^i_k}{t_3}$ and $\TReadXOne{v^i_{\ell}}{t_3}$
reads from thread $t_2$, which corresponds to assigning $\btrue$ to at least one of $s_j$, $s_k$ and $s_{\ell}$ (shown in \cref{subfig:at_least_one_gadget_rlx_true_false_false,subfig:at_least_one_gadget_rlx_false_true_false,subfig:at_least_one_gadget_rlx_false_false_true}).
Note that this gadget, by itself, allows for two or even all three of the literals to be assigned to $\btrue$, however, these cases are not shown as they are prohibited by the at-most-one-true gadgets above.

\input{figures/at_least_gadget_rlx}

\Paragraph{Putting the gadgets together.}
We serially connect all gadgets in their common threads by $\po$.
In particular, $\CopyGadget{i_1}{j_1}$ appears before $\CopyGadget{i_2}{j_2}$ if $i_1<i_2$ or $i_1=i_2$ and $j_1<j_2$;
$\CopyGadgetDown{i_1}{j_1}$ appears before $\CopyGadgetDown{i_2}{j_2}$ if $i_1<i_2$ or $i_1=i_2$ and $j_1<j_2$;
each $\AtMostOneGadget{c}{i_1}{j_1}{k_1}$ appears before $\AtMostOneGadget{c}{i_2}{j_2}{k_2}$ if $i_1<i_2$, and finally
$\AtLeastOneGadget{i_1}{j_1}{k_1}{\ell_1}$ appears before $\AtLeastOneGadget{i_2}{j_2}{k_2}{\ell_2}$ if $i_1<i_2$.
As various gadgets have common threads and events, besides connecting them, we also need to specify the interleaving between them.
However, this interleaving can be arbitrary and we will not fix it here.
Finally, we have indeed used $\NumThreads=23$ threads and $\NumLocations=14$  memory locations.

%% file: figures/reduction_scheme.tex

\begin{figure}
\newcommand{\xdisposition}{0}
\newcommand{\ydisposition}{0}
\newcommand{\xtstep}{0.3}
\newcommand{\ytstep}{0.7}
\newcommand{\ybias}{-0.3 }
\newcommand{\xstep}{2}
\newcommand{\ystep}{-0.6}
\newcommand{\xtscale}{0.8}
\def\crossoutopacity{0.3}
\def \numevents{9}
\def\scale{0.85}
\centering
\scalebox{\scale}{
\begin{tikzpicture}[thick, font=\footnotesize,
pre/.style={<-,shorten >= 2pt, shorten <=2pt, very thick},
post/.style={->,shorten >= 3pt, shorten <=3pt,   thick},
seqtrace/.style={line width=1},
und/.style={very thick, draw=gray},
virt/.style={circle,draw=black!50,fill=black!20, opacity=0}]

\tikzstyle{event} = [rectangle, fill=white, inner sep=0, minimum height=4.5mm, minimum width=13mm, rounded corners]

\node[] (S11) at (1*\xstep,0.15) {\normalsize $t_1$};
\node[] (S12) at (1*\xstep,\numevents * \ystep) {};
\node[] (S21) at (2*\xstep,0.15) {\normalsize $t_2$};
\node[] (S22) at (2*\xstep,\numevents * \ystep) {};
\node[] (S31) at (3*\xstep,0.15) {\normalsize $t_3$};
\node[] (S32) at (3*\xstep,\numevents * \ystep) {};

\draw[seqtrace] (S11) to (S12);
\draw[seqtrace] (S21) to (S22);
\draw[seqtrace] (S31) to (S32);

\node[event, draw=black] (11) at (1*\xstep, 1*\ystep + 0*\ybias) {$\WriteXOne{v^{1}_{1}}$};
\node[event, draw=black] (12) at (1*\xstep, 2*\ystep + 0*\ybias) {$\WriteXOne{v^{1}_{2}}$};
\node[event, draw=black] (14) at (1*\xstep, 4*\ystep + 0*\ybias) {$\WriteXOne{v^{1}_{n}}$};

\node[event, draw=black] (17) at (1*\xstep, 7*\ystep + 0*\ybias) {$\WriteXOne{v^{m}_{1}}$};
\node[event, draw=black] (19) at (1*\xstep, 9*\ystep + 0*\ybias) {$\WriteXOne{v^{m}_{n}}$};

\node[event, draw=black] (21) at (2*\xstep, 1*\ystep + 0*\ybias) {$\WriteXOne{v^{1}_{1}}$};
\node[event, draw=black] (22) at (2*\xstep, 2*\ystep + 0*\ybias) {$\WriteXOne{v^{1}_{2}}$};
\node[event, draw=black] (24) at (2*\xstep, 4*\ystep + 0*\ybias) {$\WriteXOne{v^{1}_{n}}$};

\node[event, draw=black] (27) at (2*\xstep, 7*\ystep + 0*\ybias) {$\WriteXOne{v^{m}_{1}}$};
\node[event, draw=black] (29) at (2*\xstep, 9*\ystep + 0*\ybias) {$\WriteXOne{v^{m}_{n}}$};

\node[event, draw=black] (31) at (3*\xstep, 1*\ystep + 0*\ybias) {$\ReadXOne{v^{1}_{1}}$};
\node[event, draw=black] (32) at (3*\xstep, 2*\ystep + 0*\ybias) {$\ReadXOne{v^{1}_{2}}$};
\node[event, draw=black] (34) at (3*\xstep, 4*\ystep + 0*\ybias) {$\ReadXOne{v^{1}_{n}}$};

\node[event, draw=black] (37) at (3*\xstep, 7*\ystep + 0*\ybias) {$\ReadXOne{v^{m}_{1}}$};
\node[event, draw=black] (39) at (3*\xstep, 9*\ystep + 0*\ybias) {$\ReadXOne{v^{m}_{n}}$};

\node[] (1) at (-1*\xstep, 4.5*\ystep) {};
\node[] (2) at (5*\xstep, 4.5*\ystep) {};

\node[] (3) at (-1*\xstep, 6.5*\ystep) {};
\node[] (4) at (5*\xstep, 6.5*\ystep) {};

\draw[gray, dashed, thick] (1) to (2);
\draw[gray, dashed, thick] (3) to (4);

\node[] at (-0.5*\xstep,2.5*\ystep) {\large Phase 1};
\node[] at (-0.5*\xstep,5.5*\ystep) {\large \vdots};
\node[] at (1*\xstep + \xtstep,5.5*\ystep) {\Large \vdots};
\node[] at (2*\xstep + \xtstep,5.5*\ystep) {\Large \vdots};
\node[] at (3*\xstep + \xtstep,5.5*\ystep) {\Large \vdots};
\node[] at (-0.5*\xstep,8*\ystep) {\large Phase $m$};

\node[] at (4*\xstep,1*\ystep) {\normalsize Step 1};
\node[] at (4*\xstep,2*\ystep) {\normalsize Step 2};
\node[] at (4*\xstep,3*\ystep) {\normalsize \vdots};
\node[] at (4*\xstep,4*\ystep) {\normalsize Step n};

\node[] at (4*\xstep,5.5*\ystep) {\normalsize \vdots};

\node[] at (4*\xstep,7*\ystep) {\normalsize Step 1};
\node[] at (4*\xstep,8*\ystep) {\normalsize \vdots};
\node[] at (4*\xstep,9*\ystep) {\normalsize Step n};

\end{tikzpicture}
}
\caption{\label{fig:reduction_scheme}
The schematic reduction from a monotone formula $\varphi$ to an abstract execution $\expartial$.
}
\end{figure}
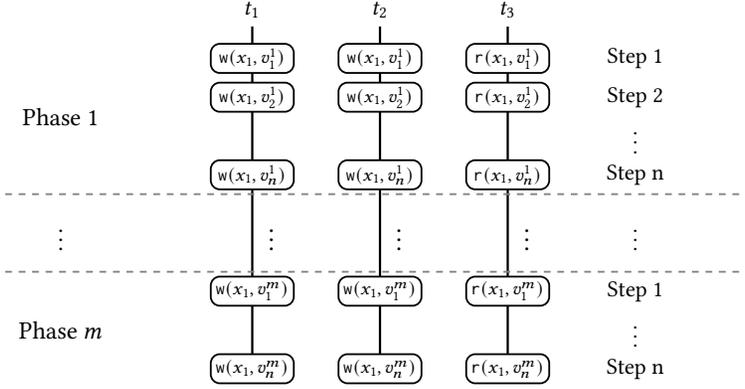

%% file: figures/copy_gadget_rlx.tex

\begin{figure}
\newcommand{\xdisposition}{3.4*\xstep}
\newcommand{\ydisposition}{2*\ystep}
\newcommand{\xtstep}{0.75}
\newcommand{\ytstep}{0.7}
\newcommand{\ybias}{-0.3 }
\newcommand{\xstep}{1.435}
\newcommand{\ystep}{-1}
\newcommand{\xtscale}{0.8}
\def \numevents{3}
\def\scale{0.825}
\newcommand{\BaseCopyGadget}{
\begin{scope}[shift={(0*\xdisposition,0*\ydisposition)}]

\foreach \x [evaluate=\x as \i using ({int(\x)})] in {1,...,3}{
\node[] (T\i1) at (\i*\xstep,0.15) {\normalsize $t_{\x}$};
\node[] (T\i2) at (\i*\xstep,\numevents * \ystep) {};
\draw[seqtrace] (T\i1) to (T\i2);
}

\foreach \x [evaluate=\x as \i using ({int(\x+3)})] in {1,...,6}{
\node[] (T\i1) at (\i*\xstep,0.15) {\normalsize $f_{\x}$};
\node[] (T\i2) at (\i*\xstep,\numevents * \ystep) {};
\draw[seqtrace] (T\i1) to (T\i2);
}

\foreach \x [evaluate=\x as \i using ({int(\x+6)})] in {4,...,6}{
\node[] (T\i1) at (\i*\xstep,0.15) {\normalsize $t_{\x}$};
\node[] (T\i2) at (\i*\xstep,\numevents * \ystep) {};
\draw[seqtrace] (T\i1) to (T\i2);
}

\node[event] (T11) at (1*\xstep, 1*\ystep + 0*\ybias) {$\ReadYOne{v^{i}_{j}}$};
\node[event] (T12) at (1*\xstep, 2*\ystep + 0*\ybias) {$\ReadYOne{u^{i}_{j}}$};
\node[event, draw=black] (T13) at (1*\xstep, 3*\ystep + 0*\ybias) {$\WriteXOne{v^{i}_{j}}$};

\node[event] (T21) at (2*\xstep, 1*\ystep + 0*\ybias) {$\ReadYTwo{v^{i}_{j}}$};
\node[event] (T22) at (2*\xstep, 2*\ystep + 0*\ybias) {$\ReadYTwo{u^{i}_{j}}$};
\node[event, draw=black] (T23) at (2*\xstep, 3*\ystep + 0*\ybias) {$\WriteXOne{v^{i}_{j}}$};

\node[event, draw=black] (T31) at (3*\xstep, 1*\ystep + 0*\ybias) {$\ReadXOne{v^{i}_{j}}$};
\node[event] (T32) at (3*\xstep, 2*\ystep + 0*\ybias) {$\WriteYOne{v^{i}_{j}}$};
\node[event] (T33) at (3*\xstep, 3*\ystep + 0*\ybias) {$\WriteYTwo{v^{i}_{j}}$};

\node[event] (T41) at (4*\xstep, 1*\ystep + 0*\ybias) {$\WriteYOne{\ov{v}^{i}_{j}}$};
\node[event] (T42) at (4*\xstep, 2*\ystep + 0*\ybias) {$\WriteYOne{v^{i}_{j}}$};

\node[event] (T51) at (5*\xstep, 1*\ystep + 0*\ybias) {$\WriteYOne{u^{i}_{j}}$};
\node[event] (T52) at (5*\xstep, 2*\ystep + 0*\ybias) {$\WriteYOne{\ov{u}^{i}_{j}}$};

\node[event] (T61) at (6*\xstep, 1*\ystep + 0*\ybias) {$\ReadYOne{\ov{u}^{i}_{j}}$};
\node[event] (T62) at (6*\xstep, 2*\ystep + 0*\ybias) {$\ReadYOne{\ov{v}^{i}_{j}}$};
\node[event] (T63) at (6*\xstep, 3*\ystep + 0*\ybias) {$\WriteYFour{v^{i}_{j}}$};

\node[event] (T71) at (7*\xstep, 1*\ystep + 0*\ybias) {$\WriteYTwo{\ov{v}^{i}_{j}}$};
\node[event] (T72) at (7*\xstep, 2*\ystep + 0*\ybias) {$\WriteYTwo{v^{i}_{j}}$};

\node[event] (T81) at (8*\xstep, 1*\ystep + 0*\ybias) {$\WriteYTwo{u^{i}_{j}}$};
\node[event] (T82) at (8*\xstep, 2*\ystep + 0*\ybias) {$\WriteYTwo{\ov{u}^{i}_{j}}$};

\node[event] (T91) at (9*\xstep, 1*\ystep + 0*\ybias) {$\ReadYTwo{\ov{u}^{i}_{j}}$};
\node[event] (T92) at (9*\xstep, 2*\ystep + 0*\ybias) {$\ReadYTwo{\ov{v}^{i}_{j}}$};
\node[event] (T93) at (9*\xstep, 3*\ystep + 0*\ybias) {$\WriteYThree{v^{i}_{j}}$};

\node[event] (T101) at (10*\xstep, 1*\ystep + 0*\ybias) {$\ReadYThree{v^{i}_{j}}$};
\node[event] (T102) at (10*\xstep, 2*\ystep + 0*\ybias) {$\WriteYTwo{\ov{v}^{i}_{j}}$};
\node[event, draw=black] (T103) at (10*\xstep, 3*\ystep + 0*\ybias) {$\WriteXTwo{v^{i}_{j}}$};

\node[event] (T111) at (11*\xstep, 1*\ystep + 0*\ybias) {$\ReadYFour{v^{i}_{j}}$};
\node[event] (T112) at (11*\xstep, 2*\ystep + 0*\ybias) {$\WriteYOne{\ov{v}^{i}_{j}}$};
\node[event, draw=black] (T113) at (11*\xstep, 3*\ystep + 0*\ybias) {$\WriteXTwo{v^{i}_{j}}$};

\node[event, draw=black] (T121) at (12*\xstep, 1*\ystep + 0*\ybias) {$\ReadXTwo{v^{i}_{j}}$};
\node[event] (T122) at (12*\xstep, 2*\ystep + 0*\ybias) {$\WriteYThree{v^{i}_{j}}$};
\node[event] (T123) at (12*\xstep, 3*\ystep + 0*\ybias) {$\WriteYFour{v^{i}_{j}}$};

\begin{pgfonlayer}{bg}
\draw[rf, out=-160, in=25, looseness=0.5] (T51) to (T12);
\draw[rf, out=-160, in=25, looseness=0.1] (T81) to (T22);
\draw[rf] (T52) to (T61);
\draw[rf] (T82) to (T91);
\end{pgfonlayer}

\end{scope}
}

\begin{subfigure}{\textwidth}
\centering
\scalebox{\scale}{%
\begin{tikzpicture}[thick, font=\footnotesize,
pre/.style={<-,shorten >= 2pt, shorten <=2pt, very thick},
post/.style={->,shorten >= 3pt, shorten <=3pt,   thick},
seqtrace/.style={line width=1},
und/.style={very thick, draw=gray},
virt/.style={circle,draw=black!50,fill=black!20, opacity=0}]

\BaseCopyGadget

\end{tikzpicture}
}
\caption{\label{subfig:copy_gadget_rlx}
The copy gadget $\CopyGadget{i}{j}$.
}
\end{subfigure}
\\
\begin{subfigure}{\textwidth}
\centering
\scalebox{\scale}{
\begin{tikzpicture}[thick, font=\footnotesize,
pre/.style={<-,shorten >= 2pt, shorten <=2pt, very thick},
post/.style={->,shorten >= 3pt, shorten <=3pt,   thick},
seqtrace/.style={line width=1},
und/.style={very thick, draw=gray},
virt/.style={circle,draw=black!50,fill=black!20, opacity=0}]

\BaseCopyGadget

\begin{pgfonlayer}{fg}
\draw[rf, out=50, in=-165, looseness=1.5] (T13) to  node [pos=0.2] (fir) {\circledsmall{1}} (T31);
\draw[rf, out=150, in=-20, looseness=0.5] (T42) to node [pos=0.85] (sec) {\circledsmall{2}} (T11);
\draw[mo] (T42) -- (T51);
\draw[rf, out=160, in=30, looseness=0.2] (T112) to node [pos=0.8] (thi) {\circledsmall{3}} (T62);
\draw[rf] (T123) -- (T111) node [pos=0.2] (thi) {\circledsmall{4}};
\draw[rf, out=50, in=-165, looseness=1.5] (T103) to node[pos=.65] (four) {\circledsmall{5}} (T121);
\draw[rf] (T93) to (T101);
\draw[rf] (T71) to (T92);
\draw[rf] (T33) to (T21);

\draw[mo, out=-20, in=-160, looseness=0.2] (T52) to (T112);
\draw[mo] (T82) to (T71);
\draw[mo, out=20, in=-130, looseness=1.5] (T33) to (T81);
\end{pgfonlayer}

\end{tikzpicture}
}
\caption{\label{subfig:copy_gadget_rlxFalse}
Choosing $\TReadXOne{v^i_j}{t_3}$ to read from $t_1$ forces the sequence of $\rf$ and $\mo$ edges shown.
}
\end{subfigure}
\\
\begin{subfigure}{\textwidth}
\centering
\scalebox{\scale}{%
\begin{tikzpicture}[thick, font=\footnotesize,
pre/.style={<-,shorten >= 2pt, shorten <=2pt, very thick},
post/.style={->,shorten >= 3pt, shorten <=3pt,   thick},
seqtrace/.style={line width=1},
und/.style={very thick, draw=gray},
virt/.style={circle,draw=black!50,fill=black!20, opacity=0}]

\BaseCopyGadget

\begin{pgfonlayer}{bg}

\draw[rf] (T23) -- (T31) node [pos=0.2] (fir) {\circledsmall{1}};
\draw[rf, out=160, in=-20, looseness=0.5] (T72) to node [pos=0.1] (sec) {\circledsmall{2}} (T21); 
\draw[mo] (T72) to (T81);
\draw[rf]  (T102) -- (T92) node [pos=0.2,below] (th) {\circledsmall{3}};
\draw[rf] (T122) -- (T101) node [pos=0.6] (fo) {\circledsmall{4}};
\draw[rf] (T113) -- (T121) node [pos=0.8] (fiv) {\circledsmall{5}};
\draw[rf, out=10, in=-160] (T63) to (T111);
\draw[rf] (T41) to (T62);
\draw[rf] (T32) to (T11);

\draw[mo] (T52) to (T41);
\draw[mo, out=-20, in=-160, bend left=20] (T82) to (T102);
\draw[mo] (T32) to (T51);
\end{pgfonlayer}

\end{tikzpicture}
}
\caption{\label{subfig:copy_gadget_rlxTrue}
Choosing $\TReadXOne{v^i_j}{t_3}$ to read from $t_2$ forces the sequence of $\rf$ and $\mo$ edges shown.
}
\end{subfigure}
\caption{\label{fig:copy_gadget_rlx}
The copy gadget $\CopyGadget{i}{j}$ (\subref{subfig:copy_gadget_rlx}) captures the Boolean assignment to variable $s_j$ in phase $i$.
There are two ways to realize this gadget, by choosing which of the two writes $\WriteXOne{v^i_j}$ the read $\ReadXOne{v^i_j}$ observes.
Choosing the write of $t_1$ (\subref{subfig:copy_gadget_rlxFalse}) corresponds to setting $s_j=\bfalse$ and also forces $\ReadXTwo{v^i_j}$ to read from $t_4$. 
Choosing the write of $t_2$ (\subref{subfig:copy_gadget_rlxTrue}) corresponds to setting $s_j=\btrue$ and also forces $\ReadXTwo{v^i_j}$ to read from $t_5$.
This $\rf$ coupling is formalized in \cref{lem:rlx_soundness_copy}.
The edge numbers specify the order in which $\rf$-edges are inferred.
}
\end{figure}

%% file: figures/copy_gadget_down_rlx.tex

\begin{figure}
\newcommand{\xdisposition}{3.4*\xstep}
\newcommand{\ydisposition}{2*\ystep}
\newcommand{\xtstep}{0.75}
\newcommand{\ytstep}{0.7}
\newcommand{\ybias}{-0.3 }
\newcommand{\xstep}{1.435}
\newcommand{\ystep}{-1}
\newcommand{\xtscale}{0.8}
\def \numevents{3}
\def\scale{0.825}
\newcommand{\BaseCopyGadget}{
\begin{scope}[shift={(0*\xdisposition,0*\ydisposition)}]

\foreach \x [evaluate=\x as \i using ({int(\x)})] in {1,...,3}{
\node[] (T\i1) at (\i*\xstep,0.15) {\normalsize $t_{\x}$};
\node[] (T\i2) at (\i*\xstep,\numevents * \ystep) {};
\draw[seqtrace] (T\i1) to (T\i2);
}

\foreach \x [evaluate=\x as \i using ({int(\x+3)})] in {1,...,6}{
\node[] (T\i1) at (\i*\xstep,0.15) {\normalsize $g_{\x}$};
\node[] (T\i2) at (\i*\xstep,\numevents * \ystep) {};
\draw[seqtrace] (T\i1) to (T\i2);
}

\foreach \x [evaluate=\x as \i using ({int(\x+6)})] in {4,...,6}{
\node[] (T\i1) at (\i*\xstep,0.15) {\normalsize $t_{\x}$};
\node[] (T\i2) at (\i*\xstep,\numevents * \ystep) {};
\draw[seqtrace] (T\i1) to (T\i2);
}

\node[event] (T11) at (1*\xstep, 1*\ystep + 0*\ybias) {$\ReadZOne{v^{i+1}_{j}}$};
\node[event] (T12) at (1*\xstep, 2*\ystep + 0*\ybias) {$\ReadZOne{u^{i+1}_{j}}$};
\node[event, draw=black] (T13) at (1*\xstep, 3*\ystep + 0*\ybias) {$\WriteXOne{v^{i+1}_{j}}$};

\node[event] (T21) at (2*\xstep, 1*\ystep + 0*\ybias) {$\ReadZTwo{v^{i+1}_{j}}$};
\node[event] (T22) at (2*\xstep, 2*\ystep + 0*\ybias) {$\ReadZTwo{u^{i+1}_{j}}$};
\node[event, draw=black] (T23) at (2*\xstep, 3*\ystep + 0*\ybias) {$\WriteXOne{v^{i+1}_{j}}$};

\node[event, draw=black] (T31) at (3*\xstep, 1*\ystep + 0*\ybias) {$\ReadXOne{v^{i+1}_{j}}$};
\node[event] (T32) at (3*\xstep, 2*\ystep + 0*\ybias) {$\WriteZOne{v^{i+1}_{j}}$};
\node[event] (T33) at (3*\xstep, 3*\ystep + 0*\ybias) {$\WriteZTwo{v^{i+1}_{j}}$};

\node[event] (T41) at (4*\xstep, 1*\ystep + 0*\ybias) {$\WriteZOne{\ov{v}^{i}_{j}}$};
\node[event] (T42) at (4*\xstep, 2*\ystep + 0*\ybias) {$\WriteZOne{v^{i+1}_{j}}$};

\node[event] (T51) at (5*\xstep, 1*\ystep + 0*\ybias) {$\WriteZOne{u^{i+1}_{j}}$};
\node[event] (T52) at (5*\xstep, 2*\ystep + 0*\ybias) {$\WriteZOne{\ov{u}^{i}_{j}}$};

\node[event] (T61) at (6*\xstep, 1*\ystep + 0*\ybias) {$\ReadZOne{\ov{u}^{i}_{j}}$};
\node[event] (T62) at (6*\xstep, 2*\ystep + 0*\ybias) {$\ReadZOne{\ov{v}^{i}_{j}}$};
\node[event] (T63) at (6*\xstep, 3*\ystep + 0*\ybias) {$\WriteZFour{v^{i}_{j}}$};

\node[event] (T71) at (7*\xstep, 1*\ystep + 0*\ybias) {$\WriteZTwo{\ov{v}^{i}_{j}}$};
\node[event] (T72) at (7*\xstep, 2*\ystep + 0*\ybias) {$\WriteZTwo{v^{i+1}_{j}}$};

\node[event] (T81) at (8*\xstep, 1*\ystep + 0*\ybias) {$\WriteZTwo{u^{i+1}_{j}}$};
\node[event] (T82) at (8*\xstep, 2*\ystep + 0*\ybias) {$\WriteZTwo{\ov{u}^{i}_{j}}$};

\node[event] (T91) at (9*\xstep, 1*\ystep + 0*\ybias) {$\ReadZTwo{\ov{u}^{i}_{j}}$};
\node[event] (T92) at (9*\xstep, 2*\ystep + 0*\ybias) {$\ReadZTwo{\ov{v}^{i}_{j}}$};
\node[event] (T93) at (9*\xstep, 3*\ystep + 0*\ybias) {$\WriteZThree{v^{i}_{j}}$};

\node[event] (T101) at (10*\xstep, 1*\ystep + 0*\ybias) {$\ReadZThree{v^{i}_{j}}$};
\node[event] (T102) at (10*\xstep, 2*\ystep + 0*\ybias) {$\WriteZTwo{\ov{v}^{i}_{j}}$};
\node[event, draw=black] (T103) at (10*\xstep, 3*\ystep + 0*\ybias) {$\WriteXTwo{v^{i}_{j}}$};

\node[event] (T111) at (11*\xstep, 1*\ystep + 0*\ybias) {$\ReadZFour{v^{i}_{j}}$};
\node[event] (T112) at (11*\xstep, 2*\ystep + 0*\ybias) {$\WriteZOne{\ov{v}^{i}_{j}}$};
\node[event, draw=black] (T113) at (11*\xstep, 3*\ystep + 0*\ybias) {$\WriteXTwo{v^{i}_{j}}$};

\node[event, draw=black] (T121) at (12*\xstep, 1*\ystep + 0*\ybias) {$\ReadXTwo{v^{i}_{j}}$};
\node[event] (T122) at (12*\xstep, 2*\ystep + 0*\ybias) {$\WriteZThree{v^{i}_{j}}$};
\node[event] (T123) at (12*\xstep, 3*\ystep + 0*\ybias) {$\WriteZFour{v^{i}_{j}}$};

\begin{pgfonlayer}{bg}
\draw[rf, out=-160, in=25, looseness=0.5] (T51) to (T12);
\draw[rf, out=-160, in=25, looseness=0.1] (T81) to (T22);
\draw[rf] (T52) to (T61);
\draw[rf] (T82) to (T91);
\end{pgfonlayer}

\end{scope}
}
%
%
\begin{subfigure}{\textwidth}
\centering
\scalebox{\scale}{
\begin{tikzpicture}[thick, font=\footnotesize,
pre/.style={<-,shorten >= 2pt, shorten <=2pt, very thick},
post/.style={->,shorten >= 3pt, shorten <=3pt,   thick},
seqtrace/.style={line width=1},
und/.style={very thick, draw=gray},
virt/.style={circle,draw=black!50,fill=black!20, opacity=0}]

\BaseCopyGadget

%

\begin{pgfonlayer}{fg}
\draw[rf, out=50, in=-165, looseness=1.5] (T13) to  node [pos=0.2] (fir) {\circledsmall{1}} (T31);
\draw[rf, out=150, in=-20, looseness=0.5] (T42) to node [pos=0.85] (sec) {\circledsmall{2}} (T11);
\draw[mo] (T42) -- (T51);
\draw[rf, out=160, in=30, looseness=0.2] (T112) to node [pos=0.8] (thi) {\circledsmall{3}} (T62);
\draw[rf] (T123) -- (T111) node [pos=0.2] (thi) {\circledsmall{4}};
\draw[rf, out=50, in=-165, looseness=1.5] (T103) to node[pos=.65] (four) {\circledsmall{5}} (T121);
\draw[rf] (T93) to (T101);
\draw[rf] (T71) to (T92);
\draw[rf] (T33) to (T21);

\draw[mo, out=-20, in=-160, looseness=0.2] (T52) to (T112);
\draw[mo] (T82) to (T71);
\draw[mo, out=20, in=-130, looseness=1.5] (T33) to (T81);
\end{pgfonlayer}

\end{tikzpicture}
}
\caption{\label{subfig:copy_gadget_down_rlxFalse}
Choosing $\TReadXOne{v^i_j}{t_3}$ to read from $t_1$ forces the sequence of $\rf$ and $\mo$ edges shown.
}
\end{subfigure}
\\
\begin{subfigure}{\textwidth}
\centering
\scalebox{\scale}{%
\begin{tikzpicture}[thick, font=\footnotesize,
pre/.style={<-,shorten >= 2pt, shorten <=2pt, very thick},
post/.style={->,shorten >= 3pt, shorten <=3pt,   thick},
seqtrace/.style={line width=1},
und/.style={very thick, draw=gray},
virt/.style={circle,draw=black!50,fill=black!20, opacity=0}]

\BaseCopyGadget

%

\begin{pgfonlayer}{bg}

\draw[rf] (T23) -- (T31) node [pos=0.2] (fir) {\circledsmall{1}};
\draw[rf, out=160, in=-20, looseness=0.5] (T72) to node [pos=0.1] (sec) {\circledsmall{2}} (T21); 
\draw[mo] (T72) to (T81);
\draw[rf]  (T102) -- (T92) node [pos=0.2,below] (th) {\circledsmall{3}};
\draw[rf] (T122) -- (T101) node [pos=0.6] (fo) {\circledsmall{4}};
\draw[rf] (T113) -- (T121) node [pos=0.8] (fiv) {\circledsmall{5}};
\draw[rf, out=10, in=-160] (T63) to (T111);
\draw[rf] (T41) to (T62);
\draw[rf] (T32) to (T11);

\draw[mo] (T52) to (T41);
\draw[mo, out=-20, in=-160, bend left=20] (T82) to (T102);
\draw[mo] (T32) to (T51);
\end{pgfonlayer}

\end{tikzpicture}
}
\caption{\label{subfig:copy_gadget_down_rlxTrue}
Choosing $\TReadXOne{v^i_j}{t_3}$ to read from $t_2$ forces the sequence of $\rf$ and $\mo$ edges shown.
}
\end{subfigure}
\caption{\label{fig:copy_gadget_down_rlx}
The copy-down gadget $\CopyGadgetDown{i}{j}$ is very similar to $\CopyGadget{i}{j}$.
Choosing $\ReadXOne{v^{i+1}_j}$ to read from $t_1$ (\subref{subfig:copy_gadget_down_rlxFalse}) corresponds to setting $s_j=\bfalse$ and also forces $\ReadXTwo{v^i_j}$ to read from $t_4$. 
Choosing $\ReadXOne{v^{i+1}_j}$ to read from $t_2$ (\subref{subfig:copy_gadget_down_rlxTrue}) corresponds to setting $s_j=\btrue$ and also forces $\ReadXTwo{v^i_j}$ to read from $t_5$.
This $\rf$ coupling is formalized in \cref{lem:rlx_soundness_copy}.
The edge numbers specify the order in which $\rf$-edges are inferred.
}
\end{figure}

%% file: figures/at_most_gadget_rlx.tex

\begin{figure}
\newcommand{\xdisposition}{0}
\newcommand{\ydisposition}{0}
\newcommand{\xtstep}{0.75}
\newcommand{\ytstep}{0.7}
\newcommand{\ybias}{-0.3 }
\newcommand{\xstep}{1.4}
\newcommand{\ystep}{-0.8}
\newcommand{\xtscale}{0.8}
\def\crossoutopacity{0.3}
\def \numevents{4}
\newcommand{\BaseAtMostOneGadget}[1]{

\node[] (S11) at (1*\xstep,0.15) {\normalsize $t_1$};
\node[] (S12) at (1*\xstep,\numevents * \ystep) {};
\node[] (S21) at (2*\xstep,0.15) {\normalsize $t_2$};
\node[] (S22) at (2*\xstep,\numevents * \ystep) {};
\node[] (S31) at (3*\xstep,0.15) {\normalsize $t_3$};
\node[] (S32) at (3*\xstep,\numevents * \ystep) {};
\node[] (S41) at (4*\xstep,0.15) {\normalsize $h_{c}$};
\node[] (S42) at (4*\xstep,\numevents * \ystep) {};

\draw[seqtrace] (S11) to (S12);
\draw[seqtrace] (S21) to (S22);
\draw[seqtrace] (S31) to (S32);
\draw[seqtrace] (S41) to (S42);

\node[event, draw=black] (12) at (1*\xstep, 2*\ystep + 0*\ybias) {$\WriteXOne{v^{i}_{j}}$};
\node[event, draw=black] (14) at (1*\xstep, 4*\ystep + 0*\ybias) {$\WriteXOne{v^{i}_{k}}$};

\node[event] (21) at (2*\xstep, 1*\ystep + 0*\ybias) {$\ReadA{v^{i}_{j}}{c}$};
\node[event, draw=black] (22) at (2*\xstep, 2*\ystep + 0*\ybias) {$\WriteXOne{v^{i}_{j}}$};
\node[event] (23) at (2*\xstep, 3*\ystep + 0*\ybias) {$\ReadA{v^{i}_{k}}{c}$};
\node[event, draw=black] (24) at (2*\xstep, 4*\ystep + 0*\ybias) {$\WriteXOne{v^{i}_{k}}$};

\node[event, draw=black] (31) at (3*\xstep, 1*\ystep + 0*\ybias) {$\ReadXOne{v^{i}_{j}}$};
\node[event] (32) at (3*\xstep, 2*\ystep + 0*\ybias) {$\WriteA{v^{i}_{j}}{c}$};
\node[event, draw=black] (33) at (3*\xstep, 3*\ystep + 0*\ybias) {$\ReadXOne{v^{i}_{k}}$};
\node[event] (34) at (3*\xstep, 4*\ystep + 0*\ybias) {$\WriteA{v^{i}_{k}}{c}$};

\ifthenelse{\equal{#1}{0}}{
\node[event] (42) at (4*\xstep, 2*\ystep + 0*\ybias) {$\WriteA{v^{i}_{k}}{c}$};
\node[event] (44) at (4*\xstep, 4*\ystep + 0*\ybias) {$\WriteA{v^{i}_{j}}{c}$};
}{
\ifthenelse{\equal{#1}{1}}{
\node[event, cross out, draw, draw opacity=\crossoutopacity] (42) at (4*\xstep, 2*\ystep + 0*\ybias) {$\WriteA{v^{i}_{k}}{c}$};
\node[event] (44) at (4*\xstep, 4*\ystep + 0*\ybias) {$\WriteA{v^{i}_{j}}{c}$};
}{
\ifthenelse{\equal{#1}{2}}{
\node[event] (42) at (4*\xstep, 2*\ystep + 0*\ybias) {$\WriteA{v^{i}_{k}}{c}$};
\node[event, cross out, draw, draw opacity=\crossoutopacity] (44) at (4*\xstep, 4*\ystep + 0*\ybias) {$\WriteA{v^{i}_{j}}{c}$};
}{
\node[event, cross out, draw, draw opacity=\crossoutopacity] (42) at (4*\xstep, 2*\ystep + 0*\ybias) {$\WriteA{v^{i}_{k}}{c}$};
\node[event, cross out, draw, draw opacity=\crossoutopacity] (44) at (4*\xstep, 4*\ystep + 0*\ybias) {$\WriteA{v^{i}_{j}}{c}$};
}}}

}
\def\scale{1}
\centering
\begin{subfigure}[b]{0.475\textwidth}
\centering
\scalebox{\scale}{
\begin{tikzpicture}[thick, font=\footnotesize,
pre/.style={<-,shorten >= 2pt, shorten <=2pt, very thick},
post/.style={->,shorten >= 3pt, shorten <=3pt,   thick},
seqtrace/.style={line width=1},
und/.style={very thick, draw=gray},
virt/.style={circle,draw=black!50,fill=black!20, opacity=0}]

\BaseAtMostOneGadget{0}

\end{tikzpicture}
}
\caption{\label{subfig:at_most_one_gadget_rlx}
The at-most-one-true gadget $\AtMostOneGadget{c}{i}{j}{k}$.
}
\end{subfigure}
\hfill
\begin{subfigure}[b]{0.475\textwidth}
\centering
\scalebox{\scale}{
\begin{tikzpicture}[thick, font=\footnotesize,
pre/.style={<-,shorten >= 2pt, shorten <=2pt, very thick},
post/.style={->,shorten >= 3pt, shorten <=3pt,   thick},
seqtrace/.style={line width=1},
und/.style={very thick, draw=gray},
virt/.style={circle,draw=black!50,fill=black!20, opacity=0}]

\BaseAtMostOneGadget{1}

\begin{pgfonlayer}{bg}
\draw[rf] (22) -- (31) node [pos=0.6] (fir) {\circledsmall{1}};
\draw[rf] (44) -- (21) node [pos=0.1] (sec) {\circledsmall{2}};
\draw[rf] (34) -- (23) node [pos=0.3] (thr) {\circledsmall{3}};
\draw[rf] (14) -- (33) node [pos=0.1, above] (four) {\circledsmall{4}};

\draw[mo] (22) to (14);
\draw[mo] (44) to (34);
\end{pgfonlayer}

\end{tikzpicture}
}
\caption{\label{subfig:at_most_one_gadget_rlx_true_false}
Resolving with $s_j=\btrue$ and $s_k=\bfalse$.
}
\end{subfigure}
\\
\begin{subfigure}[b]{0.475\textwidth}
\centering
\scalebox{\scale}{
\begin{tikzpicture}[thick, font=\footnotesize,
pre/.style={<-,shorten >= 2pt, shorten <=2pt, very thick},
post/.style={->,shorten >= 3pt, shorten <=3pt,   thick},
seqtrace/.style={line width=1},
und/.style={very thick, draw=gray},
virt/.style={circle,draw=black!50,fill=black!20, opacity=0}]

\BaseAtMostOneGadget{2}

\begin{pgfonlayer}{bg}
\draw[rf] (24) -- (33) node [pos=0.4] (fir) {\circledsmall{1}};
\draw[rf] (42) -- (23) node [pos=0.4] (sec) {\circledsmall{2}};;
\draw[rf] (32) -- (21) node [pos=0.4] (third) {\circledsmall{3}};
\draw[rf] (12) -- (31) node [pos=0.1,above] (four) {\circledsmall{4}};

\draw[mo] (12) to (24);
\draw[mo] (32) to (42);
\end{pgfonlayer}

\end{tikzpicture}
}
\caption{\label{subfig:at_most_one_gadget_rlx_false_true}
Resolving with $s_j=\bfalse$ and $s_k=\btrue$.
}
\end{subfigure}
\hfill
\begin{subfigure}[b]{0.475\textwidth}
\centering
\scalebox{\scale}{
\begin{tikzpicture}[thick, font=\footnotesize,
pre/.style={<-,shorten >= 2pt, shorten <=2pt, very thick},
post/.style={->,shorten >= 3pt, shorten <=3pt,   thick},
seqtrace/.style={line width=1},
und/.style={very thick, draw=gray},
virt/.style={circle,draw=black!50,fill=black!20, opacity=0}]

\BaseAtMostOneGadget{3}

\begin{pgfonlayer}{bg}
\draw[rf] (12) -- (31) node [pos=0.1,above=.1mm] (fir) {\circledsmall{1}};
\draw[rf] (14) -- (33) node [pos=0.1,above=.1mm] (fir) {\circledsmall{1}};
\draw[rf] (32) -- (21) node [pos=0.4] (sec) {\circledsmall{2}};
\draw[rf] (34) -- (23) node [pos=0.4] (sec) {\circledsmall{2}};

\end{pgfonlayer}

\end{tikzpicture}
}
\caption{\label{subfig:at_most_one_gadget_rlx_false_false}
Resolving with $s_j=\bfalse$ and $s_k=\bfalse$.
}
\end{subfigure}
\caption{\label{fig:at_most_one_gadget_rlx}
The at-most-one-true gadget $\AtMostOneGadget{c}{i}{j}{k}$ parameterized by $c\in [3]$,
and the three ways to resolve it depending on the boolean assignment to $s_j$ and $s_k$ (\subref{subfig:at_most_one_gadget_rlx_true_false}, \subref{subfig:at_most_one_gadget_rlx_false_true}, \subref{subfig:at_most_one_gadget_rlx_false_false}).
These $\rf$ constraints  are formalized in \cref{lem:rlx_soundness_at_most_one}.
The edge numbers specify the order in which $\rf$-edges are inferred.
The crossed-out events are ignored in our analysis for simplicity (see \cref{lem:rlx_completeness_inactive_writes}).
}
\end{figure}

%% file: figures/at_least_gadget_rlx.tex

\begin{figure}
\newcommand{\xdisposition}{0}
\newcommand{\ydisposition}{0}
\newcommand{\xtstep}{0.75}
\newcommand{\ytstep}{0.7}
\newcommand{\ybias}{-0.3 }
\newcommand{\xstep}{1.4}
\newcommand{\ystep}{-0.8}
\newcommand{\xtscale}{0.8}
\def\crossoutopacity{0.3}
\newcommand{\BaseAtLeastOneGadget}[1]{
\node[] (S11) at (1*\xstep,0.15) {\normalsize $t_1$};
\node[] (S12) at (1*\xstep,\numevents * \ystep) {};
\node[] (S21) at (2*\xstep,0.15) {\normalsize $t_2$};
\node[] (S22) at (2*\xstep,\numevents * \ystep) {};
\node[] (S31) at (3*\xstep,0.15) {\normalsize $t_3$};
\node[] (S32) at (3*\xstep,\numevents * \ystep) {};
\node[] (S41) at (4*\xstep,0.15) {\normalsize $p$};
\node[] (S42) at (4*\xstep,\numevents * \ystep) {};
\node[] (S51) at (5*\xstep,0.15) {\normalsize $q$};
\node[] (S52) at (5*\xstep,\numevents * \ystep) {};

\draw[seqtrace] (S11) to (S12);
\draw[seqtrace] (S21) to (S22);
\draw[seqtrace] (S31) to (S32);
\draw[seqtrace] (S41) to (S42);
\draw[seqtrace] (S51) to (S52);

\node[event] (11) at (1*\xstep, 1*\ystep + 0*\ybias) {$\ReadB{v^{i}_j}$};
\node[event, draw=black] (12) at (1*\xstep, 2*\ystep + 0*\ybias) {$\WriteXOne{v^{i}_{j}}$};
\node[event] (13) at (1*\xstep, 3*\ystep + 0*\ybias) {$\ReadB{v^{i}_k}$};
\node[event, draw=black] (14) at (1*\xstep, 4*\ystep + 0*\ybias) {$\WriteXOne{v^{i}_{k}}$};
\node[event] (15) at (1*\xstep, 5*\ystep + 0*\ybias) {$\ReadB{v^{i}_\ell}$};
\node[event, draw=black] (16) at (1*\xstep, 6*\ystep + 0*\ybias) {$\WriteXOne{v^{i}_{\ell}}$};

\node[event, draw=black] (22) at (2*\xstep, 2*\ystep + 0*\ybias) {$\WriteXOne{v^{i}_{j}}$};
\node[event, draw=black] (24) at (2*\xstep, 4*\ystep + 0*\ybias) {$\WriteXOne{v^{i}_{k}}$};
\node[event, draw=black] (26) at (2*\xstep, 6*\ystep + 0*\ybias) {$\WriteXOne{v^{i}_{\ell}}$};

\node[event, draw=black] (31) at (3*\xstep, 1*\ystep + 0*\ybias) {$\ReadXOne{v^{i}_{j}}$};
\node[event] (32) at (3*\xstep, 2*\ystep + 0*\ybias) {$\WriteB{v^{i}_j}$};
\node[event, draw=black] (33) at (3*\xstep, 3*\ystep + 0*\ybias) {$\ReadXOne{v^{i}_{k}}$};
\node[event] (34) at (3*\xstep, 4*\ystep + 0*\ybias) {$\WriteB{v^{i}_k}$};
\node[event, draw=black] (35) at (3*\xstep, 5*\ystep + 0*\ybias) {$\ReadXOne{v^{i}_{\ell}}$};
\node[event] (36) at (3*\xstep, 6*\ystep + 0*\ybias) {$\WriteB{v^{i}_\ell}$};

%

\ifthenelse{\equal{#1}{0}}{
\node[event] (42) at (4*\xstep, 2*\ystep + 0*\ybias) {$\WriteB{v^{i}_\ell}$};
\node[event] (44) at (4*\xstep, 4*\ystep + 0*\ybias) {$\WriteB{v^{i}_j}$};

\node[event] (53) at (5*\xstep, 3*\ystep + 0*\ybias) {$\WriteB{v^{i}_\ell}$};
\node[event] (55) at (5*\xstep, 5*\ystep + 0*\ybias) {$\WriteB{v^{i}_k}$};
}{
\ifthenelse{\equal{#1}{1}}{
\node[event, cross out, draw, draw opacity=\crossoutopacity] (42) at (4*\xstep, 2*\ystep + 0*\ybias) {$\WriteB{v^{i}_\ell}$};
\node[event] (44) at (4*\xstep, 4*\ystep + 0*\ybias) {$\WriteB{v^{i}_j}$};

\node[event, cross out, draw, draw opacity=\crossoutopacity] (53) at (5*\xstep, 3*\ystep + 0*\ybias) {$\WriteB{v^{i}_\ell}$};
\node[event] (55) at (5*\xstep, 5*\ystep + 0*\ybias) {$\WriteB{v^{i}_k}$};
}{
\ifthenelse{\equal{#1}{2}}{
\node[event, cross out, draw, draw opacity=\crossoutopacity] (42) at (4*\xstep, 2*\ystep + 0*\ybias) {$\WriteB{v^{i}_\ell}$};
\node[event] (44) at (4*\xstep, 4*\ystep + 0*\ybias) {$\WriteB{v^{i}_j}$};

\node[event] (53) at (5*\xstep, 3*\ystep + 0*\ybias) {$\WriteB{v^{i}_\ell}$};
\node[event, cross out, draw, draw opacity=\crossoutopacity] (55) at (5*\xstep, 5*\ystep + 0*\ybias) {$\WriteB{v^{i}_k}$};
}{
\ifthenelse{\equal{#1}{3}}{
\node[event] (42) at (4*\xstep, 2*\ystep + 0*\ybias) {$\WriteB{v^{i}_\ell}$};
\node[event, cross out, draw, draw opacity=\crossoutopacity] (44) at (4*\xstep, 4*\ystep + 0*\ybias) {$\WriteB{v^{i}_j}$};

\node[event, cross out, draw, draw opacity=\crossoutopacity] (53) at (5*\xstep, 3*\ystep + 0*\ybias) {$\WriteB{v^{i}_\ell}$};
\node[event] (55) at (5*\xstep, 5*\ystep + 0*\ybias) {$\WriteB{v^{i}_k}$};
}}}}

}
\def \numevents{6}
\def\scale{0.97}
\centering
\begin{subfigure}[b]{0.475\textwidth}
\centering
\scalebox{\scale}{%
\begin{tikzpicture}[thick, font=\footnotesize,
pre/.style={<-,shorten >= 2pt, shorten <=2pt, very thick},
post/.style={->,shorten >= 3pt, shorten <=3pt,   thick},
seqtrace/.style={line width=1},
und/.style={very thick, draw=gray},
virt/.style={circle,draw=black!50,fill=black!20, opacity=0}]

\BaseAtLeastOneGadget{0}

\end{tikzpicture}
}
\caption{\label{subfig:at_least_one_gadget_rlx}
The at-least-one-true gadget $\AtLeastOneGadget{i}{j}{k}{\ell}$.
}
\end{subfigure}
\hfill
\begin{subfigure}[b]{0.475\textwidth}
\centering
\scalebox{\scale}{%
\begin{tikzpicture}[thick, font=\footnotesize,
pre/.style={<-,shorten >= 2pt, shorten <=2pt, very thick},
post/.style={->,shorten >= 3pt, shorten <=3pt,   thick},
seqtrace/.style={line width=1},
und/.style={very thick, draw=gray},
virt/.style={circle,draw=black!50,fill=black!20, opacity=0}]

\BaseAtLeastOneGadget{1}

\begin{pgfonlayer}{bg}
\draw[rf] (12) -- (31) node [pos=0.7] (fir) {\circled{1}};
\draw[rf] (44) -- (11) node [pos=0.1] (sec1) {\circled{2}};
\draw[rf] (14) -- (33) node [pos=0.7] (fir1) {\circled{1}};
\draw[rf] (55) -- (13) node [pos=0.1] (sec2) {\circled{2}};
\draw[rf] (36) -- (15) node [pos=0.7] (th) {\circled{3}};
\draw[rf] (26) -- (35) node [pos=0.4,above=0.1mm] (four) {\circled{4}};


\draw[mo, bend left] (14) to (26);
\draw[mo] (44) to (55);
\draw[mo] (55) to (36);
\end{pgfonlayer}

\end{tikzpicture}
}
\caption{\label{subfig:at_least_one_gadget_rlx_false_false_true}
Resolving with $s_j=\bfalse$, $s_k=\bfalse$ and $s_\ell=\btrue$.
}
\end{subfigure}
\\[2em]
\begin{subfigure}[b]{0.475\textwidth}
\centering
\scalebox{\scale}{%
\begin{tikzpicture}[thick, font=\footnotesize,
pre/.style={<-,shorten >= 2pt, shorten <=2pt, very thick},
post/.style={->,shorten >= 3pt, shorten <=3pt,   thick},
seqtrace/.style={line width=1},
und/.style={very thick, draw=gray},
virt/.style={circle,draw=black!50,fill=black!20, opacity=0}]

\BaseAtLeastOneGadget{2}

\begin{pgfonlayer}{bg}
\draw[rf] (12) -- (31) node [pos=0.7] (fir) {\circled{1}};
\draw[rf] (44) -- (11) node [pos=0.9] (fir) {\circled{2}};
\draw[rf] (24) -- (33) node [pos=0.5] (fir) {\circled{4}};;
\draw[rf] (16) -- (35) node [pos=0.7] (fir) {\circled{1}};
\draw[rf] (53) -- (15) node [pos=0.9] (fir) {\circled{2}};
\draw[rf] (34) -- (13) node [pos=0.8] (fir) {\circled{3}};


\draw[mo, bend left] (12) to (24);
\draw[mo, bend left] (24) to (16);
\draw[mo] (44) to (34);
\draw[mo, bend left =10] (34) to (53);
\end{pgfonlayer}

\end{tikzpicture}
}
\caption{\label{subfig:at_least_one_gadget_rlx_false_true_false}
Resolving with $s_j=\bfalse$, $s_k=\btrue$ and $s_\ell=\bfalse$.
}
\end{subfigure}
\hfill
\begin{subfigure}[b]{0.475\textwidth}
\centering
\scalebox{\scale}{%
\begin{tikzpicture}[thick, font=\footnotesize,
pre/.style={<-,shorten >= 2pt, shorten <=2pt, very thick},
post/.style={->,shorten >= 3pt, shorten <=3pt,   thick},
seqtrace/.style={line width=1},
und/.style={very thick, draw=gray},
virt/.style={circle,draw=black!50,fill=black!20, opacity=0}]

\BaseAtLeastOneGadget{3}

\begin{pgfonlayer}{bg}
\draw[rf] (22) -- (31) node [pos=0.5,above=.1mm] (fir) {\circled{4}};
\draw[rf] (14) -- (33) node [pos=0.8] (fir) {\circled{1}};
\draw[rf] (16) -- (35) node [pos=0.8] (fir) {\circled{1}};
\draw[rf] (32) -- (11) node [pos=0.8] (fir) {\circled{3}};
\draw[rf] (55) -- (13) node [pos=0.9] (fir) {\circled{2}};
\draw[rf] (42) -- (15) node [pos=0.9] (fir) {\circled{2}};

\draw[mo, bend right,in=-180] (22) to (14);
\draw[mo, bend right] (14) to (26);
\draw[mo] (32) to (55);
\draw[mo] (55) to (42);
\end{pgfonlayer}

\end{tikzpicture}
}
\caption{\label{subfig:at_least_one_gadget_rlx_true_false_false}
Resolving with $s_j=\btrue$, $s_k=\bfalse$ and $s_\ell=\bfalse$.
}
\end{subfigure}
\caption{\label{fig:at_least_one_gadget_rlx}
The at-least-one-true gadget $\AtLeastOneGadget{i}{j}{k}{\ell}$ (\subref{subfig:at_least_one_gadget_rlx}) and the three ways to resolve it depending on the boolean assignment to $s_j$, $s_k$ and $s_\ell$ (\subref{subfig:at_least_one_gadget_rlx_true_false_false}, \subref{subfig:at_least_one_gadget_rlx_false_true_false}, \subref{subfig:at_least_one_gadget_rlx_false_false_true}).
These $\rf$ constraints  are formalized in \cref{lem:rlx_soundness_at_least_one}.
The edge numbers specify the order in which $\rf$-edges are inferred.
The crossed-out events are ignored in our analysis for simplicity (see \cref{lem:rlx_completeness_inactive_writes}).
}
\end{figure}

%% file: lower_threads_locations_rlx_soundness.tex

\subsection{Soundness} 
\label{SEC:LOWER_THREADS_LOCATIONS_RLX_SOUNDNESS}

We first establish the soundness of the reduction, i.e., if $\expartial$ is consistent (using an extension $\ex$) within $\arlxmm$, then $\varphi$ has a satisfying assignment.
In this direction, we will establish some intermediate lemmas.
Recall that we obtain the satisfying assignment for $\varphi$ by assigning $s_j=\btrue$ if $(\TWriteXOne{v^i_j}{t_2}, \TReadXOne{v^i_j}{t_3})\in \rf$ for all $i\in [m]$, and $s_j=\bfalse$ if $(\TWriteXOne{v^i_j}{t_1}, \TReadXOne{v^i_j}{t_3})\in \rf$ for all $i\in[m]$.
The first lemma is based on the copy gadgets $\CopyGadget{i}{j}$ and $\CopyGadgetDown{i}{j}$ and states that each phase of $\ex$ makes consistent choices for the writer of $\TReadXOne{v^i_j}{t_3}$ (i.e., whether it reads from $t_1$ or $t_2$), which makes the above assignment well-defined.
It follows from the high-level description of these two gadgets and the accompanying \cref{fig:copy_gadget_rlx} and \cref{fig:copy_gadget_down_rlx}.

\begin{restatable}{lemma}{lemrlxsoundnesscopy}\label{lem:rlx_soundness_copy}
Let $\ex=(\E, \po,\rf,\mo)$ be a concrete extension of $\expartial$ with $\ex\models \arlxmm$.
For all $i_1, i_2\in[m]$, $j\in[n]$, we have that $(\TWriteXOne{v^{i_1}_j}{t_1}, \TReadXOne{v^{i_1}_j}{t_3}) \in \rf$ iff $(\TWriteXOne{v^{i_2}_j}{t_1}, \TReadXOne{v^{i_2}_j}{t_3}) \in \rf$.
\end{restatable}

\begin{proof}
We argue by induction that for every $i\in[m-1]$, we have $(\TWriteXOne{v^{i}_j}{t_1}, \TReadXOne{v^{i}_j}{t_3}) \in \rf$ iff $(\TWriteXOne{v^{i+1}_j}{t_1}, \TReadXOne{v^{i+1}_j}{t_3}) \in \rf$. 
First, note that if $(\TWriteXOne{v^{i}_j}{t_1}, \TReadXOne{v^{i}_j}{t_3}) \in \rf$, then the copy gadget $\CopyGadget{i}{j}$ forces $(\TWriteXTwo{v^{i}_j}{t_4}, \TReadXTwo{v^{i}_j}{t_6}) \in \rf$.
Indeed, we have the following inferred sequence of $\rf$ edges (see \circledsmall{1}-\circledsmall{5} in 
\cref{subfig:copy_gadget_rlxFalse}, where \circledsmall{1} represents $(\TWriteXOne{v^{i}_j}{t_1}, \TReadXOne{v^{i}_j}{t_3}) \in \rf$).
\begin{compactenum}
\item We have $(\TReadYOne{v^i_j}{t_1}, \TWriteYOne{v^i_j}{t_3})\in (\po\cup \rf)^+$ and thus $(\TWriteYOne{v^i_j}{t_3}, \TReadYOne{v^i_j}{t_1})\not\in\rf$ due to \PORF{}.
Thus $\TReadYOne{v^i_j}{t_1}$ is forced to read from the only other available write of the same value, i.e., $(\TWriteYOne{v^i_j}{f_1}, \TReadYOne{v^i_j}{t_1})\in\rf$, depicted as \circled{2} in \cref{subfig:copy_gadget_rlxFalse}.
\item We now have $(\TWriteYOne{v^i_j}{f_1}, \TReadYOne{u^i_j}{t_1})\in(\po_{\VarYOne}\cup\rf_{\VarYOne})^+$, 
and since $(\TWriteYOne{u^i_j}{f_2}, \TReadYOne{u^i_j}{t_1})\in \rf$, we have $(\TWriteYOne{v^i_j}{f_1}, \TWriteYOne{u^i_j}{f_2})\in \mo$ due to \RlxRCoh{}.
Observe that now $(\TWriteYOne{\ov{v}^i_j}{f_1}, \TWriteYOne{\ov{u}^i_j}{f_2})\in(\po_{\VarYOne}\cup \mo_{\VarYOne})^+$, 
thus due to \RlxWCoh{}, we have 
$(\TWriteYOne{\ov{v}^i_j}{f_1}, \TWriteYOne{\ov{u}^i_j}{f_2})\in \mo_{\VarYOne}$.
Since $(\TWriteYOne{\ov{u}^i_j}{f_2}, \TReadYOne{\ov{v}^i_j}{f_3})\in (\po_{\VarYOne}\cup \rf_{\VarYOne})^+$, 
we have $(\TWriteYOne{\ov{v}^i_j}{f_1}, \TReadYOne{\ov{v}^i_j}{f_3})\not \in \rf$ due to \RlxRCoh{}.
Thus $\TReadYOne{\ov{v}^i_j}{f_3}$ is forced to read from the only other available write, i.e., $(\TWriteYOne{\ov{v}^i_j}{t_5}, \TReadYOne{\ov{v}^i_j}{f_3})\in \rf$, depicted by \circledsmall{3}. 
\item We now have $(\TReadYFour{v^i_j}{t_5}, \TWriteYFour{v^i_j}{f_3})\in (\po\cup \rf)^+$ and thus $(\TWriteYFour{v^i_j}{f_3}, \TReadYFour{v^i_j}{t_5})\not\in\rf$, due to \PORF{}.
Thus $\TReadYFour{v^i_j}{t_5}$ is forced to read from the only other available write, i.e., $(\TWriteYFour{v^i_j}{t_6}, \TReadYFour{v^i_j}{t_5})\in\rf$, 
depicted by \circledsmall{4}.
\item We now have $(\TReadXTwo{v^i_j}{t_6}, \TWriteXTwo{v^i_j}{t_5}) \in (\po\cup\rf)^+$ and thus $(\TWriteXTwo{v^i_j}{t_5}, \TReadXTwo{v^i_j}{t_6})\not \in \rf$ due to \PORF{}.
Thus $\TReadXTwo{v^i_j}{t_6}$ is forced to read from the only other available write, i.e., $(\TWriteXTwo{v^i_j}{t_4}, \TReadXTwo{v^i_j}{t_6}) \in \rf$,
depicted by \circledsmall{5}. 
\end{compactenum}
On the other hand, if $(\TWriteXOne{v^{i}_j}{t_2}, \TReadXOne{v^{i}_j}{t_3}) \in \rf$, then the copy gadget $\CopyGadget{i}{j}$ forces $(\TWriteXTwo{v^{i}_j}{t_5}, \TReadXTwo{v^{i}_j}{t_6}) \in \rf$, by a similar analysis (see \cref{subfig:copy_gadget_rlxTrue}, depicted by \circledsmall{1}-\circledsmall{5}).

Finally, a similar analysis on the copy-down gadget $\CopyGadgetDown{i}{j}$ (see \cref{fig:copy_gadget_down_rlx}, and the forced $\rf$ edges \circledsmall{1}-\circledsmall{5}) concludes that $(\TWriteXOne{v^{i+1}_j}{t_1}, \TReadXOne{v^{i+1}_j}{t_3}) \in \rf$ iff $(\TWriteXTwo{v^{i}_j}{t_4}, \TReadXTwo{v^{i}_j}{t_6}) \in \rf$,
and hence we have $(\TWriteXOne{v^{i}_j}{t_1}, \TReadXOne{v^{i}_j}{t_3}) \in \rf$ iff $(\TWriteXOne{v^{i+1}_j}{t_1}, \TReadXOne{v^{i+1}_j}{t_3}) \in \rf$, as desired.
\end{proof}

The next lemma is based on the at-most-one-true gadgets $\AtMostOneGadget{c}{i}{j}{k}$, and it is used to show that for every clause $\Clause_i$ and for each of the three pairs of literals $(s_j, s_k)$ in $\Clause_i$, at most one of them is assigned to true.
Again, it follows by a direct analysis of the accompanying figure, \cref{fig:at_most_one_gadget_rlx}.

\begin{restatable}{lemma}{lemrlxsoundnessatmostone}\label{lem:rlx_soundness_at_most_one}
Let $\ex=(\E, \po,\rf,\mo)$ be a concrete extension of $\expartial$ with $\ex\models \arlxmm$.
For every $i\in[m]$ and $j,k\in [n]$
such that $s_j$ and $s_k$ appear in clause $\Clause_i$, we have $(\TWriteXOne{v^i_j}{t_2}, \TReadXOne{v^i_j}{t_3})\not \in \rf$ or $(\TWriteXOne{v^i_k}{t_2}, \TReadXOne{v^i_k}{t_3})\not \in \rf$.
\end{restatable}

\begin{proof}
The statement follows by analyzing the at-most-one-true gadget $\AtMostOneGadget{c}{i}{j}{k}$, where $c\in[3]$ is such that $(s_j, s_k)$ is the $c$-th pair of variables in $\Clause_i$ (see \cref{fig:at_most_one_gadget_rlx}).

First, if $(\TWriteXOne{v^i_j}{t_2}, \TReadXOne{v^i_j}{t_3}) \in \rf$ holds (marked \circled{1} in \cref{subfig:at_most_one_gadget_rlx_true_false})
then $(\TWriteXOne{v^i_k}{t_2}, \TReadXOne{v^i_k}{t_3})\not \in \rf$.
Indeed, we have the following sequence of $\rf$ edges (see \circled{1}-\circled{4}, \cref{subfig:at_most_one_gadget_rlx_true_false}).
\begin{compactenum}
\item We have $(\TReadA{v^i_j}{c}{t_2}, \TWriteA{v^i_j}{c}{t_3})\in (\po\cup\rf)^+$ and thus $(\TWriteA{v^i_j}{c}{t_3}, \TReadA{v^i_j}{c}{t_2})\not \in \rf$ due to \PORF{}.
Thus $\TReadA{v^i_j}{c}{t_2}$ is forced to read from the only other available write, i.e., $(\TWriteA{v^i_j}{c}{h_c}, \TReadA{v^i_j}{c}{t_2}) \in \rf$, depicted by \circled{2}. 

\item We now have $(\TWriteA{v^i_j}{c}{h_c}, \TReadA{v^i_k}{c}{t_2})\in (\po_{\VarA{c}}\cup \rf_{\VarA{c}})^+$, 
and due to \RlxWCoh{}, we also have $(\TWriteA{v^i_k}{c}{h_c},\TWriteA{v^i_j}{c}{h_c})\in \mo_{\VarA{c}}$.
In turn, this implies $(\TWriteA{v^i_k}{c}{h_c}, \TReadA{v^i_k}{c}{t_2})\not \in \rf$ due to \RlxRCoh{}.
Thus $\TReadA{v^i_k}{c}{t_2}$ is forced to read from the only other available write, i.e., $(\TWriteA{v^i_k}{c}{t_3}, \TReadA{v^i_k}{c}{t_2}) \in \rf$, 
depicted by \circled{3}. 

\item We now have $(\TReadXOne{v^i_k}{t_3}, \TWriteXOne{v^i_k}{t_2})\in (\po\cup\rf)^+$ and thus $(\TWriteXOne{v^i_k}{t_2}, \TReadXOne{v^i_k}{t_3})\not \in \rf$ due to \PORF{}.
Thus $\TReadXOne{v^i_k}{t_3}$ is forced to read from the only other available write, i.e., $(\TWriteXOne{v^i_k}{t_1}, \TReadXOne{v^i_k}{t_3}) \in \rf$, 
depicted by \circled{4}.
\end{compactenum}

Second, if $(\TWriteXOne{v^i_k}{t_2}, \TReadXOne{v^i_k}{t_3}) \in \rf$ marked 
\circled{1} in \cref{subfig:at_most_one_gadget_rlx_false_true}), 
then $(\TWriteXOne{v^i_j}{t_2}, \TReadXOne{v^i_j}{t_3})\not \in \rf$.
Indeed, we have the following forced sequence of $\rf$ edges (see \cref{subfig:at_most_one_gadget_rlx_false_true}).
\begin{compactenum}
\item We have $(\TReadA{v^i_k}{c}{t_2}, \TWriteA{v^i_k}{c}{t_3})\in (\po\cup\rf)^+$ and thus $(\TWriteA{v^i_k}{c}{t_3}, \TReadA{v^i_k}{c}{t_2})\not \in \rf$ due to \PORF{}.
Thus $\TReadA{v^i_k}{c}{t_2}$ is forced to read from the only other available write, i.e., $(\TWriteA{v^i_k}{c}{h_c}, \TReadA{v^i_k}{c}{t_2}) \in \rf$, marked \circled{2}.
\item Due to \RlxWCoh{}, we have $(\TWriteA{v^i_k}{c}{h_c},\TWriteA{v^i_j}{c}{h_c})\in \mo_{\VarA{c}}$
We thus have $(\TWriteA{v^i_j}{c}{h_c}, \TReadA{v^i_j}{c}{t_2})\not \in \rf$, as this would imply that  $(\TWriteA{v^i_j}{c}{h_c}, \TReadA{v^i_k}{c}{t_2})\in (\po_{\VarA{c}}\cup\rf_{\VarA{c}})$, which would violate \RlxRCoh{}.
Thus $\TReadA{v^i_j}{c}{t_2}$ is forced to read from the only other available write, i.e., $(\TWriteA{v^i_j}{c}{t_3}, \TReadA{v^i_j}{c}{t_2})\in \rf$, 
marked \circled{3}.
\item We now have $(\TReadXOne{v^i_j}{t_3}, \TWriteXOne{v^i_j}{t_2})\in (\po\cup\rf)^+$ and thus $(\TWriteXOne{v^i_j}{t_2}, \TReadXOne{v^i_j}{t_3})\not \in \rf$ due to \PORF{}. 
Thus $\TReadXOne{v^i_j}{t_3}$ is forced to read from the only other available write, i.e., $(\TWriteXOne{v^i_j}{t_1}, \TReadXOne{v^i_j}{t_3}) \in \rf$, 
depicted by \circled{4}.
\qedhere
\end{compactenum}
\end{proof}

The third lemma is based on the at-least-one-true gadget $\AtLeastOneGadget{i}{j}{k}{\ell}$, and it is used to show that for every clause $\Clause_i=(s_j, s_k, s_{\ell})$, at least one of its literals is assigned to true.
Again, it follows by a direct analysis on the accompanying figure, \cref{fig:at_least_one_gadget_rlx}.

\begin{restatable}{lemma}{lemrlxsoundnessatleastone}\label{lem:rlx_soundness_at_least_one}
Let $\ex=(\E, \po,\rf,\mo)$ be a concrete extension of $\expartial$ with $\ex\models \arlxmm$.
For every $i\in[m]$ and $j,k,\ell\in [n]$ such that $s_j$, $s_k$ and $s_{\ell}$ appear in clause $\Clause_i$, we have $(\TWriteXOne{v^i_j}{t_2}, \TReadXOne{v^i_j}{t_3}) \in \rf$ or $(\TWriteXOne{v^i_k}{t_2}, \TReadXOne{v^i_k}{t_3}) \in \rf$ or $(\TWriteXOne{v^i_{\ell}}{t_2}, \TReadXOne{v^i_{\ell}}{t_3}) \in \rf$.
\end{restatable}

\begin{proof}

First, if $(\TWriteXOne{v^i_j}{t_2}, \TReadXOne{v^i_j}{t_3})\not \in \rf$ and $(\TWriteXOne{v^i_k}{t_2}, \TReadXOne{v^i_k}{t_3})\not \in \rf$ (hence both read from $t_1$, marked \circled{1} in \cref{subfig:at_least_one_gadget_rlx_false_false_true}) then $(\TWriteXOne{v^i_{\ell}}{t_2}, \TReadXOne{v^i_{\ell}}{t_3}) \in \rf$.
Indeed, we have the following forced sequence of $\rf$ edges (see \cref{subfig:at_least_one_gadget_rlx_false_false_true}, \circled{1}-\circled{4}).
\begin{compactenum}
\item We have $(\TReadB{v^i_j}{t_1}, \TWriteB{v^i_j}{t_3})\in (\po\cup\rf)^+$ and thus $(\TWriteB{v^i_j}{t_3}, \TReadB{v^i_j}{t_1})\not \in \rf$ due to \PORF{}.
Thus $\TReadB{v^i_j}{t_1}$ is forced to read from the only other available write, i.e., $(\TWriteB{v^i_j}{p}, \TReadB{v^i_j}{t_1}) \in \rf$.
Similarly, we have $(\TReadB{v^i_k}{t_1}, \TWriteB{v^i_k}{t_3})\in (\po\cup\rf)^+$ and thus $(\TWriteB{v^i_k}{t_3}, \TReadB{v^i_k}{t_1})\not \in \rf$ due to \PORF{}.
Thus $\TReadB{v^i_k}{t_1}$ is forced to read from the only other available write, i.e., $(\TWriteB{v^i_k}{q}, \TReadB{v^i_k}{t_1}) \in \rf$. 
These two edges are depicted \circled{2}.
\item We now have $(\TWriteB{v^i_j}{p}, \TReadB{v^i_{\ell}}{t_1})\in (\po_{\VarB} \cup \rf_{\VarB})^+$, 
and due to \RlxWCoh{}, we also have 
$(\TWriteB{v^i_{\ell}}{p},\TWriteB{v^i_j}{p})\in \mo_{\VarB}$.
In turn, this implies $(\TWriteB{v^i_{\ell}}{p}, \TReadB{v^i_{\ell}}{t_1})\not \in \rf$ due to \RlxRCoh{}.
Similarly, we now have $(\TWriteB{v^i_j}{q}, \TReadB{v^i_{\ell}}{t_1})\in (\po_{\VarB}\cup \rf_{\VarB})^+$, 
and due to \RlxWCoh{}, we also have 
$(\TWriteB{v^i_{\ell}}{q},\TWriteB{v^i_j}{q})\in \mo_{\VarB}$.
In turn, this implies $(\TWriteB{v^i_{\ell}}{q}, \TReadB{v^i_{\ell}}{t_1})\not \in \rf$ due to \RlxRCoh{}.
Thus $\TReadB{v^i_{\ell}}{t_1}$ is forced to read from the only other available write, i.e., $(\TWriteB{v^i_{\ell}}{t_3}, \TReadB{v^i_{\ell}}{t_1}) \in \rf$, depicted \circled{3}. 
\item We now have $(\TReadXOne{v^i_{\ell}}{t_3}, \TWriteXOne{v^i_{\ell}}{t_1})\in (\po\cup\rf)^+$ and thus $(\TWriteXOne{v^i_{\ell}}{t_1}, \TReadXOne{v^i_{\ell}}{t_3})\not \in \rf$ due to \PORF{}. 
Thus $\TReadXOne{v^i_{\ell}}{t_3}$ is forced to read from the only other available write, i.e., $(\TWriteXOne{v^i_{\ell}}{t_2}, \TReadXOne{v^i_{\ell}}{t_3}) \in \rf$, 
depicted by \circled{4}.
\end{compactenum}
A similar analysis establishes that if $(\TWriteXOne{v^i_j}{t_2}, \TReadXOne{v^i_j}{t_3})\not \in \rf$ and $(\TWriteXOne{v^i_{\ell}}{t_2}, \TReadXOne{v^i_{\ell}}{t_3})\not \in \rf$ then $(\TWriteXOne{v^i_k}{t_2}, \TReadXOne{v^i_k}{t_3}) \in \rf$ (see \cref{subfig:at_least_one_gadget_rlx_false_true_false}), as well as that if $(\TWriteXOne{v^i_k}{t_2}, \TReadXOne{v^i_k}{t_3})\not \in \rf$ and $(\TWriteXOne{v^i_{\ell}}{t_2}, \TReadXOne{v^i_{\ell}}{t_3})\not \in \rf$ then $(\TWriteXOne{v^i_j}{t_2}, \TReadXOne{v^i_j}{t_3}) \in \rf$ (see \cref{subfig:at_least_one_gadget_rlx_true_false_false}).
\end{proof}

\cref{lem:rlx_soundness_copy} states that our truth assignment for $\varphi$ is valid, while \cref{lem:rlx_soundness_at_most_one} and \cref{lem:rlx_soundness_at_least_one} guarantee that in every clause, exactly one literal is set to true.
Hence we have the following corollary.

\begin{restatable}{corollary}{corrlxlowerthreadssoundness}\label{cor:rlx_lower_threads_soundness}
If  $\expartial\models \arlxmm$ then $\varphi$ is satisfiable.
\end{restatable}

%% file: lower_threads_locations_rlx_completeness.tex

\subsection{Completeness}
\label{SEC:LOWER_THREADS_LOCATIONS_RLX_COMPLETENESS}

We now turn our attention to the completeness property, i.e., if $\varphi$ is satisfiable then $\expartial$ is consistent in the $\arlxmm$ model.
To this end, we construct a reads-from relation $\rf$ and a partial modification order $\mopartial$ as follows.
\begin{compactenum}
\item For each gadget, we insert $\rf$-edges according to 
\cref{fig:copy_gadget_rlx,fig:copy_gadget_down_rlx,fig:at_most_one_gadget_rlx,fig:at_least_one_gadget_rlx}
and the truth assignments on literals $s_j$, $s_k$, $s_{\ell}$ involved in that gadget.
In particular, this implies that, for each $i\in[m]$ and $j\in [n]$, we have $(\TWriteXOne{v^i_j}{t_1}, \TReadXOne{v^i_j}{t_3})\in \rf$ if $s_j=\bfalse$ and $(\TWriteXOne{v^i_j}{t_2}, \TReadXOne{v^i_j}{t_3})\in \rf$ if $s_j=\btrue$.
Observe that $\rf$ is fully specified, i.e., every read has been assigned a write.
\item For every two conflicting writes $\wt_1, \wt_2$ with
(i)~$\tid(\wt_1)\neq \tid(\wt_2)$, and
(ii)~there exist two reads $\rd_1$, $\rd_2$ with $(\rd_1, \rd_2)\in \po$ such that $(\wt_i, \rd_i)\in \rf$ for each $i\in[2]$,
we have $(\wt_1, \wt_2)\in\mopartial$.
\end{compactenum}
We call a triplet $(\wt, \rd, \wt')$ on location $x$ \emph{safe} in $\arlxmm$ if either $(\wt', \rd)\not \in (\po_x \cup \rf_x)^+$ or $(\wt', \wt)\in (\po_x\cup \rf_x \cup \mopartial_x)^+$.
To prove the consistency of $\expartial$ it suffices to prove the following:
\begin{compactenum}
\item $(\po\cup\rf)$ is acyclic, and
\item $\mopartial$ is \emph{minimally coherent} for $\arlxmm$, namely, that
(i)~$(\po_x\cup\rf_x\cup \mopartial_x)$ is acyclic for every location $x$, and
(ii)~every triplet is safe.
\end{compactenum} 
Indeed, minimal-coherence guarantees that, for each location $x$, there exists a total extension $\mo_x$ of $\mopartial_x$ that satisfies \RlxWCoh{} and \RlxRCoh{}~\cite{Tunc2023}.

To simplify our analysis, we ignore the writes in threads $\{h_{c}\}_{c\in[3]}$, $p$ and $q$ that are not read in $\rf$ (crossed-out in \cref{fig:at_most_one_gadget_rlx} and \cref{fig:at_least_one_gadget_rlx}).
This allows us to make our formal statements more uniform, while this omission does not affect the correctness of the analysis.
Indeed, let
$
\InactiveWrites= \setpred{\wt\in \W}{\tid(\wt)\in\set{ h_1, h_2, h_3, p, q} \text{ and } \not\exists \rd\in\R.~(\wt, \rd)\in \rf}
$.
The following lemma is straightforward.

\begin{restatable}{lemma}{lemrlxcompletenessignoreevents}\label{lem:rlx_completeness_inactive_writes}
The following statements hold.
\begin{compactenum}
\item If there is a $(\po \cup \rf)$-cycle then there is such a cycle without any write in $\InactiveWrites$.
\item If there is a $(\po_x \cup \rf_x \cup \mopartial_x)$-cycle for some memory location $x$ then there is such a cycle  without any write in $\InactiveWrites$.
\item If there is an unsafe triplet, then there is such a triplet $(\wt, \rd, \wt')$ where $\wt'\not \in \InactiveWrites$.
\end{compactenum}
\end{restatable}

\begin{proof}
We prove each item separately.
\begin{compactenum}
\item 
We first observe that for any $\wt \in \InactiveWrites$, there is no outgoing $\rf$ or even $\mopartial$ edge.
Hence any cycle containing a write $\wt\in \InactiveWrites$ must contain a sequence of edges $\event_1\LTo{\po}{\wt}\LTo{\po}\event_2$, which can be replaced by the single edge $\event_1\LTo{\po}\event_2$.
\item Proof same as above.
\item Consider a memory location $x$ and an unsafe triplet $(\wt, \rd, \wt')$ on $x$.
This means that there is a $(\po_x \cup \rf_x)$-path $P\colon \wt' \LPath{\po_x \cup \rf_x} \rd$.
Since the threads $\{h_{c}\}_{c\in [3]}$, $p$ and $q$ contain only same-location writes, 
$P$ must be of the form $P\colon \wt' \LPath{\po_x} \wt''\LTo{\rf_x}\rd''\LPath{\po_x \cup \rf_x} \rd$ ,where, $\wt''$ and $\rd''$ conflict with $\wt'$. 
Note that $\rd'' \neq \rd$, otherwise $\wt''=\wt$ since $(\wt, \rd) \in \rf$, implying that $(\wt, \rd, \wt')$ is safe.
Since $(\wt, \rd, \wt')$ is unsafe, we have $(\wt', \wt)\not \in (\po_x \cup \rf_x \cup\mopartial_x)^+$ and hence, $(\wt'', \wt)\not \in (\po_x \cup \rf_x \cup\mopartial_x)^+$. 
Thus $(\wt, \rd, \wt'')$ is also unsafe, while clearly $\wt''\not\in \InactiveWrites$.
\qedhere
\end{compactenum}
\end{proof}

The completeness can now be established in three lemmas, one for each of the properties (1), (2i), and (2ii) above
(i.e., $(\po\cup\rf)$-acyclicity and minimal coherence).
Before we proceed, we will use the following notation to make our analysis easier.

\Paragraph{Notation.}
Given an event $\event$, we say that \emph{$\event$ appears in phase $i$}, written $\Phase(\event)=i$, if it writes/reads a value superscripted by $i$ (i.e., of the form $v^i$ or $\ov{v}^i$).
Similarly, we say \emph{that $\event$ appears in step $j$}, written $\Step(\event)=j$, if it writes/reads a value subscripted by $j$ (i.e., of the form $v_j$ or $\ov{v}_j$).
We define a quasi order $\OrderedBefore$ on the event set $\E$ with $\event_1\OrderedBefore\event_2$ if either
(i)~$\Phase(\event_1)<\Phase(\event_2)$ or
(ii)~$\Phase(\event_1)=\Phase(\event_2)$ and $\Step(\event_1)\leq\Step(\event_2)$.
We also write $\event_1\StrictOrderedBefore \event_2$ to denote that $\event_1\OrderedBefore\event_2$ and $\event_2\not\OrderedBefore\event_1$.
Now consider an arbitrary path $P=\event_1,\dots, \event_k$.
We say that $P$ \emph{crosses} a thread $t$ if $\tid(\event_{\ell})$ for some
for some $\ell\in[k]$.
We call $P$ \emph{monotonic} if it linearizes $\OrderedBefore$, i.e., $\event_{\ell}\OrderedBefore\event_{\ell+1}$ for all $\ell\in[k-1]$.

We first establish the safety of each triplet.

\begin{restatable}{lemma}{lemrlxcompletenesssafetriplets}\label{lem:rlx_completeness_safe_triplets}
Every triplet $(\wt, \rd, \wt')$ is safe.
\end{restatable}

\begin{proof}
First, observe that the following hold by construction.
\begin{compactenum}
\item\label{item:rlx_completeness_safe_triplets_no_read_written} There is no memory location that is both written and read by the same thread.
\item\label{item:rlx_completeness_safe_triplets_no_crossing} For every two conflicting writes $\wt_1$, $\wt_2$, if there exist two reads $\rd_1$, $\rd_ 2$ such that $(\wt_i,\rd_i)\in \rf$ for each $i\in[2]$ and $(\rd_1, \rd_2) \in \po$, then $(\wt_2, \wt_1)\not \in \po$.
\end{compactenum}
Now, let $x$ be the location accessed by a triplet $(\wt, \rd, \wt')$, and assume there exists a $(\po_x\cup \rf_x)$-path $P\colon \wt'\LPath{\po_x\cup \rf_x}\rd$.
Due to \cref{item:rlx_completeness_safe_triplets_no_read_written}, 
$P$ must leave a thread and enter another at most once.
Thus, $P$ must be of the form $P\colon \wt' \LPath{\po^?_x} \wt'' \LTo{\rf_x} \rd''\LPath{\po^?_x} \rd$.
If $(\wt'', \wt)\in \po^?_x$, we are done.
Otherwise, due to \cref{item:rlx_completeness_safe_triplets_no_crossing}, we can't have that $(\wt, \wt'')\in \po_x$.
Hence, $\wt$ and $\wt''$ are in different threads, and by the construction of $\mopartial_x$, we  have $(\wt'', \wt)\in \mopartial_x$, implying that $(\wt', \wt)\in (\po_x\cup\mopartial_x)^+$.
Hence the triplet is safe, as desired.
\end{proof}

Next, we establish the second ingredient of completeness, i.e., the acyclicity of $(\po\cup\rf)$.

\begin{restatable}{lemma}{lemrlxcompletenessnocycleporf}\label{lem:rlx_completeness_no_cycle_po_rf}
$(\po\cup\rf)$ is acyclic.
\end{restatable}

\begin{proof}
First, note that there is no $(\po\cup\rf)$-cycle crossing any of $g_2$ and $g_5$,
as these threads only contain writes, and hence there is no way to enter them by an $\rf$-edge.
Moreover, by construction, any $(\po\cup\rf)$-path $P$ that starts from a thread other than $g_2$ and $g_5$ is necessarily monotonic. 

Indeed, all $\rf$-edges connect events of the same phase and step, while the only $\po$-edges $\event_1\LTo{\po}\event_2$ with $\Phase(\event_1)<\Phase(\event_2)$ appear in $g_2$ and $g_5$.
Finally, the only $\po$-edges $\event_1\LTo{\po}\event_2$ with $\Phase(\event_1)=\Phase(\event_2)$ and $\Step(\event_1)> \Step(\event_2)$ occur in threads $h_1$, $h_2$, $h_3$, $p$ and $q$ (see \cref{fig:at_most_one_gadget_rlx} and \cref{fig:at_least_one_gadget_rlx}).
However, in light of \cref{lem:wra_ra_sra_completeness_inactive_writes}, these edges can be ignored in our analysis.

Thus, any potential $(\po\cup\rf)$-cycle has to traverse events of the same phase and step.
A straightforward analysis of each instantiation of each gadget 
(see \cref{subfig:copy_gadget_rlxFalse,subfig:copy_gadget_rlxTrue},
\cref{subfig:copy_gadget_down_rlxFalse,subfig:copy_gadget_down_rlxTrue},
\cref{subfig:at_most_one_gadget_rlx_true_false,subfig:at_most_one_gadget_rlx_false_true,subfig:at_most_one_gadget_rlx_false_false}, 
\cref{subfig:at_least_one_gadget_rlx_false_false_true,subfig:at_least_one_gadget_rlx_false_true_false,subfig:at_least_one_gadget_rlx_true_false_false}), and their combinations on common events, establishes that no $(\po\cup\rf)$-cycle exists, as desired.
\end{proof}

Finally, we establish the acyclicity of $(\po_x\cup \rf_x\cup\mopartial_x)$.

\begin{restatable}{lemma}{lemrlxcompletenessnocycleporfmo}\label{lem:rlx_completeness_no_cycle_po_rf_mo}
$(\po_x\cup \rf_x\cup\mopartial_x)$ is acyclic for all locations $x$.
\end{restatable}

\begin{proof}
Observe that no memory location is both read and written by the same thread in any of the gadgets. 
Hence, any $(\po_x\cup \rf_x\cup\mopartial_x)$-cycle $C$ following an $\rf$-edge $\wt \LTo{\rf} \rd$ enters a thread that it cannot exit.
This means that $C$ can only contain events of the same thread, which is forbidden by \cref{lem:rlx_completeness_no_cycle_po_rf}.
Thus we only need to reason about the absence of $(\po_x\cup\mopartial_x)$ cycles.

Observe that for every memory location $x$ except $\VarZOne$ and $\VarZTwo$, every $(\po_x\cup\mopartial_x)$-path is monotonic.
Hence any potential $(\po_x\cup\mopartial_x)$-cycle must traverse paths of the same phase and step, and thus be contained in one of the gadgets.
The absence of such cycles can be easily established by looking at the instantiations of these gadgets (\cref{fig:copy_gadget_rlx,fig:copy_gadget_down_rlx,fig:at_most_one_gadget_rlx,fig:at_least_one_gadget_rlx}).

On the other hand, consider the case  $x=\VarZOne$ (the analysis is similar for $x=\VarZTwo$).
Let $C$ be a shortest $(\po_{\VarZOne}\cup\mopartial_{\VarZOne})$-cycle.
If $C$ contains only events of the same phase and step, it can be dismissed again by looking at the gadget in \cref{fig:copy_gadget_down_rlx}.
Otherwise, $C$ is non-monotonic and must thus traverse the non-monotonic edge $\TWriteZOne{u^{i+1}_j}{g_2}\LTo{\po_{\VarZOne}}\TWriteZOne{\ov{u}^{i}_j}{g_2}$.
From there it can continue to to threads $g_1$ and thread $t_5$, following events $\event$ with $\TWriteZOne{\ov{u}^{i}_j}{g_2}\OrderedBefore \event$.
It can be easily verified on \cref{fig:copy_gadget_down_rlx} that the next time $C$ crosses $g_2$, it is on an event $\event'$ with $(\TWriteZOne{u^{i+1}_j}{g_2}, \event')\in \po_{\VarZOne}$, which contradicts the fact that $C$ is a shortest cycle.
The desired result follows.
\end{proof}

\cref{lem:rlx_completeness_safe_triplets}, \cref{lem:rlx_completeness_no_cycle_po_rf} and \cref{lem:rlx_completeness_no_cycle_po_rf_mo} imply the completeness of the reduction.

\begin{corollary}\label{cor:rlx_lower_threads_completeness}
If $\varphi$ is satisfiable then $\expartial\models \arlxmm$.
\end{corollary}

%% file: lower_threads_locations_wra_ra_sra.tex

\section{Hardness for WRA, RA and SRA}\label{SEC:LOWER_THREADS_LOCATIONS_WRA_RA_SRA}

In this section we prove \cref{thm:lower_ra} for the range of models $\sramm \mmorder \MemModel\mmorder \wramm$,
which also implies hardness specifically for $\sramm$, $\ramm$, and $\wramm$.
Similarly to \cref{SEC:LOWER_THREADS_LOCATIONS_RLX}, our reduction uses unbounded values.
Again, we will make the final step towards \cref{thm:lower_ra} in \cref{SEC:BOUNDED_VALUES}, which consists of a simple modification of the technique presented here.

\input{figures/copy_gadget_wra_ra_sra}
\input{figures/copy_gadget_down_wra_ra_sra}
\input{figures/at_most_gadget_wra_ra_sra}
\input{figures/at_least_gadget_wra_ra_sra}

\subsection{Reduction}
Given a monotone formula $\varphi=\{\Clause_i\}_{i\in[m]}$ over $n$ variables $\{s_j\}_{j\in[n]}$, we construct an abstract execution $\expartial=(\E, \po)$ with $\NumThreads=23$ threads and accessing $\NumLocations=26$ memory locations such that $\varphi$ is satisfiable iff $\expartial\models \MemModel$, for any memory model $\MemModel$ such that $\sramm \mmorder \MemModel\mmorder \wramm$.
$\expartial$ follows the general scheme of \cref{fig:reduction_scheme}.
Again, in phase $i$ and step $j$ we have a read event $\TReadXOne{v^i_j}{t_3}$ that can read from either of two writes $\TWriteXOne{v^i_j}{t_1}$ (corresponding to $s_j=\bfalse$) or $\TWriteXOne{v^i_j}{t_2}$ (corresponding to $s_j=\btrue$).
We use the same types of gadgets as in the previous section to force the desirable interaction patterns between threads. 
However, the contents of each gadget are different to account for the different memory models.
In particular, these gadgets now rely on the \WRCoh{} axiom to couple the readers in threads $t_3$ and $t_6$ and force the 1-in-3 property in each clause $\Clause_i$, as opposed to the \PORF{} and \RlxRCoh{} axioms used in \cref{SEC:LOWER_THREADS_LOCATIONS_RLX} for $\arlxmm$.

\Paragraph{The copy gadget $\CopyGadget{i}{j}$.}
The copy gadget $\CopyGadget{i}{j}$ (\cref{fig:copy_gadget_wra_ra_sra}) guarantees that $\TReadXOne{v^i_j}{t_3}$ reads from thread $t_1$ iff $\TReadXTwo{v^i_j}{t_6}$ reads from thread $t_4$.
Besides $\VarXOne$ and $\VarXTwo$, this gadget also uses locations $\VarYOne$, $\VarYTwo$, $\VarYThree$, $\VarYFour$, $\VarYFive$, $\VarYSix$, $\VarYSeven$ and $\VarYEight$.
Moreover, the gadget also contains two $\hb$-edges out of $\TReadXTwo{v^i_j}{t_6}$.
Though not shown explicitly, these $\hb$-edges can be easily enforced by an $\rf$-edge $\wt\LTo{\rf}\rd$, where $\wt$ and $\rd$ access a new memory location.
These events are independent to our analysis later, and will thus be ignored, i.e., we will only be considering the $\hb$-edges as they appear in the gadget.

\Paragraph{The copy-down gadget $\CopyGadgetDown{i}{j}$.}
The copy-down gadget $\CopyGadgetDown{i}{j}$ (\cref{fig:copy_gadget_down_wra_ra_sra}), as before, has a similar structure to the copy gadget $\CopyGadget{i}{j}$ and guarantees that $\TReadXOne{v^{i+1}_j}{t_3}$ reads from thread $t_1$ iff $\TReadXTwo{v^i_j}{t_6}$ reads from thread $t_4$.
Together, the two copy gadgets $\CopyGadget{i}{j}$ and $\CopyGadgetDown{i}{j}$
ensure that $\TReadXOne{v^i_j}{t_3}$ reads from thread $t_1$ iff $\TReadXOne{v^{i+1}_j}{t_3}$ reads from thread $t_1$, guaranteeing a valid truth assignment on $s_j$.
Besides $\VarXOne$ and $\VarXTwo$, this gadget also uses locations $\VarZOne$, $\VarZTwo$, $\VarZThree$, $\VarZFour$, $\VarZFive$, $\VarZSix$, $\VarZSeven$ and $\VarZEight$.
Moreover, the gadget also contains two $\hb$-edges out of $\TReadXTwo{v^i_j}{t_6}$, which, as argued above, will be ignored in our analysis.

\Paragraph{The at-most-one-true gadgets $\AtMostOneGadget{c}{i}{j}{k}$.}
For each $c\in[3]$, the at-most-one-true gadget (\cref{fig:at_most_one_gadget_wra_ra_sra}) guarantees that for every clause $\Clause_i$, the $c$-th pair of literals $(s_j, s_k)$ in $\Clause_i$ is such that 
at most one of $s_j$ and $s_k$ is set to true.
For each $c\in[3]$, the corresponding gadget contains one additional memory location $\VarA{c}$.

\Paragraph{The at-least-one-true gadget $\AtLeastOneGadget{i}{j}{k}{\ell}$.}
The at-least-one-true gadget (\cref{fig:at_least_one_gadget_wra_ra_sra}) guarantees that for every clause $\Clause_i$ with three literals $s_j, s_k, s_{\ell}$, at least one of them is set to true.
It contains one additional memory location $\VarB$.

\Paragraph{Putting the gadgets together.}
We obtain the abstract execution $\expartial$ by appropriately connecting all gadgets and specifying the interleaving of events in common threads.

First, we serially connect all gadgets in their common threads by $\po$.
In particular, $\CopyGadget{i_1}{j_1}$ appears before $\CopyGadget{i_2}{j_2}$ if $i_1<i_2$ or $i_1=i_2$ and $j_1<j_2$;
$\CopyGadgetDown{i_1}{j_1}$ appears before $\CopyGadgetDown{i_2}{j_2}$ if $i_1<i_2$ or $i_1=i_2$ and $j_1<j_2$;
$\AtMostOneGadget{c}{i_1}{j_1}{k_1}$ appears before $\AtMostOneGadget{c}{i_2}{j_2}{k_2}$ if $i_1<i_2$, and finally
$\AtLeastOneGadget{i_1}{j_1}{k_1}{\ell_1}$ appears before $\AtLeastOneGadget{i_2}{j_2}{k_2}{\ell_2}$ if $i_1<i_2$.
Second, observe that the threads $t_{i}$, for $i\in[6]$, appear in multiple gadgets, thus we have to specify how to interleave their corresponding events. 
We do so by first fixing a total order on memory locations
$
\sigma=\VarYOne,\VarYTwo, \VarYThree, \VarYFour, \VarZOne, \VarZTwo, \VarZThree, \VarZFour, \VarA{1}, \VarA{2}, \VarA{3}, \VarB
$.
\begin{compactenum}
\item\label{item:wra_ra_sra_interleaving1} The read events succeeding every $\TReadXOne{v^i_j}{t_3}$ in each gadget (i.e., those on locations $\{\VarYEll, \VarZEll\}_{\ell\in [4]}$,  $\{\VarA{c}\}_{c\in [3]}$ and $\VarB$) are $\po$-ordered according to $\sigma$
(note that reads on some locations, such as $\VarA{c}, \VarB$, might not appear at all after $\TReadXOne{v^i_j}{t_3}$, e.g., if clause $\Clause_i$ does not contain variable $s_j$).
Moreover, all these reads appear before any read on $\VarXOne$ that is a $\po$-successor of $\TReadXOne{v^i_j}{t_3}$ (in particular, reads on $\VarXOne$ with phase $\geq i$ or reads with phase $i$ and step $\geq j$).
\item\label{item:wra_ra_sra_interleaving2} 
Similarly, the write events preceding every $\TWriteXOne{v^i_j}{t_1}$ in each gadget are $\po$-ordered according to $\sigma$.
Moreover, all these writes appear after any write on $\VarXOne$ that is a $\po$-predecessor of $\TWriteXOne{v^i_j}{t_1}$.
Likewise for $\TWriteXOne{v^i_j}{t_2}$, $\TWriteXTwo{v^i_j}{t_4}$ and $\TWriteXTwo{v^i_j}{t_5}$.
\end{compactenum}
Finally, observe that we have indeed used $\NumThreads=23$ threads and $\NumLocations=26$  memory locations.
For the latter, we also count one extra memory location for each $\hb$ edge out of thread $t_6$, thus $4$ in total.

\subsection{Soundness}
We start with the soundness of the reduction, in particular, if $\expartial\models \wramm$ (and thus also $\expartial\models \MemModel$)  then $\varphi$ is satisfiable.
We achieve this by proving a sequence of intermediate lemmas similar to \cref{SEC:LOWER_THREADS_LOCATIONS_RLX}, however, each lemma now requires reasoning about the semantics of $\wramm$.
Recall that we assign $s_j=\btrue$ if $(\TWriteXOne{v^i_j}{t_2}, \TReadXOne{v^i_j}{t_3})\in \rf$ for all $i\in [m]$, and $s_j=\bfalse$ if $(\TWriteXOne{v^i_j}{t_1}, \TReadXOne{v^i_j}{t_3})\in \rf$ for all $i\in[m]$.
The first lemma is based on the copy gadgets $\CopyGadget{i}{j}$ and $\CopyGadgetDown{i}{j}$ and states that each phase of $\ex$ makes consistent choices for the writer of $\TReadXOne{v^i_j}{t_3}$ (i.e., whether it reads from $t_1$ or $t_2$), which makes the above assignment well-defined.
It follows from the high-level description of these two gadgets and the accompanying \cref{fig:copy_gadget_wra_ra_sra} and \cref{fig:copy_gadget_down_wra_ra_sra}.

\begin{restatable}{lemma}{lemmawrarasrasoundnesscopy}\label{lem:wra_ra_sra_soundness_copy}
Let $\ex=(\E, \po,\rf,\mo)$ be a concrete extension of $\expartial$ that satisfies \WRCoh{}.
For all $i_1, i_2\in[m]$, $j\in[n]$, we have $(\TWriteXOne{v^{i_1}_j}{t_1}, \TReadXOne{v^{i_1}_j}{t_3}) \in \rf$ iff $(\TWriteXOne{v^{i_2}_j}{t_1}, \TReadXOne{v^{i_2}_j}{t_3}) \in \rf$.
\end{restatable}

The next lemma is based on the at-most-one-true gadgets $\AtMostOneGadget{c}{i}{j}{k}$ and states that for every clause $\Clause_i$ and for each of the $c\in[3]$ pair of variables $(s_j, s_k)$ in $\Clause_i$, at most one of them is assigned to true.
Again, it follows by a direct analysis on the accompanying \cref{fig:at_most_one_gadget_wra_ra_sra}.

\begin{restatable}{lemma}{lemwrarasrasoundnessatmostone}\label{lem:wra_ra_sra_soundness_at_most_one}
Let $\ex=(\E, \po,\rf,\mo)$ be a concrete extension of $\expartial$ that satisfies \WRCoh{}.
For every $i\in[m]$ and $j,k\in [n]$
such that $s_j$ and $s_k$ appear in clause $\Clause_i$, we have $(\TWriteXOne{v^i_j}{t_2}, \TReadXOne{v^i_j}{t_3})\not \in \rf$ or $(\TWriteXOne{v^i_k}{t_2}, \TReadXOne{v^i_k}{t_3})\not \in \rf$.
\end{restatable}

The third lemma is based on the at-least-one-true gadget $\AtLeastOneGadget{i}{j}{k}{\ell}$, and it is used to show that for every clause $\Clause_i=(s_j, s_k, s_{\ell})$, at least one of its variables is assigned to true.
Again, it follows by a direct analysis on the accompanying \cref{fig:at_least_one_gadget_wra_ra_sra}.

\begin{restatable}{lemma}{lemwrarasrasoundnessatleastone}\label{lem:wra_ra_sra_soundness_at_least_one}
Let $\ex=(\E, \po,\rf,\mo)$ be a concrete extension of $\expartial$ that satisfies \WRCoh{}.
For every $i\in[m]$ and $j,k,\ell\in [n]$ such that $s_j$, $s_k$ and $s_{\ell}$ appear in clause $\Clause_i$, we have $(\TWriteXOne{v^i_j}{t_2}, \TReadXOne{v^i_j}{t_3}) \in \rf$ or $(\TWriteXOne{v^i_k}{t_2}, \TReadXOne{v^i_k}{t_3}) \in \rf$ or $(\TWriteXOne{v^i_{\ell}}{t_2}, \TReadXOne{v^i_{\ell}}{t_3}) \in \rf$.
\end{restatable}

\cref{lem:wra_ra_sra_soundness_copy} states that our truth assignment for $\varphi$ is valid, while \cref{lem:wra_ra_sra_soundness_at_most_one} and \cref{lem:wra_ra_sra_soundness_at_least_one} guarantee that in every clause, exactly one literal is set to true.
Hence we have the following corollary.

\begin{restatable}{corollary}{corwrarasralowerthreadslocationssoundness}\label{cor:wra_ra_sra_lower_threads_locations_soundness}
If  $\expartial\models \wramm$ then $\varphi$ is satisfiable.
\end{restatable}

\subsection{Completeness}

We now turn our attention to the completeness property, i.e., if $\varphi$ is satisfiable then $\expartial\models \sramm$ (and thus also $\expartial\models \MemModel$).
We use the notions of phase, step, and monotonicity as in \cref{SEC:LOWER_THREADS_LOCATIONS_RLX}.
We construct a reads-from relation $\rf$ and a partial modification order $\mopartial$ as follows.
\begin{compactenum}
\item  \label{item:wra_ra_sra_intra_edges}
For each gadget, we insert $\rf$ and $\mo$ edges according to \cref{fig:copy_gadget_wra_ra_sra,fig:copy_gadget_down_wra_ra_sra,fig:at_most_one_gadget_wra_ra_sra,fig:at_least_one_gadget_wra_ra_sra} and the truth assignments on variables $s_j$, $s_k$, $s_{\ell}$ involved in that gadget.
In particular, this implies that, for each $i\in[m]$ and $j\in [n]$, we have $(\TWriteXOne{v^i_j}{t_1}, \TReadXOne{v^i_j}{t_3})\in \rf$ if $s_j=\bfalse$ and $(\TWriteXOne{v^i_j}{t_2}, \TReadXOne{v^i_j}{t_3})\in \rf$ if $s_j=\btrue$.
Observe that $\rf$ is fully specified, i.e., every read has a write to read from.
Note that for every pair $(\wt,\rd)\in \rf$ we have $\Phase(\wt)=\Phase(\rd)$ and $\Step(\wt)=\Step(\rd)$, thus $\rd\OrderedBefore\wt$.
\item \label{item:wra_ra_sra_inter_edges}
For every two conflicting writes $\wt_1, \wt_2$ such that 
(i)~$\wt_1\StrictOrderedBefore\wt_2$, and
(ii)~$(\wt_1,\wt_2)\not \in \hb$,
we have $(\wt_1, \wt_2)\in \mopartial$.
Thus, for any two conflicting writes $\wt_1, \wt_2$ with  $\wt_1 \StrictOrderedBefore\wt_2$ we have $(\wt_1, \wt_2) \in (\hb \cup \mopartial)^+$.
\end{compactenum}

We call a triplet $(\wt, \rd, \wt')$ \emph{safe} if either $(\wt', \rd)\not \in \hb$ or $(\wt', \wt)\in (\hb\cup \mopartial)^+$.
To prove the consistency of $\expartial$, it suffices to argue that $\mopartial$ is \emph{minimally coherent},
namely, that 
(i)~$(\hb \cup \mopartial)^+$ is acyclic, and
(ii)~every triplet $(\wt, \rd, \wt')$ is safe.
Indeed,  these two conditions guarantee that $\mopartial$ can be linearized to a total $\mo$ such that any extension $\ex$ of $\expartial$ satisfies $\ex=(\E, \po, \rf, \mo)\models \sramm$~\cite{Tunc2023},
which implies that also $\ex\models \MemModel$.

In order to simplify our analysis, we again ignore the writes in threads $\{h_{c}\}_{c\in[3]}$, $p$, and $q$ that are not read in $\rf$ (crossed-out in 
\cref{fig:at_most_one_gadget_wra_ra_sra} and \cref{fig:at_least_one_gadget_wra_ra_sra}).
This allows us to make our formal statements more uniform, while it does not affect the correctness of the analysis.
Indeed, let 
$
\InactiveWrites= \{\wt\in \W\colon \tid(\wt)\in\{ h_1, h_2, h_3, p, q\} \text{ and } \not\exists \rd\in\R.~(\wt, \rd)\in \rf\}
$.
The following lemma is straightforward.

\begin{restatable}{lemma}{lemwrarasracompletenessignoreevents}\label{lem:wra_ra_sra_completeness_inactive_writes}
The following statements hold.
\begin{compactenum}
\item If there is an $(\hb\cup \mopartial)$-cycle then there is such a cycle without any write in $\InactiveWrites$.
\item If there is an unsafe triplet, then there is such a triplet $(\wt, \rd, \wt')$ where $\wt'\not \in \InactiveWrites$.
\end{compactenum}
\end{restatable}

We start with condition (i) of minimal coherence, i.e., we need to show that $(\hb\cup\mopartial)$ is acyclic.
Observe that each individual gadget is free from $(\hb\cup \mopartial)$-cycles, regardless of how we resolve the $\rf$ edges associated with it (see \cref{subfig:copy_gadget_wra_ra_sraFalse,subfig:copy_gadget_wra_ra_sraTrue},
\cref{subfig:copy_gadget_down_wra_ra_sraFalse,subfig:copy_gadget_down_wra_ra_sraTrue},
 \cref{subfig:at_most_one_gadget_wra_ra_sra_true_false,subfig:at_most_one_gadget_wra_ra_sra_false_true,subfig:at_most_one_gadget_wra_ra_sra_false_false}, \cref{subfig:at_least_one_gadget_wra_ra_sra_false_false_true,subfig:at_least_one_gadget_wra_ra_sra_false_true_false,subfig:at_least_one_gadget_wra_ra_sra_true_false_false}).
However, we have to also argue that the interleaving of these gadgets is free from $(\hb\cup\mopartial)$-cycles.

Our first key lemma states that $\hb$ paths between writes are, without loss of generality, monotonic.
This is based on three observations. 
First, due to \cref{lem:wra_ra_sra_completeness_inactive_writes}, we can ignore writes in the threads $h_1$, $h_2$, $h_3$, $p$ and $q$, which contain $\po$-edges that would violate this statement.
Second, all $\rf$-edges connect events of the same phase and step.
Third, the only $\po$-edges that are non-monotonic enter read events (in particular, a read $\TReadZSix{g_1}{v^i_j}$ or a read $\TReadZEight{g_4}{v^i_j}$ in $\CopyGadgetDown{i}{j}$).
Since the only possible continuation of an $\hb$-path out of a read event is to take another $\po$-edge, we can remove the first non-monotonic edge (as $\po$ is transitive) and obtain a new valid $\hb$-path.
Repeating this process results in a monotonic $\hb$-path between the writes.
Formally, we have the following.

\begin{restatable}{lemma}{lemwrarasracompletenessmonotonichbpaths}\label{lem:wra_ra_sra_completeness_monotonic_hb_paths}
For every two writes $\wt_1$, $\wt_2$, if $(\wt_1,\wt_2)\in  \hb$ then there exists a monotonic $\hb$-path $\wt_1\LPath{\hb}\wt_2$. 
\end{restatable}

We can now prove the acyclicity condition of minimal coherence.
Intuitively, any potential $(\hb\cup\mopartial)$-cycle $C$ can be seen as a sequence of write events connected by $\hb$ and $\mopartial$.
By construction,  every edge $\wt_1\LTo{\mopartial}\wt_2$ is monotonic,
while, due to \cref{lem:wra_ra_sra_completeness_monotonic_hb_paths}, every subpath $\wt_1\LPath{\hb}\wt_2$ of $C$ is, without loss of generality, monotonic.
Thus $C$ is monotonic, and since it is a cycle, every event in $C$ has the same phase and step.
The absence of such cycles can then be directly established by inspecting the gadgets in \cref{fig:copy_gadget_wra_ra_sra,fig:copy_gadget_down_wra_ra_sra,fig:at_most_one_gadget_wra_ra_sra,fig:at_least_one_gadget_wra_ra_sra}.

\begin{restatable}{lemma}{lemwrarasracompletenessnocyclehbmo}\label{lem:wra_ra_sra_completeness_no_cycle_hb_mo}
$(\hb\cup\mopartial)$ is acyclic.
\end{restatable}

Next, we turn our attention to the second condition of minimal coherence, i.e., we argue that every triplet is safe.
We first prove a general statement that prohibits $\hb$-paths to a read $\rd$ from writes $\wt'$ that are $\po$-successors to the write $\wt$ that $\rd$ reads from.
This will help us establish the safety of each triplet, and will also prove useful later in \cref{SUBSEC:CAUSAL_MEMORY} when we address Causal Memory.

\begin{restatable}{lemma}{lemwrarasrarfimmediate}\label{lem:wra_ra_sra_rf_immediate}
For every pair $(\wt, \rd)\in \rf$ and write $\wt'$ with $(\wt, \wt')\in \po$, we have $(\wt',\rd)\not \in \hb$.
\end{restatable}

To realize \cref{lem:wra_ra_sra_rf_immediate}, we first argue that any path $P\colon\wt'\LPath{\hb}\rd$ contains events of the same phase and step.
Indeed, as no location is ever written and read by the same thread, $P$ has the general form $P\colon \wt'\LPath{\hb^?}\wt''\LTo{\rf}\rd''\LTo{\po^?}\rd$ for some write $\wt''$ and read $\rd''$.
Due to \cref{lem:wra_ra_sra_completeness_monotonic_hb_paths}, the subpath $\wt'\LPath{\hb^?}\wt''$ is monotonic (wlog), while the last two edges of $P$ are also monotonic ($\rf$-edges are monotonic, while non-monotonic $\po$-edges go from writes to reads).
Hence $P$ is monotonic.
On the other hand,  we have $\wt\OrderedBefore\wt'$ (as $(\wt, \wt')\in \po$), while, by construction, $\rd\OrderedBefore\wt$.
Hence $\rd\OrderedBefore\wt'$, and since $P$ is monotonic, it must contain only events of the same phase and step.
In particular, $P$ must be contained in the gadgets in \cref{fig:copy_gadget_wra_ra_sra,fig:copy_gadget_down_wra_ra_sra,fig:at_most_one_gadget_wra_ra_sra,fig:at_least_one_gadget_wra_ra_sra}.
The absence of such paths $P$ can then be established by a careful inspection of these gadgets.

We can now prove the safety of each triplet $(\wt, \rd,\wt')$.
Intuitively, if $\wt'\StrictOrderedBefore\wt$, then we have $(\wt', \wt)\in \mopartial$ by construction.
On the other hand, if $\wt\StrictOrderedBefore \wt'$, \cref{lem:wra_ra_sra_completeness_monotonic_hb_paths} and \cref{lem:wra_ra_sra_rf_immediate} exclude the existence of $\hb$-paths $\wt'\LPath{\hb}\rd$.
Hence, it again suffices to only consider $\hb$-paths $P\colon \wt'\LPath{\hb}\rd$ that are contained in the same gadget. 
Again, a careful inspection of \cref{fig:copy_gadget_wra_ra_sra,fig:copy_gadget_down_wra_ra_sra,fig:at_most_one_gadget_wra_ra_sra,fig:at_least_one_gadget_wra_ra_sra} and the use of \cref{lem:wra_ra_sra_rf_immediate} show that $(\wt, \rd,\wt')$ is indeed safe.

\begin{restatable}{lemma}{lemwrarasracompletenesssafetriplets}\label{lem:wra_ra_sra_completeness_safe_triplets}
Every triplet $(\wt, \rd, \wt')$ is safe.
\end{restatable}

\cref{lem:wra_ra_sra_completeness_no_cycle_hb_mo} and \cref{lem:wra_ra_sra_completeness_safe_triplets} show that $\mopartial$ is indeed minimally coherent, which implies that $\ex=(\E, \po, \rf, \mo)\models \sramm$.
Thus we have the following corollary.

\begin{corollary}\label{cor:wra_ra_sra_lower_threads_completeness}
If $\varphi$ is satisfiable then $\expartial\models \sramm$.
\end{corollary}

Together, \cref{cor:wra_ra_sra_lower_threads_locations_soundness} and \cref{cor:wra_ra_sra_lower_threads_completeness} establish the correctness of the reduction, i.e., $\varphi$ is satisfiable iff $\expartial\models\MemModel$ for any memory model $\sramm\mmorder \MemModel\mmorder\wramm$.

%% file: figures/copy_gadget_wra_ra_sra.tex
\begin{figure}
\newcommand{\xdisposition}{3.4*\xstep}
\newcommand{\ydisposition}{2*\ystep}
\newcommand{\xtstep}{0.75}
\newcommand{\ytstep}{0.7}
\newcommand{\ybias}{-0.3 }
\newcommand{\xstep}{1.435}
\newcommand{\ystep}{-0.6}
\newcommand{\xtscale}{0.8}
\def \numevents{5}
\def\scale{0.825}
\newcommand{\BaseCopyGadget}{
\begin{scope}[shift={(0*\xdisposition,0*\ydisposition)}]

\foreach \x [evaluate=\x as \i using ({int(\x)})] in {1,...,3}{
\node[] (T\i1) at (\i*\xstep,0.15) {\normalsize $t_{\x}$};
\node[] (T\i2) at (\i*\xstep,\numevents * \ystep) {};
\draw[seqtrace] (T\i1) to (T\i2);
}

\foreach \x [evaluate=\x as \i using ({int(\x+3)})] in {1,...,6}{
\node[] (T\i1) at (\i*\xstep,0.15) {\normalsize $f_{\x}$};
\node[] (T\i2) at (\i*\xstep,\numevents * \ystep) {};
\draw[seqtrace] (T\i1) to (T\i2);
}

\foreach \x [evaluate=\x as \i using ({int(\x+6)})] in {4,...,6}{
\node[] (T\i1) at (\i*\xstep,0.15) {\normalsize $t_{\x}$};
\node[] (T\i2) at (\i*\xstep,\numevents * \ystep) {};
\draw[seqtrace] (T\i1) to (T\i2);
}

\node[event] (T11) at (1*\xstep, 1*\ystep + 0*\ybias) {$\WriteYOne{v^{i}_{j}}$};
\node[event] (T13) at (1*\xstep, 3*\ystep + 0*\ybias) {$\WriteYOne{\ov{v}^{i}_{j}}$};
\node[event, draw=black] (T15) at (1*\xstep, 5*\ystep + 0*\ybias) {$\WriteXOne{v^{i}_{j}}$};

\node[event] (T21) at (2*\xstep, 1*\ystep + 0*\ybias) {$\WriteYTwo{v^{i}_{j}}$};
\node[event] (T23) at (2*\xstep, 3*\ystep + 0*\ybias) {$\WriteYTwo{\ov{v}^{i}_{j}}$};
\node[event, draw=black] (T25) at (2*\xstep, 5*\ystep + 0*\ybias) {$\WriteXOne{v^{i}_{j}}$};

\node[event, draw=black] (T31) at (3*\xstep, 1*\ystep + 0*\ybias) {$\ReadXOne{v^{i}_{j}}$};
\node[event] (T32) at (3*\xstep, 2*\ystep + 0*\ybias) {$\ReadYOne{v^{i}_{j}}$};
\node[event] (T33) at (3*\xstep, 3*\ystep + 0*\ybias) {$\ReadYTwo{v^{i}_{j}}$};
\node[event] (T34) at (3*\xstep, 4*\ystep + 0*\ybias) {$\ReadYThree{v^{i}_{j}}$};
\node[event] (T35) at (3*\xstep, 5*\ystep + 0*\ybias) {$\ReadYFour{v^{i}_{j}}$};

\node[event] (T41) at (4*\xstep, 1*\ystep + 0*\ybias) {$\ReadYFive{v^{i}_{j}}$};
\node[event] (T43) at (4*\xstep, 3*\ystep + 0*\ybias) {$\WriteYTwo{v^{i}_{j}}$};
\node[event] (T45) at (4*\xstep, 5*\ystep + 0*\ybias) {$\ReadYSix{v^{i}_{j}}$};

\node[event] (T51) at (5*\xstep, 1*\ystep + 0*\ybias) {$\WriteYThree{v^{i}_{j}}$};
\node[event] (T53) at (5*\xstep, 3*\ystep + 0*\ybias) {$\WriteYThree{\ov{v}^{i}_{j}}$};
\node[event] (T55) at (5*\xstep, 5*\ystep + 0*\ybias) {$\WriteYFive{v^{i}_{j}}$};

\node[event] (T63) at (6*\xstep, 3*\ystep + 0*\ybias) {$\WriteYFive{v^{i}_{j}}$};

\node[event] (T71) at (7*\xstep, 1*\ystep + 0*\ybias) {$\ReadYSeven{v^{i}_{j}}$};
\node[event] (T73) at (7*\xstep, 3*\ystep + 0*\ybias) {$\WriteYOne{v^{i}_{j}}$};
\node[event] (T75) at (7*\xstep, 5*\ystep + 0*\ybias) {$\ReadYEight{v^{i}_{j}}$};

\node[event] (T81) at (8*\xstep, 1*\ystep + 0*\ybias) {$\WriteYFour{v^{i}_{j}}$};
\node[event] (T83) at (8*\xstep, 3*\ystep + 0*\ybias) {$\WriteYFour{\ov{v}^{i}_{j}}$};
\node[event] (T85) at (8*\xstep, 5*\ystep + 0*\ybias) {$\WriteYSeven{v^{i}_{j}}$};

\node[event] (T93) at (9*\xstep, 3*\ystep + 0*\ybias) {$\WriteYSeven{v^{i}_{j}}$};

\node[event] (T101) at (10*\xstep, 1*\ystep + 0*\ybias) {$\WriteYSix{v^{i}_{j}}$};
\node[event] (T103) at (10*\xstep, 3*\ystep + 0*\ybias) {$\WriteYSix{\ov{v}^{i}_{j}}$};
\node[event, draw=black] (T105) at (10*\xstep, 5*\ystep + 0*\ybias) {$\WriteXTwo{v^{i}_{j}}$};

\node[event] (T111) at (11*\xstep, 1*\ystep + 0*\ybias) {$\WriteYEight{v^{i}_{j}}$};
\node[event] (T113) at (11*\xstep, 3*\ystep + 0*\ybias) {$\WriteYEight{\ov{v}^{i}_{j}}$};
\node[event, draw=black] (T115) at (11*\xstep, 5*\ystep + 0*\ybias) {$\WriteXTwo{v^{i}_{j}}$};

\node[event, draw=black] (T121) at (12*\xstep, 1*\ystep + 0*\ybias) {$\ReadXTwo{v^{i}_{j}}$};

\begin{pgfonlayer}{bg}
\draw[hb, out=-155, in=15, looseness=0.5] (T121) to (T63);
\draw[hb, out=-155, in=15, looseness=0.5] (T121) to node[below, pos=0.1] {$\hb$}(T93);
\draw[rf, out=-170, in=-5, looseness=0.5] (T51) to (T34);
\draw[rf] (T81) to node[above, pos=0.25] {$\rf$} (T35);
\draw[rf] (T101) to (T45);
\draw[rf] (T111) to (T75);
\end{pgfonlayer}

\end{scope}
}

\begin{subfigure}{\textwidth}
\centering
\scalebox{\scale}{%
\begin{tikzpicture}[thick, font=\footnotesize,
pre/.style={<-,shorten >= 2pt, shorten <=2pt, very thick},
post/.style={->,shorten >= 3pt, shorten <=3pt,   thick},
seqtrace/.style={line width=1},
und/.style={very thick, draw=gray},
virt/.style={circle,draw=black!50,fill=black!20, opacity=0}]

\BaseCopyGadget

\end{tikzpicture}
}
\caption{\label{subfig:copy_gadget_wra_ra_sra}
The copy gadget $\CopyGadget{i}{j}$.
}
\end{subfigure}
\\
\begin{subfigure}{\textwidth}
\centering
\scalebox{\scale}{
\begin{tikzpicture}[thick, font=\footnotesize,
pre/.style={<-,shorten >= 2pt, shorten <=2pt, very thick},
post/.style={->,shorten >= 3pt, shorten <=3pt,   thick},
seqtrace/.style={line width=1},
und/.style={very thick, draw=gray},
virt/.style={circle,draw=black!50,fill=black!20, opacity=0}]

\BaseCopyGadget

\begin{pgfonlayer}{bg}
\draw[rf, out=70, in=-165, looseness=0.75] (T15) to node [pos=0.15] (first) {\circled{1}} (T31);
\draw[rf, out=160, in=0, looseness=0.5] (T73) to node [pos=0.85] (sec) {\circled{2}} (T32);
\draw[rf] (T93) -- (T71) node [pos=0.65] (third) {\circled{3}};
\draw[rf, out=70, in=-165, looseness=0.75] (T105) to node [pos=0.15] (four) {\circled{4}} (T121);
\draw[rf] (T55) to (T41);
\draw[rf] (T21) to (T33);

\draw[mo, out=-30, in=210, looseness=0.1] (T13) to (T73);
\end{pgfonlayer}

\end{tikzpicture}
}
\caption{\label{subfig:copy_gadget_wra_ra_sraFalse}
Choosing $\TReadXOne{v^i_j}{t_3}$ to read from $t_1$ forces the sequence of $\rf$ and $\mo$ edges shown.
}
\end{subfigure}
\\
\begin{subfigure}{\textwidth}
\centering
\scalebox{\scale}{%
\begin{tikzpicture}[thick, font=\footnotesize,
pre/.style={<-,shorten >= 2pt, shorten <=2pt, very thick},
post/.style={->,shorten >= 3pt, shorten <=3pt,   thick},
seqtrace/.style={line width=1},
und/.style={very thick, draw=gray},
virt/.style={circle,draw=black!50,fill=black!20, opacity=0}]

\BaseCopyGadget

\begin{pgfonlayer}{bg}
\draw[rf, out=30, in=-155, looseness=0.5] (T25) to node [pos=0.15] (first) {\circled{1}} (T31);
\draw[rf, out=145, in=35, looseness=0.5] (T43) to  node [pos=0.1,above=.1mm] (seco) {\circled{2}} (T33);
\draw[rf] (T63) -- (T41) node [pos=0.3] (thir) {\circled{3}};
\draw[rf, out=30] (T115) -- (T121) node [pos=0.3] (fourt) {\circled{4}};
\draw[rf] (T85) to (T71);
\draw[rf, out=-30, in=170, looseness=0.5] (T11) to (T32);

\draw[mo, out=-20, in=200, looseness=0.3] (T23) to (T43);
\end{pgfonlayer}

\end{tikzpicture}
}
\caption{\label{subfig:copy_gadget_wra_ra_sraTrue}
Choosing $\TReadXOne{v^i_j}{t_3}$ to read from $t_2$ forces the sequence of $\rf$ and $\mo$ edges shown.
}
\end{subfigure}
\caption{\label{fig:copy_gadget_wra_ra_sra}
The copy gadget $\CopyGadget{i}{j}$ (\subref{subfig:copy_gadget_wra_ra_sra}) captures the Boolean assignment to variable $s_j$ in phase $i$.
There are two ways to realize this gadget, by choosing which of the two writes $\WriteXOne{v^i_j}$ the read $\ReadXOne{v^i_j}$ observes.
Choosing the write of $t_1$ (\subref{subfig:copy_gadget_wra_ra_sraFalse}) corresponds to setting $s_j=\bfalse$ and also forces $\ReadXTwo{v^i_j}$ to read from $t_4$. 
Choosing the write of $t_2$ (\subref{subfig:copy_gadget_wra_ra_sraTrue}) corresponds to setting $s_j=\btrue$ and also forces $\ReadXTwo{v^i_j}$ to read from $t_5$.
This $\rf$ coupling is formalized in \cref{lem:wra_ra_sra_soundness_copy}.
The edge numbers specify the order in which $\rf$-edges are inferred.
}
\end{figure}

%% file: figures/copy_gadget_down_wra_ra_sra.tex

\begin{figure}
\newcommand{\xdisposition}{3.4*\xstep}
\newcommand{\ydisposition}{2*\ystep}
\newcommand{\xtstep}{0.75}
\newcommand{\ytstep}{0.7}
\newcommand{\ybias}{-0.3 }
\newcommand{\xstep}{1.435}
\newcommand{\ystep}{-0.6}
\newcommand{\xtscale}{0.8}
\def \numevents{5}
\def\scale{0.825}
\newcommand{\BaseCopyGadget}{
\begin{scope}[shift={(0*\xdisposition,0*\ydisposition)}]

\foreach \x [evaluate=\x as \i using ({int(\x)})] in {1,...,3}{
\node[] (T\i1) at (\i*\xstep,0.15) {\normalsize $t_{\x}$};
\node[] (T\i2) at (\i*\xstep,\numevents * \ystep) {};
\draw[seqtrace] (T\i1) to (T\i2);
}

\foreach \x [evaluate=\x as \i using ({int(\x+3)})] in {1,...,6}{
\node[] (T\i1) at (\i*\xstep,0.15) {\normalsize $g_{\x}$};
\node[] (T\i2) at (\i*\xstep,\numevents * \ystep) {};
\draw[seqtrace] (T\i1) to (T\i2);
}

\foreach \x [evaluate=\x as \i using ({int(\x+6)})] in {4,...,6}{
\node[] (T\i1) at (\i*\xstep,0.15) {\normalsize $t_{\x}$};
\node[] (T\i2) at (\i*\xstep,\numevents * \ystep) {};
\draw[seqtrace] (T\i1) to (T\i2);
}

\node[event] (T11) at (1*\xstep, 1*\ystep + 0*\ybias) {$\WriteZOne{v^{i+1}_{j}}$};
\node[event] (T13) at (1*\xstep, 3*\ystep + 0*\ybias) {$\WriteZOne{\ov{v}^{i+1}_{j}}$};
\node[event, draw=black] (T15) at (1*\xstep, 5*\ystep + 0*\ybias) {$\WriteXOne{v^{i+1}_{j}}$};

\node[event] (T21) at (2*\xstep, 1*\ystep + 0*\ybias) {$\WriteZTwo{v^{i+1}_{j}}$};
\node[event] (T23) at (2*\xstep, 3*\ystep + 0*\ybias) {$\WriteZTwo{\ov{v}^{i+1}_{j}}$};
\node[event, draw=black] (T25) at (2*\xstep, 5*\ystep + 0*\ybias) {$\WriteXOne{v^{i+1}_{j}}$};

\node[event, draw=black] (T31) at (3*\xstep, 1*\ystep + 0*\ybias) {$\ReadXOne{v^{i+1}_{j}}$};
\node[event] (T32) at (3*\xstep, 2*\ystep + 0*\ybias) {$\ReadZOne{v^{i+1}_{j}}$};
\node[event] (T33) at (3*\xstep, 3*\ystep + 0*\ybias) {$\ReadZTwo{v^{i+1}_{j}}$};
\node[event] (T34) at (3*\xstep, 4*\ystep + 0*\ybias) {$\ReadZThree{v^{i+1}_{j}}$};
\node[event] (T35) at (3*\xstep, 5*\ystep + 0*\ybias) {$\ReadZFour{v^{i+1}_{j}}$};

\node[event] (T41) at (4*\xstep, 1*\ystep + 0*\ybias) {$\ReadZFive{v^{i+1}_{j}}$};
\node[event] (T43) at (4*\xstep, 3*\ystep + 0*\ybias) {$\WriteZTwo{v^{i+1}_{j}}$};
\node[event] (T45) at (4*\xstep, 5*\ystep + 0*\ybias) {$\ReadZSix{v^{i}_{j}}$};

\node[event] (T51) at (5*\xstep, 1*\ystep + 0*\ybias) {$\WriteZThree{v^{i+1}_{j}}$};
\node[event] (T53) at (5*\xstep, 3*\ystep + 0*\ybias) {$\WriteZThree{\ov{v}^{i+1}_{j}}$};
\node[event] (T55) at (5*\xstep, 5*\ystep + 0*\ybias) {$\WriteZFive{v^{i+1}_{j}}$};

\node[event] (T63) at (6*\xstep, 3*\ystep + 0*\ybias) {$\WriteZFive{v^{i+1}_{j}}$};

\node[event] (T71) at (7*\xstep, 1*\ystep + 0*\ybias) {$\ReadZSeven{v^{i+1}_{j}}$};
\node[event] (T73) at (7*\xstep, 3*\ystep + 0*\ybias) {$\WriteZOne{v^{i+1}_{j}}$};
\node[event] (T75) at (7*\xstep, 5*\ystep + 0*\ybias) {$\ReadZEight{v^{i}_{j}}$};

\node[event] (T81) at (8*\xstep, 1*\ystep + 0*\ybias) {$\WriteZFour{v^{i+1}_{j}}$};
\node[event] (T83) at (8*\xstep, 3*\ystep + 0*\ybias) {$\WriteZFour{\ov{v}^{i+1}_{j}}$};
\node[event] (T85) at (8*\xstep, 5*\ystep + 0*\ybias) {$\WriteZSeven{v^{i+1}_{j}}$};

\node[event] (T93) at (9*\xstep, 3*\ystep + 0*\ybias) {$\WriteZSeven{v^{i+1}_{j}}$};

\node[event] (T101) at (10*\xstep, 1*\ystep + 0*\ybias) {$\WriteZSix{v^{i}_{j}}$};
\node[event] (T103) at (10*\xstep, 3*\ystep + 0*\ybias) {$\WriteZSix{\ov{v}^{i}_{j}}$};
\node[event, draw=black] (T105) at (10*\xstep, 5*\ystep + 0*\ybias) {$\WriteXTwo{v^{i}_{j}}$};

\node[event] (T111) at (11*\xstep, 1*\ystep + 0*\ybias) {$\WriteZEight{v^{i}_{j}}$};
\node[event] (T113) at (11*\xstep, 3*\ystep + 0*\ybias) {$\WriteZEight{\ov{v}^{i}_{j}}$};
\node[event, draw=black] (T115) at (11*\xstep, 5*\ystep + 0*\ybias) {$\WriteXTwo{v^{i}_{j}}$};

\node[event, draw=black] (T121) at (12*\xstep, 1*\ystep + 0*\ybias) {$\ReadXTwo{v^{i}_{j}}$};

\begin{pgfonlayer}{bg}
\draw[hb, out=-155, in=15, looseness=0.5] (T121) to (T63);
\draw[hb, out=-155, in=15, looseness=0.5] (T121) to node[below, pos=0.1] {$\hb$}(T93);
\draw[rf, out=-170, in=-5, looseness=0.5] (T51) to (T34);
\draw[rf] (T81) to node[above, pos=0.25] {$\rf$} (T35);
\draw[rf] (T101) to (T45);
\draw[rf] (T111) to (T75);
\end{pgfonlayer}

\end{scope}
}
%
%
\begin{subfigure}{\textwidth}
\centering
\scalebox{\scale}{
\begin{tikzpicture}[thick, font=\footnotesize,
pre/.style={<-,shorten >= 2pt, shorten <=2pt, very thick},
post/.style={->,shorten >= 3pt, shorten <=3pt,   thick},
seqtrace/.style={line width=1},
und/.style={very thick, draw=gray},
virt/.style={circle,draw=black!50,fill=black!20, opacity=0}]

\BaseCopyGadget

%

\begin{pgfonlayer}{bg}
\draw[rf, out=70, in=-165, looseness=0.75] (T15) to node [pos=0.15] (first) {\circled{1}} (T31);
\draw[rf, out=160, in=0, looseness=0.5] (T73) to node [pos=0.85] (sec) {\circled{2}} (T32);
\draw[rf] (T93) -- (T71) node [pos=0.65] (third) {\circled{3}};
\draw[rf, out=70, in=-165, looseness=0.75] (T105) to node [pos=0.15] (four) {\circled{4}} (T121);
\draw[rf] (T55) to (T41);
\draw[rf] (T21) to (T33);

\draw[mo, out=-30, in=210, looseness=0.1] (T13) to (T73);
\end{pgfonlayer}

\end{tikzpicture}
}
\caption{\label{subfig:copy_gadget_down_wra_ra_sraFalse}
Choosing $\TReadXOne{v^{i+1}_j}{t_3}$ to observe from $t_1$ forces the sequence of $\rf$ and $\mo$ edges shown.
}
\end{subfigure}
\\
\begin{subfigure}{\textwidth}
\centering
\scalebox{\scale}{%
\begin{tikzpicture}[thick, font=\footnotesize,
pre/.style={<-,shorten >= 2pt, shorten <=2pt, very thick},
post/.style={->,shorten >= 3pt, shorten <=3pt,   thick},
seqtrace/.style={line width=1},
und/.style={very thick, draw=gray},
virt/.style={circle,draw=black!50,fill=black!20, opacity=0}]

\BaseCopyGadget

%

\begin{pgfonlayer}{bg}
\draw[rf, out=30, in=-155, looseness=0.5] (T25) to node [pos=0.15] (first) {\circled{1}} (T31);
\draw[rf, out=145, in=35, looseness=0.5] (T43) to  node [pos=0.1,above=.1mm] (seco) {\circled{2}} (T33);
\draw[rf] (T63) -- (T41) node [pos=0.3] (thir) {\circled{3}};
\draw[rf, out=30] (T115) -- (T121) node [pos=0.3] (fourt) {\circled{4}};
\draw[rf] (T85) to (T71);
\draw[rf, out=-30, in=170, looseness=0.5] (T11) to (T32);

\draw[mo, out=-20, in=200, looseness=0.3] (T23) to (T43);
\end{pgfonlayer}

\end{tikzpicture}
}
\caption{\label{subfig:copy_gadget_down_wra_ra_sraTrue}
Choosing $\TReadXOne{v^{i+1}_j}{t_3}$ to read from $t_2$ forces the sequence of $\rf$ and $\mo$ edges shown.
}
\end{subfigure}
\caption{\label{fig:copy_gadget_down_wra_ra_sra}
The copy-down gadget $\CopyGadgetDown{i}{j}$ is very similar to $\CopyGadget{i}{j}$.
Choosing $\ReadXOne{v^{i+1}_j}$ to read from $t_1$ (\subref{subfig:copy_gadget_down_wra_ra_sraFalse}) corresponds to setting $s_j=\bfalse$ and also forces $\ReadXTwo{v^i_j}$ to read from $t_4$. 
Choosing $\ReadXOne{v^{i+1}_j}$ to read from $t_2$ (\subref{subfig:copy_gadget_down_wra_ra_sraTrue}) corresponds to setting $s_j=\btrue$ and also forces $\ReadXTwo{v^i_j}$ to read from $t_5$.
This $\rf$ coupling is formalized in \cref{lem:wra_ra_sra_soundness_copy}.
The edge numbers specify the order in which $\rf$-edges are inferred.
}
\end{figure}

%% file: figures/at_most_gadget_wra_ra_sra.tex
\begin{figure}
\newcommand{\xdisposition}{0}
\newcommand{\ydisposition}{0}
\newcommand{\xtstep}{0.75}
\newcommand{\ytstep}{0.7}
\newcommand{\ybias}{-0.3 }
\newcommand{\xstep}{1.4}
\newcommand{\ystep}{-0.8}
\newcommand{\xtscale}{0.8}
\def\crossoutopacity{0.3}
\def \numevents{6}
\newcommand{\BaseAtMostOneGadget}[1]{
\node[] (S11) at (1*\xstep,0.15) {\normalsize $t_1$};
\node[] (S12) at (1*\xstep,\numevents * \ystep) {};
\node[] (S21) at (2*\xstep,0.15) {\normalsize $t_2$};
\node[] (S22) at (2*\xstep,\numevents * \ystep) {};
\node[] (S31) at (3*\xstep,0.15) {\normalsize $t_3$};
\node[] (S32) at (3*\xstep,\numevents * \ystep) {};
\node[] (S41) at (4*\xstep,0.15) {\normalsize $h_{c}$};
\node[] (S42) at (4*\xstep,\numevents * \ystep) {};

\draw[seqtrace] (S11) to (S12);
\draw[seqtrace] (S21) to (S22);
\draw[seqtrace] (S31) to (S32);
\draw[seqtrace] (S41) to (S42);

\node[event, draw=black] (13) at (1*\xstep, 3*\ystep + 0*\ybias) {$\WriteXOne{v^{i}_{j}}$};
\node[event, draw=black] (16) at (1*\xstep, 6*\ystep + 0*\ybias) {$\WriteXOne{v^{i}_{k}}$};

\node[event] (21) at (2*\xstep, 1*\ystep + 0*\ybias) {$\WriteA{v^{i}_{j}}{c}$};
\node[event] (22) at (2*\xstep, 2*\ystep + 0*\ybias) {$\WriteA{\ov{v}^{i}_{j}}{c}$};
\node[event, draw=black] (23) at (2*\xstep, 3*\ystep + 0*\ybias) {$\WriteXOne{v^{i}_{j}}$};
\node[event] (24) at (2*\xstep, 4*\ystep + 0*\ybias) {$\WriteA{v^{i}_{k}}{c}$};
\node[event] (25) at (2*\xstep, 5*\ystep + 0*\ybias) {$\WriteA{\ov{v}^{i}_{k}}{c}$};
\node[event, draw=black] (26) at (2*\xstep, 6*\ystep + 0*\ybias) {$\WriteXOne{v^{i}_{k}}$};

\node[event, draw=black] (32) at (3*\xstep, 2*\ystep + 0*\ybias) {$\ReadXOne{v^{i}_{j}}$};
\node[event] (33) at (3*\xstep, 3*\ystep + 0*\ybias) {$\ReadA{v^{i}_{j}}{c}$};
\node[event, draw=black] (35) at (3*\xstep, 5*\ystep + 0*\ybias) {$\ReadXOne{v^{i}_{k}}$};
\node[event] (36) at (3*\xstep, 6*\ystep + 0*\ybias) {$\ReadA{v^{i}_{k}}{c}$};

\ifthenelse{\equal{#1}{0}}{
\node[event] (43) at (4*\xstep, 3*\ystep + 0*\ybias) {$\WriteA{v^{i}_{k}}{c}$};
\node[event] (45) at (4*\xstep, 5*\ystep + 0*\ybias) {$\WriteA{v^{i}_{j}}{c}$};
}{
\ifthenelse{\equal{#1}{1}}{
\node[event, cross out, draw, draw opacity=\crossoutopacity] (43) at (4*\xstep, 3*\ystep + 0*\ybias) {$\WriteA{v^{i}_{k}}{c}$};
\node[event] (45) at (4*\xstep, 5*\ystep + 0*\ybias) {$\WriteA{v^{i}_{j}}{c}$};
}{
\ifthenelse{\equal{#1}{2}}{
\node[event] (43) at (4*\xstep, 3*\ystep + 0*\ybias) {$\WriteA{v^{i}_{k}}{c}$};
\node[event, cross out, draw, draw opacity=\crossoutopacity] (45) at (4*\xstep, 5*\ystep + 0*\ybias) {$\WriteA{v^{i}_{j}}{c}$};
}{
\node[event, cross out, draw, draw opacity=\crossoutopacity] (43) at (4*\xstep, 3*\ystep + 0*\ybias) {$\WriteA{v^{i}_{k}}{c}$};
\node[event, cross out, draw, draw opacity=\crossoutopacity] (45) at (4*\xstep, 5*\ystep + 0*\ybias) {$\WriteA{v^{i}_{j}}{c}$};
}}}
}
\def\scale{1}
\centering
\begin{subfigure}[b]{0.475\textwidth}
\centering
\scalebox{\scale}{
\begin{tikzpicture}[thick, font=\footnotesize,
pre/.style={<-,shorten >= 2pt, shorten <=2pt, very thick},
post/.style={->,shorten >= 3pt, shorten <=3pt,   thick},
seqtrace/.style={line width=1},
und/.style={very thick, draw=gray},
virt/.style={circle,draw=black!50,fill=black!20, opacity=0}]

\BaseAtMostOneGadget{0}

\end{tikzpicture}
}
\caption{\label{subfig:at_most_one_gadget_wra_ra_sra}
The at-most-one-true gadget $\AtMostOneGadget{c}{i}{j}{k}$.
}
\end{subfigure}
\hfill
\begin{subfigure}[b]{0.475\textwidth}
\centering
\scalebox{\scale}{
\begin{tikzpicture}[thick, font=\footnotesize,
pre/.style={<-,shorten >= 2pt, shorten <=2pt, very thick},
post/.style={->,shorten >= 3pt, shorten <=3pt,   thick},
seqtrace/.style={line width=1},
und/.style={very thick, draw=gray},
virt/.style={circle,draw=black!50,fill=black!20, opacity=0}]

\BaseAtMostOneGadget{1}

\begin{pgfonlayer}{bg}
\draw[rf] (23) -- (32) node [pos=0.6] (fir) {\circledsmall{1}};
\draw[rf] (45) -- (33) node [pos=0.6] (seco) {\circledsmall{2}};
\draw[rf] (24) -- (36) node [pos=0.2] (thi) {\circledsmall{3}}; 
\draw[rf] (16) -- (35) node [pos=0.6] (four) {\circledsmall{4}};

\draw[mo] (22) to (45);
\draw[mo] (45) to (24);
\end{pgfonlayer}

\end{tikzpicture}
}
\caption{\label{subfig:at_most_one_gadget_wra_ra_sra_true_false}
Resolving with $s_j=\btrue$ and $s_k=\bfalse$.
}
\end{subfigure}
\\
\begin{subfigure}[b]{0.475\textwidth}
\centering
\scalebox{\scale}{
\begin{tikzpicture}[thick, font=\footnotesize,
pre/.style={<-,shorten >= 2pt, shorten <=2pt, very thick},
post/.style={->,shorten >= 3pt, shorten <=3pt,   thick},
seqtrace/.style={line width=1},
und/.style={very thick, draw=gray},
virt/.style={circle,draw=black!50,fill=black!20, opacity=0}]

\BaseAtMostOneGadget{2}

\begin{pgfonlayer}{bg}
\draw[rf] (26) -- (35) node [pos=0.4] (fir) {\circledsmall{1}};
\draw[rf, bend left=15] (43) -- (36) node [pos=0.2] (sec) {\circledsmall{2}};
\draw[rf] (21) -- (33) node [pos=0.2] (thi) {\circledsmall{3}};
\draw[rf] (13) -- (32) node [pos=0.65] (for) {\circledsmall{4}};

\draw[mo] (25) to (43);
\end{pgfonlayer}

\end{tikzpicture}
}
\caption{\label{subfig:at_most_one_gadget_wra_ra_sra_false_true}
Resolving with $s_j=\bfalse$ and $s_k=\btrue$.
}
\end{subfigure}
\hfill
\begin{subfigure}[b]{0.475\textwidth}
\centering
\scalebox{\scale}{
\begin{tikzpicture}[thick, font=\footnotesize,
pre/.style={<-,shorten >= 2pt, shorten <=2pt, very thick},
post/.style={->,shorten >= 3pt, shorten <=3pt,   thick},
seqtrace/.style={line width=1},
und/.style={very thick, draw=gray},
virt/.style={circle,draw=black!50,fill=black!20, opacity=0}]

\BaseAtMostOneGadget

\begin{pgfonlayer}{bg}{3}
\draw[rf] (13) -- (32) node [pos=0.6] (fir) {\circledsmall{1}};
\draw[rf] (21) -- (33) node [pos=0.2] (sec) {\circledsmall{2}};
\draw[rf] (16) -- (35) node [pos=0.6] (thi) {\circledsmall{1}};
\draw[rf] (24) -- (36) node [pos=0.2] (fo) {\circledsmall{2}};

\end{pgfonlayer}

\end{tikzpicture}
}
\caption{\label{subfig:at_most_one_gadget_wra_ra_sra_false_false}
Resolving with $s_j=\bfalse$ and $s_k=\bfalse$.
}
\end{subfigure}
\caption{\label{fig:at_most_one_gadget_wra_ra_sra}
The at-most-one-true gadget $\AtMostOneGadget{c}{i}{j}{k}$ parameterized by $c\in [3]$,
and the three ways to resolve it depending on the boolean assignment to $s_j$ and $s_k$ (\subref{subfig:at_most_one_gadget_wra_ra_sra_true_false}, \subref{subfig:at_most_one_gadget_wra_ra_sra_false_true}, \subref{subfig:at_most_one_gadget_wra_ra_sra_false_false}).
These $\rf$ constraints  are formalized in \cref{lem:wra_ra_sra_soundness_at_most_one}.
The edge numbers specify the order in which $\rf$-edges are implied.
The crossed-out events are ignored in our analysis for simplicity (see \cref{lem:wra_ra_sra_completeness_inactive_writes}).
}
\end{figure}

%% file: figures/at_least_gadget_wra_ra_sra.tex

\begin{figure}
\newcommand{\xdisposition}{0}
\newcommand{\ydisposition}{0}
\newcommand{\xtstep}{0.75}
\newcommand{\ytstep}{0.7}
\newcommand{\ybias}{-0.3 }
\newcommand{\xstep}{1.4}
\newcommand{\ystep}{-0.6}
\newcommand{\xtscale}{0.8}
\def\crossoutopacity{0.3}
\def \numevents{9}
\newcommand{\BaseAtLeastOneGadget}[1]{
\node[] (S11) at (1*\xstep,0.15) {\normalsize $t_1$};
\node[] (S12) at (1*\xstep,\numevents * \ystep) {};
\node[] (S21) at (2*\xstep,0.15) {\normalsize $t_2$};
\node[] (S22) at (2*\xstep,\numevents * \ystep) {};
\node[] (S31) at (3*\xstep,0.15) {\normalsize $t_3$};
\node[] (S32) at (3*\xstep,\numevents * \ystep) {};
\node[] (S41) at (4*\xstep,0.15) {\normalsize $p$};
\node[] (S42) at (4*\xstep,\numevents * \ystep) {};
\node[] (S51) at (5*\xstep,0.15) {\normalsize $q$};
\node[] (S52) at (5*\xstep,\numevents * \ystep) {};

\draw[seqtrace] (S11) to (S12);
\draw[seqtrace] (S21) to (S22);
\draw[seqtrace] (S31) to (S32);
\draw[seqtrace] (S41) to (S42);
\draw[seqtrace] (S51) to (S52);

\node[event] (11) at (1*\xstep, 1*\ystep + 0*\ybias) {$\WriteB{v^{i}_j}$};
\node[event] (12) at (1*\xstep, 2*\ystep + 0*\ybias) {$\WriteB{\ov{v}^{i}_j}$};
\node[event, draw=black] (13) at (1*\xstep, 3*\ystep + 0*\ybias) {$\WriteXOne{v^{i}_j}$};
\node[event] (14) at (1*\xstep, 4*\ystep + 0*\ybias) {$\WriteB{v^{i}_k}$};
\node[event] (15) at (1*\xstep, 5*\ystep + 0*\ybias) {$\WriteB{\ov{v}^{i}_k}$};
\node[event, draw=black] (16) at (1*\xstep, 6*\ystep + 0*\ybias) {$\WriteXOne{v^{i}_k}$};
\node[event] (17) at (1*\xstep, 7*\ystep + 0*\ybias) {$\WriteB{v^{i}_{\ell}}$};
\node[event] (18) at (1*\xstep, 8*\ystep + 0*\ybias) {$\WriteB{\ov{v}^{i}_{\ell}}$};
\node[event, draw=black] (19) at (1*\xstep, 9*\ystep + 0*\ybias) {$\WriteXOne{v^{i}_{\ell}}$};

\node[event, draw=black] (23) at (2*\xstep, 3*\ystep + 0*\ybias) {$\WriteXOne{v^{i}_j}$};
\node[event, draw=black] (26) at (2*\xstep, 6*\ystep + 0*\ybias) {$\WriteXOne{v^{i}_k}$};
\node[event, draw=black] (29) at (2*\xstep, 9*\ystep + 0*\ybias) {$\WriteXOne{v^{i}_{\ell}}$};

\node[event, draw=black] (32) at (3*\xstep, 2*\ystep + 0*\ybias) {$\ReadXOne{v^{i}_j}$};
\node[event] (33) at (3*\xstep, 3*\ystep + 0*\ybias) {$\ReadB{v^{i}_j}$};
\node[event, draw=black] (35) at (3*\xstep, 5*\ystep + 0*\ybias) {$\ReadXOne{v^{i}_k}$};
\node[event] (36) at (3*\xstep, 6*\ystep + 0*\ybias) {$\ReadB{v^{i}_k}$};
\node[event, draw=black] (38) at (3*\xstep, 8*\ystep + 0*\ybias) {$\ReadXOne{v^{i}_\ell}$};
\node[event] (39) at (3*\xstep, 9*\ystep + 0*\ybias) {$\ReadB{v^{i}_{\ell}}$};

\ifthenelse{\equal{#1}{0}}{
\node[event] (43) at (4*\xstep, 3*\ystep + 0*\ybias) {$\WriteB{v^{i}_{\ell}}$};
\node[event] (45) at (4*\xstep, 5*\ystep + 0*\ybias) {$\WriteB{v^{i}_j}$};

\node[event] (56) at (5*\xstep, 6*\ystep + 0*\ybias) {$\WriteB{v^{i}_{\ell}}$};
\node[event] (58) at (5*\xstep, 8*\ystep + 0*\ybias) {$\WriteB{v^{i}_k}$};
}{
\ifthenelse{\equal{#1}{1}}{
\node[event, cross out, draw, draw opacity=\crossoutopacity] (43) at (4*\xstep, 3*\ystep + 0*\ybias) {$\WriteB{v^{i}_{\ell}}$};
\node[event] (45) at (4*\xstep, 5*\ystep + 0*\ybias) {$\WriteB{v^{i}_j}$};

\node[event, cross out, draw, draw opacity=\crossoutopacity] (56) at (5*\xstep, 6*\ystep + 0*\ybias) {$\WriteB{v^{i}_{\ell}}$};
\node[event] (58) at (5*\xstep, 8*\ystep + 0*\ybias) {$\WriteB{v^{i}_k}$};
}{
\ifthenelse{\equal{#1}{2}}{
\node[event, cross out, draw, draw opacity=\crossoutopacity] (43) at (4*\xstep, 3*\ystep + 0*\ybias) {$\WriteB{v^{i}_{\ell}}$};
\node[event] (45) at (4*\xstep, 5*\ystep + 0*\ybias) {$\WriteB{v^{i}_j}$};

\node[event] (56) at (5*\xstep, 6*\ystep + 0*\ybias) {$\WriteB{v^{i}_\ell}$};
\node[event, cross out, draw, draw opacity=\crossoutopacity] (58) at (5*\xstep, 8*\ystep + 0*\ybias) {$\WriteB{v^{i}_k}$};
}{
\ifthenelse{\equal{#1}{3}}{
\node[event] (43) at (4*\xstep, 3*\ystep + 0*\ybias) {$\WriteB{v^{i}_{\ell}}$};
\node[event, cross out, draw, draw opacity=\crossoutopacity] (45) at (4*\xstep, 5*\ystep + 0*\ybias) {$\WriteB{v^{i}_j}$};

\node[event, cross out, draw, draw opacity=\crossoutopacity] (56) at (5*\xstep, 6*\ystep + 0*\ybias) {$\WriteB{v^{i}_{\ell}}$};
\node[event] (58) at (5*\xstep, 8*\ystep + 0*\ybias) {$\WriteB{v^{i}_k}$};
}}}}
}
\def\scale{0.97}
\centering
\begin{subfigure}[b]{0.475\textwidth}
\centering
\scalebox{\scale}{%
\begin{tikzpicture}[thick, font=\footnotesize,
pre/.style={<-,shorten >= 2pt, shorten <=2pt, very thick},
post/.style={->,shorten >= 3pt, shorten <=3pt,   thick},
seqtrace/.style={line width=1},
und/.style={very thick, draw=gray},
virt/.style={circle,draw=black!50,fill=black!20, opacity=0}]

\BaseAtLeastOneGadget{0}

\end{tikzpicture}
}
\caption{\label{subfig:at_least_one_gadget_wra_ra_sra}
The at-least-one-true gadget $\AtLeastOneGadget{i}{j}{k}{\ell}$.
}
\end{subfigure}
\hfill
\begin{subfigure}[b]{0.475\textwidth}
\centering
\scalebox{\scale}{%
\begin{tikzpicture}[thick, font=\footnotesize,
pre/.style={<-,shorten >= 2pt, shorten <=2pt, very thick},
post/.style={->,shorten >= 3pt, shorten <=3pt,   thick},
seqtrace/.style={line width=1},
und/.style={very thick, draw=gray},
virt/.style={circle,draw=black!50,fill=black!20, opacity=0}]

\BaseAtLeastOneGadget{1}

\begin{pgfonlayer}{bg}
\draw[rf, bend left=10] (13) -- (32) node [pos=0.7] (fir) {\circledsmall{1}};
\draw[rf] (45) -- (33) node [pos=0.6] (sec) {\circledsmall{2}};
\draw[rf, bend left=10] (16) -- (35) node [pos=0.6] (thr) {\circledsmall{1}};
\draw[rf] (58) -- (36) node [pos=0.6] (fo) {\circledsmall{2}};
\draw[rf] (17) -- (39) node [pos=0.6] (fiv) {\circledsmall{3}};
\draw[rf, bend left=10] (29) -- (38) node [pos=0.5,above] (six) {\circledsmall{4}};

\draw[mo] (12) to (45);
\draw[mo] (45) to (58);
\draw[mo] (15) to (58);
\draw[mo, bend right=5] (58) to (17);
\end{pgfonlayer}

\end{tikzpicture}
}
\caption{\label{subfig:at_least_one_gadget_wra_ra_sra_false_false_true}
Resolving with $s_j=\bfalse$, $s_k=\bfalse$ and $s_\ell=\btrue$.
}
\end{subfigure}
\\[2em]
\begin{subfigure}[b]{0.475\textwidth}
\centering
\scalebox{\scale}{%
\begin{tikzpicture}[thick, font=\footnotesize,
pre/.style={<-,shorten >= 2pt, shorten <=2pt, very thick},
post/.style={->,shorten >= 3pt, shorten <=3pt, thick},
seqtrace/.style={line width=1},
und/.style={very thick, draw=gray},
virt/.style={circle,draw=black!50,fill=black!20, opacity=0}]

\BaseAtLeastOneGadget{2}

\begin{pgfonlayer}{bg}
\draw[rf, bend left=10] (13) -- (32) node [pos=0.7] (fir) {\circledsmall{1}};
\draw[rf] (45) -- (33) node [pos=0.7] (sec) {\circledsmall{2}};
\draw[rf, bend left=10] (26) -- (35) node [pos=0.7,below] (thr) {\circledsmall{4}};
\draw[rf] (14) -- (36) node [pos=0.3] (fo) {\circledsmall{3}};
\draw[rf] (56) -- (39) node [pos=0.3] (fiv) {\circledsmall{2}};
\draw[rf, bend left=10] (19) -- (38) node [pos=0.65] (six) {\circledsmall{1}};

\draw[mo] (12) to (45);
\draw[mo, bend right=10] (45) to (14);
\draw[mo] (18) to (56);
\end{pgfonlayer}

\end{tikzpicture}
}
\caption{\label{subfig:at_least_one_gadget_wra_ra_sra_false_true_false}
Resolving with $s_j=\bfalse$, $s_k=\btrue$ and $s_\ell=\bfalse$.
}
\end{subfigure}
\hfill
\begin{subfigure}[b]{0.475\textwidth}
\centering
\scalebox{\scale}{%
\begin{tikzpicture}[thick, font=\footnotesize,
pre/.style={<-,shorten >= 2pt, shorten <=2pt, very thick},
post/.style={->,shorten >= 3pt, shorten <=3pt,   thick},
seqtrace/.style={line width=1},
und/.style={very thick, draw=gray},
virt/.style={circle,draw=black!50,fill=black!20, opacity=0}]

\BaseAtLeastOneGadget{3}

\begin{pgfonlayer}{bg}
\draw[rf, bend left=10] (23) -- (32) node [pos=0.5,below] (one) {\circledsmall{4}};
\draw[rf] (11) -- (33) node [pos=0.3] (two) {\circledsmall{3}};
\draw[rf, bend left=10] (16) -- (35) node [pos=0.8] (three) {\circledsmall{1}};
\draw[rf, bend right=-50] (58) -- (36)node [pos=0.7] (four) {\circledsmall{2}};
\draw[rf, out=-150, in=20, looseness=0.5] (43) -- (39) node [pos=0.7] (five) {\circledsmall{2}};
\draw[rf, bend left=10] (19) -- (38) node [pos=0.7] (six) {\circledsmall{1}};

\draw[mo] (15) to (58);
\draw[mo] (58) to (43);
\draw[mo] (18) to (43);
\end{pgfonlayer}

\end{tikzpicture}
}
\caption{\label{subfig:at_least_one_gadget_wra_ra_sra_true_false_false}
Resolving with $s_j=\btrue$, $s_k=\bfalse$ and $s_\ell=\bfalse$.
}
\end{subfigure}
\caption{\label{fig:at_least_one_gadget_wra_ra_sra}
The at-least-one-true gadget $\AtLeastOneGadget{i}{j}{k}{\ell}$ (\subref{subfig:at_least_one_gadget_wra_ra_sra}) and the three ways to resolve it depending on the boolean assignment to $s_j$ and $s_k$ (\subref{subfig:at_least_one_gadget_wra_ra_sra_true_false_false}, \subref{subfig:at_least_one_gadget_wra_ra_sra_false_true_false}, \subref{subfig:at_least_one_gadget_wra_ra_sra_false_false_true}).
These $\rf$ constraints  are formalized in \cref{lem:wra_ra_sra_soundness_at_least_one}.
The edge numbers specify the order in which $\rf$-edges are implied.
The crossed-out events are ignored in our analysis for simplicity (see \cref{lem:wra_ra_sra_completeness_inactive_writes}).
}
\end{figure}

%% file: other_models.tex
\section{Implications and Other Memory Models}\label{SEC:OTHER_MODELS}

Our proof of \cref{thm:lower_ra} is strong enough to yield hardness on other popular memory models across different domains.
In this section, we explore its implications.

\Paragraph{Causal Consistency models.}
In a distributed setting, consistency commonly captures the concept of \emph{causality}.
Three of the most standard causal models are the basic Causal Consistency ($\ccmm$), Causal Convergence ($\cvmm$), and Causal Memory ($\cmmm$)~\cite{Bouajjani2017}.
It was recently shown that $\ccmm$ coincides with $\wramm$ while $\cvmm$ coincides with $\sramm$~\cite{Lahav:2022}. 
Thus, \cref{thm:lower_ra} implies $\NP$-completeness for all models between $\cvmm$ and $\ccmm$.
In \cref{SUBSEC:CAUSAL_MEMORY} we also establish $\NP$-completeness for $\cmmm$, by extending the proof of \cref{thm:lower_ra}, thereby completing \cref{thm:lower_cc}.

\Paragraph{Hardware memory models.}
Next, we turn our attention to some popular hardware memory models, namely, for the POWER and x86-TSO architectures, as well as PSO.
We show that \cref{thm:lower_ra} implies $\NP$-completeness for POWER,
but consistency checks for TSO/PSO run in polynomial time.

\Paragraph{The fully Relaxed model.}
Finally, we consider the Relaxed model which does not require $(\po\cup\rf)$-acyclicity, 
and remark that this brings the problem in polynomial time.
Interestingly, $\rlxmm$ is the only non-multi-copy atomic model in our list for which consistency testing is tractable.

\input{causal_memory}

\input{power}

\input{tso_pso}
\input{rlx}

%% file: causal_memory.tex
\subsection{Implications for Causal Memory}\label{SUBSEC:CAUSAL_MEMORY}

Here we prove the case (ii) of \cref{thm:lower_cc}, i.e., the hardness of bounded consistency testing for any memory model $\MemModel$ with $\ccmm\mmorder \MemModel\mmorder \cmmm$.
Instead of performing a separate reduction, we reuse our reduction from \cref{SEC:LOWER_THREADS_LOCATIONS_WRA_RA_SRA}.
Let $\varphi=\{\Clause_i\}_{i\in[m]}$ be a Boolean formula over $n$ variables $\{s_j\}_{j\in[n]}$ and $m$ clauses of the form $\Clause_i=(s_j, s_k, s_{\ell})$, for which we have to solve Monotone 3SAT.
Moreover, let $\expartial=(\E, \po)$ be the abstract execution as constructed in \cref{SEC:LOWER_THREADS_LOCATIONS_WRA_RA_SRA}.
Since $\wramm$ coincides with $\ccmm$~\cite{Lahav:2022}, we have the following soundness corollary.

\begin{corollary}\label{cor:cm_soundness}
If $\expartial\models \ccmm$ then $\varphi$ is satisfiable.
\end{corollary}

Thus, to complete the theorem, it remains to show that if $\expartial\models \cmmm$ then $\varphi$ is satisfiable.
Towards this, we next formalize $\cmmm$ and then prove the statement.

\Paragraph{Causal Memory.}
Causal memory is similar to Causal Consistency but further requires that each thread has a locally-consistent view of the order which different writes have been executed~\cite{Ahamad1995,Bouajjani2017}.
This is made formal by introducing one additional relation for each thread, called the \emph{observed-before} relation $\ob{}$.

\Paragraph{The observed-before relation $\ob{}$.}
Given an event $\event$, the \emph{observed-before relation}\footnote{Other works refer to this relation as ``happened-before'' for $\event$. We avoid this term here to not confuse it with $\hb$.} for $\event$ is the smallest transitive relation $\ob{\event}\subseteq \E\times \E$ with the following properties.
\begin{compactenum}
\item For every $(\event_1, \event_2)\in \hb$ such that $(\event_i, \event)\in \hb$ for each $i\in[2]$, we have $(\event_1, \event_2)\in \ob{\event}$.
\item For every conflicting triplet $(\wt, \rd, \wt')$ such that 
(i)~$(\wt',\rd)\in \ob{\event}$ and
(ii)~$(\rd, \event)\in \po^?$,
we have $(\wt', \wt)\in \ob{\event}$.
\end{compactenum}
Intuitively, when $t$ executed $\rd$, it must have observed $\wt$ after $\wt'$, so that $\rd$ indeed obtained its value from $\wt$.
The $\ob{\event}$ relation specifies that this ordering cannot change later when $t$ executes $\event$.
Notice the fixpoint style of the above definition.
As we order $(\wt_1, \wt_2)\in \ob{\event}$ (and since the relation is transitive and contains $\hb$), 
more and more write-read pairs satisfy property (i)~$(\wt', \rd)\in \ob{\event}$, triggering the addition of new orderings $(\wt', \wt)\in \ob{\event}$.
For any two events $\event_1$, $\event_2$ with $(\event_1, \event_2)\in \po$, we have $\ob{\event_1}\subseteq \ob{\event_2}$, i.e., $\ob{}$ grows monotonically as we go downwards in each thread.
The observed-before relation for a thread $t$ is defined as $\ob{t}=\ob{\event^{\max}}$, where $\event^{\max}$ is the $\po$-maximal event of $t$.
The new axiom requires that $\ob{t}$ is irreflexive~\cite{Bouajjani2017}.

\begin{compactitem}
\item $\ob{t}$ is irreflexive for each thread $t$ \hfill (\OB{})    
\end{compactitem}

In turn, $\cmmm$ is equal to $\wramm$ with \OB{} as an extra axiom.

\begin{compactitem}
\item $(\text{\PORF{}})\land(\text{\WRCoh{}})\land (\text{\OB{}})$ \hfill [$\cmmm$]
\end{compactitem}

\input{figures/cm}

Observe that $\wramm\mmorder \cmmm$, but $\cmmm$ is incomparable with $\ramm$/$\sramm$, i.e., $\cmmm$ allows executions that are inconsistent in $\ramm$/$\sramm$ and vice versa.
See \cref{subfig:cm_not_ra} and \cref{subfig:sra_not_cm} for illustrations.

Next, we prove the completeness of the construction, i.e., if $\varphi$ is satisfiable then $\expartial\models \cmmm$.
Consider the reads-from relation $\rf$ and the partial modification order $\mopartial$ exactly as constructed in the completeness argument of \cref{SEC:LOWER_THREADS_LOCATIONS_WRA_RA_SRA} (i.e., following the gadgets in \cref{fig:copy_gadget_wra_ra_sra,fig:copy_gadget_down_wra_ra_sra,fig:at_most_one_gadget_wra_ra_sra,fig:at_least_one_gadget_wra_ra_sra}).
Let $\ex=(\E, \po, \rf, \mo)$ be the execution witnessing the $\sramm$-consistency of $\expartial$ according to \cref{cor:wra_ra_sra_lower_threads_completeness}.
Since $\sramm$ satisfies the \PORF{} and \WRCoh{} axioms, we only need to argue that $\ex$ also satisfies \OB{} to conclude that $\ex\models \cmmm$.
For this, we have to establish some additional lemmas.

Our first lemma stems from \cref{lem:wra_ra_sra_rf_immediate} and states an important property of $\rf$:~for every pair $(\wt,\rd)\in \rf$, $\wt$ has no $\hb$-path to $\po$-predecessors of $\rd$.
In other words, the first event of the thread of $\rd$ that $\wt$ can reach by means of an $\hb$-path is $\rd$ itself via the $\rf$-edge $\wt\LTo{\rf}\rd$.

\begin{restatable}{lemma}{lemcmrfimmediate}\label{lem:cm_rf_immediate}
For every $(\wt, \rd)\in \rf$ and event $\event$ such that $(\event, \rd)\in \po$, we have $(\wt, \event)\not \in \hb$.
\end{restatable}

Next, we define a ``one-hop'' variant of $\ob{}$.
Given an event $\event$, the \emph{one-hop observed-before relation} for $\event$ is the smallest transitive relation $\obOne{\event}\subseteq \E\times \E$ with the following properties.
\begin{compactenum}
\item For every $(\event_1, \event_2)\in \hb$ such that $(\event_i, \event)\in \hb$ for each $i\in[2]$, we have $(\event_1, \event_2)\in \obOne{\event}$.
\item For every event $\event$ and conflicting triplet $(\wt, \rd, \wt')$ such that 
(i)~$(\wt',\rd)\in \hb$ and
(ii)~$(\rd, \event)\in \po^?$,
we have $(\wt', \wt)\in \obOne{\event}$.
\end{compactenum}
Contrasting $\obOne{\event}$ to $\ob{\event}$, the only difference is in condition (2i):~$\ob{\event}$ checks whether $(\wt',\rd)\in \ob{\event}$, while $\obOne{\event}$ checks the weaker condition $(\wt',\rd)\in \hb$.
Thus $\obOne{\event}$ does not have the fixpoint style of $\ob{\event}$.
Similarly to $\ob{t}$, we let $\obOne{t}=\obOne{\event^{\max}}$, where $\event^{\max}$ is the $\po$-maximal event of thread $t$.

Our next lemma states that for our execution $\ex$,
$\ob{\event}$ coincides with $\obOne{\event}$.
In other words, $\ob{\event}$ reaches a fixpoint after only one iteration.
This observation stems from \cref{lem:cm_rf_immediate}.
Intuitively, since for each triplet $(\wt, \rd, \wt')$, $\wt$ cannot $\hb$-reach any $\po$-predecessor of $\rd$,
traversing an edge $\wt'\LTo{\ob{\event}}\wt$ cannot lead to any events of the thread of $\rd$ that weren't already reachable via the $\hb$-path $\wt'\LPath{\hb}\rd$ that made us insert $(\wt',\wt)\in \ob{\event}$ in the first place (see~\cref{subfig:ob}).
Hence, adding such an ordering $(\wt',\wt)\in \ob{\event}$ cannot lead to further firings of condition (2i) of $\ob{e}$.
Formally, we have the following.

\begin{restatable}{lemma}{lemcmob}\label{lem:cm_ob}
For every thread $t$, we have $\obOne{t}=\ob{t}$.
\end{restatable}

Finally, observe that whenever we add $(\wt', \wt) \in\obOne{\event}$, we have $(\wt', \rd)\in \hb$.
Due to \cref{lem:wra_ra_sra_completeness_safe_triplets}, the triplet $(\wt, \rd, \wt')$ is safe, thus $(\wt', \wt)\in (\hb\cup\mopartial)^+$.
Hence, the acyclicity of $\ob{t}=\obOne{t}$ follows from the acyclicity of $(\hb\cup \mopartial)^+$ (\cref{lem:wra_ra_sra_completeness_no_cycle_hb_mo}).
We thus have the following lemma, which, together with \cref{cor:cm_soundness}, completes the proof case (ii) of \cref{thm:lower_cc}.

\begin{restatable}{lemma}{lemcmcompleteness}\label{lem:cm_completeness}
If $\varphi$ is satisfiable, then $\expartial\models \cmmm$.
\end{restatable}

%% file: figures/cm.tex

\begin{figure}
\centering
\begin{subfigure}[b]{0.20\textwidth}
\centering
\scalebox{0.9}{
\begin{tikzpicture}[yscale=1]
\def\ystep{1.6}
\def\xstep{2}
\def\ybias{0.6}
	\node (t10) at (0,0*\ystep+\ybias)  {\large $t_1$};
  \node (t11) at (0,0*\ystep)  {$\wt(x)$};
  \node (t12) at (0,-1*\ystep) {$\rd(x)$};
	\node (t20) at (1*\xstep,0*\ystep+\ybias)  {\large $t_2$};
  \node (t21) at (1*\xstep,0*\ystep) {$\wt(x)$};
  \node (t22) at (1*\xstep,-1*\ystep) {$\rd(x)$};
  \draw[po] (t11) to (t12);
  \draw[po] (t21) to (t22);
\draw[rf,bend left=0] (t11) to node[above,pos=0.9, sloped]{$\rf$} (t22);
\draw[rf,bend left=0] (t21) to node[above,pos=0.9, sloped]{$\rf$} (t12);
\draw[ob,bend left=10] (t11) to node[above,pos=0.5, sloped]{$\ob{t_1}$} (t21);
\draw[ob,bend left=10] (t21) to node[below,pos=0.5, sloped]{$\ob{t_2}$} (t11);
\end{tikzpicture} 
}  
\caption{$\cmmm$-consistent}
\label{subfig:cm_not_ra}
\end{subfigure}
\hfill
\begin{subfigure}[b]{0.3\textwidth}
\centering
\scalebox{0.9}{
\begin{tikzpicture}[yscale=1]
\def\ystep{0.8}
\def\xstep{2}
\def\ybias{0.6}
	\node (t10) at (0,0*\ystep+\ybias)  {\large $t_1$};
  \node (t11) at (0*\xstep,0*\ystep)  {$\wt(z)$};
  \node (t12) at (0*\xstep,-1*\ystep) {$\wt(z)$};
  \node (t13) at (0*\xstep,-2*\ystep) {$\wt(x)$};
  \node (t14) at (0*\xstep,-3*\ystep) {$\wt(y)$};
	\node (t20) at (1*\xstep,0*\ystep+\ybias)  {\large $t_2$};
  \node (t21) at (1*\xstep,0*\ystep) {$\wt(x)$};
  \node (t22) at (1*\xstep,-1*\ystep) {$\rd(z)$};
  \node (t23) at (1*\xstep,-2*\ystep) {$\rd(y)$};
  \node (t24) at (1*\xstep,-3*\ystep) {$\rd(x)$};
  \draw[po] (t11) to (t12) to (t13) to (t14);
  \draw[po] (t21) to (t22) to (t23) to (t24);
\draw[rf,bend left=0] (t11) to node[above,pos=0.5,]{$\rf$} (t22);
\draw[rf,bend left=0] (t14) to node[above,pos=0.5,]{$\rf$} (t23);
\draw[rf,out=-60, in=60] (t21) to node[right,pos=0.5,]{$\rf$} (t24);
\draw[ob,bend left=0] (t13) to node[above,pos=0.5, sloped]{$\ob{t_2}$} (t21);
\draw[ob, out=160, in=-160, looseness=3] (t12) to node[left,pos=0.5]{$\ob{t_2}$} (t11);
\end{tikzpicture} 
}  
\caption{$\cmmm$-inconsistent}
\label{subfig:sra_not_cm}
\end{subfigure}
\hfill
\begin{subfigure}[b]{0.4\textwidth}
\centering
\scalebox{0.9}{
\begin{tikzpicture}[yscale=1]
\def\ystep{1.2}
\def\xstep{1.8}
\def\ybias{0.6}

\node (t10) at (0,0*\ystep+\ybias)  {\large $t_1$};
\node (t11) at (0*\xstep,0*\ystep)  {};
\node (t12) at (0*\xstep,-1*\ystep)  {$\wt'(x)$};
\node (t13) at (0*\xstep,-2*\ystep)  {};

\node (t20) at (1*\xstep,0*\ystep+\ybias)  {\large $t_2$};
\node (t21) at (1*\xstep,0*\ystep)  {};
\node (t22) at (1*\xstep,-1*\ystep)  {$\wt(x)$};
\node (t23) at (1*\xstep,-2*\ystep)  {};

\node (t30) at (2*\xstep,0*\ystep+\ybias)  {\large $t_3$};
\node (t31) at (2*\xstep,0*\ystep)  {};
\node (t3m) at (2*\xstep,-0.3*\ystep)  {};
\node (t32) at (2*\xstep,-1*\ystep)  {$\rd(x)$};
\node (t33) at (2*\xstep,-2*\ystep)  {};

\draw[po] (t11) to (t12) to (t13);
\draw[po] (t21) to (t22) to (t23);
\draw[po] (t31) to (t32) to (t33);

\draw[rf,bend left=0] (t22) to node[above,pos=0.5,]{$\rf$} (t32);
\draw[ob,bend left=0] (t12) to node[above,pos=0.5,]{$\ob{t_3}$} (t22);
\draw[hb,out=45, in=-180] (t12) to node[above,pos=0.8,]{$\hb$} (t3m);
\end{tikzpicture}   
} 
\caption{\label{subfig:ob}
The argument of \cref{lem:cm_ob}
}
\end{subfigure}
\caption{\label{fig:cm}
(\subref{subfig:cm_not_ra}) An execution consistent in $\cmmm$, as each of $\ob{t_1}$ and $\ob{t_2}$ is acyclic, but inconsistent in $\sramm$.
(\subref{subfig:sra_not_cm}) An execution inconsistent in $\cmmm$, as $\ob{t_2}$ creates a cycle, but consistent in $\sramm$.
(\subref{subfig:ob}) Illustration of the argument behind \cref{lem:cm_ob}.
Since $\wt'(x)$ reaches an earlier event in $t_3$ than $\wt(x)$, the edge $\wt'(x)\LTo{\ob{t_3}}\wt(x)$ does not create a new path from $\wt'(x)$ to $t_3$.
}
\end{figure}
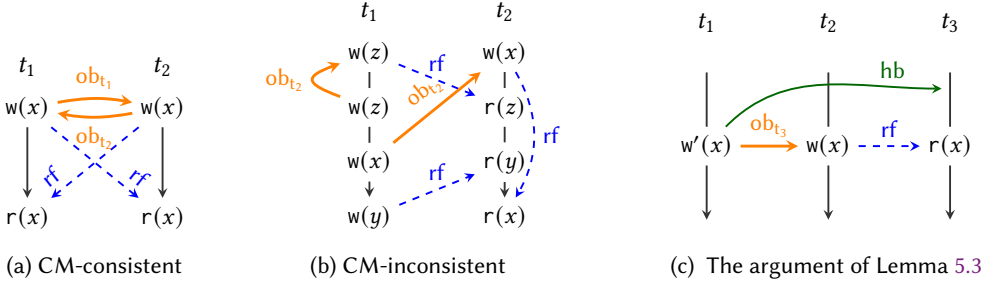

%% file: power.tex
\subsection{Implications for POWER}\label{SUBSEC:POWER}
The memory model of the $\powermm$ architecture is defined on load, store, atomic read-modify-write memory accesses, and various types of fences. 
$\powermm$ orders memory accesses based on fences, address, data, and control dependencies, while, again, coherence forces a total order on same-location accesses. 
In addition, $\powermm$ defines two global orderings, namely, \emph{happens-before}, and \emph{propagation}. 
The happens-before relation is based on dependencies, fences, and the $\rf$ relation across threads. 
The propagation relation captures the propagation of read and written values by combining fences, happens-before, $\rf$, and $\mo$. 
Based on these relations $\powermm$ defines its consistency axioms, which we will not present here; instead, we refer the interested readers to~\cite{Alglave:2014}.

\citet{Lahav:2016} showed that $\sramm$ captures precisely the guarantees of $\powermm$ for programs that are compiled from the release-acquire fragment of C/C++.
In turn, this implies that the result established in \cref{SEC:LOWER_THREADS_LOCATIONS_WRA_RA_SRA} for $\sramm$ transfers over to $\powermm$. 
We thus have the following corollary.

\corpower*

%% file: tso_pso.tex
\subsection{What About x86-TSO and PSO?}\label{SUBSEC:IMPLICATIONS_TSO_PSO}

Our results so far prove strong hardness for testing a variety of weak-memory models.
In contrast, in this section, we outline that the problem is solvable in polynomial time for x86-TSO and PSO.
Conceptually, this is less surprising for $\tsomm$, which diverges only a little from $\scmm$, but is more so for $\psomm$, which allows for behaviors that are not even causally consistent.

\input{figures/tso_pso_examples}

\Paragraph{Total Store Order.}
The $\tsomm$ model deviates from $\scmm$ by introducing a write-buffer for each thread, which acts as a FIFO queue~\cite{Sewell:2010}.
When a thread $t$ executes a write $\wt(t,x)$, this does not modify the shared memory immediately and is thus not visible to the other threads.
Instead, $\wt(t,x)$ is stored in the buffer of $t$.
The buffer non-deterministically flushes some of its writes to the shared memory, at which point they become visible to the other threads.
On the other hand, when a thread $t$ executes a read $\rd(t,x)$, it is forced to read from the most recent write $\wt(t,x)$ in $t$'s buffer.
If no such write exists then $\rd(t,x)$  reads from the shared memory.
See \cref{subfig:tso} for an illustration.

For capturing the complexity of consistency-testing of an abstract execution $\expartial=(\E, \po)$ under $\tsomm$, it is helpful to switch to operational semantics.
The semantics are defined by means of a labeled transition system $\LTS_{\tsomm}$; as this is a standard model, we defer the description of the formal model to~\cref{SUBSEC:APP_IMPLICATIONS_TSO_PSO}.
In high level, a state in $\LTS_{\tsomm}$ is a triplet $\langle \Executed, \Buffers, \Memory\rangle$, where
\begin{compactenum}
\item $\Executed\subseteq \E$ is the set of events that have been executed so far.
\item $\Buffers\colon \Threads\to (\W)^*$ maps every thread $t$ to a sequence of writes $\wt(t,x_1),\wt(t,x_2),\dots, \wt(t,x_i)$, which represents the state of the buffer of thread $t$.
\item $\Memory\colon \Locations\to\W$ maps every memory location of the shared memory to the most recent write to it.
\end{compactenum}
A counting argument shows that the size of $\LTS_{\tsomm}$ is bounded by $\NumThreads^{\NumLocations}\cdot \NumEvents^{O(\NumThreads^{2})}$, for $\NumEvents=|\E|$, $\NumThreads$ threads, $\NumLocations$ locations, and thus becomes polynomial when $\NumThreads, \NumLocations=O(1)$.
We refer to \cref{SUBSEC:APP_IMPLICATIONS_TSO_PSO} for details.

\Paragraph{Partial Store Order (PSO).}
The $\psomm$ model~\cite{pso} is similar to $\tsomm$, with the difference that every thread has a different buffer for each location. 
This allows both write-read reorderings (like $\tsomm$) and write-write reorderings on different locations. 
This induces more behaviors than $\tsomm$, but is incomparable to some other models like $\wramm$ (see \cref{subfig:pso} and \cref{subfig:not_tso_pso}).
The operational semantics can be defined by means of an LTS $\LTS_{\psomm}$ analogously to $\LTS_{\tsomm}$.
A similar analysis shows that the size of $\LTS_{\psomm}$ is bounded by $\NumEvents^{O(\NumThreads(\NumThreads+\NumLocations))}$.
Hence we have the following theorem, which differentiates $\tsomm$/$\psomm$ from the other weak-memory models we have seen so far.

\thmuppertsopso*

%% file: figures/tso_pso_examples.tex
\begin{figure}
\centering
\begin{subfigure}[b]{0.3\textwidth}
\centering
\scalebox{0.9}{
\begin{tikzpicture}[yscale=1]
\def\ystep{0.9}
\def\xstep{2}
\def\ybias{0.6}
\node (t10) at (0*\xstep,0*\ystep+\ybias)  {\large $t_1$};
\node (t11) at (0*\xstep,0*\ystep)  {$\wt(x)$};
\node (t12) at (0*\xstep,-1*\ystep) {$\wt(x)$};
\node (t13) at (0*\xstep,-2*\ystep) {$\rd(y)$};

\node (t20) at (1*\xstep,0*\ystep+\ybias)  {\large $t_2$};
\node (t21) at (1*\xstep,0*\ystep)  {$\wt(y)$};
\node (t22) at (1*\xstep,-1*\ystep) {$\wt(y)$};
\node (t23) at (1*\xstep,-2*\ystep) {$\rd(x)$};

\draw[po] (t11) to (t12) to (t13);
\draw[po] (t21) to (t22) to (t23);
\draw[rf,bend left=0] (t11) to node[above,pos=0.9, sloped]{$\rf$} (t23);
\draw[rf,bend left=0] (t21) to node[above,pos=0.9, sloped]{$\rf$} (t13);
\end{tikzpicture} 
}  
\caption{$\tsomm$-consistent}
\label{subfig:tso}
\end{subfigure}
\hfill
\begin{subfigure}[b]{0.3\textwidth}
\centering
\scalebox{0.9}{
\begin{tikzpicture}[yscale=1]
\def\ystep{0.9}
\def\xstep{2}
\def\ybias{0.6}
\node (t10) at (0*\xstep,0*\ystep+\ybias)  {\large $t_1$};
\node (t11) at (0*\xstep,0*\ystep)  {$\wt(x)$};
\node (t12) at (0*\xstep,-1*\ystep) {$\wt(x)$};
\node (t13) at (0*\xstep,-2*\ystep) {$\wt(y)$};

\node (t20) at (1*\xstep,0*\ystep+\ybias)  {\large $t_2$};
\node (t21) at (1*\xstep,0*\ystep)  {$\rd(y)$};
\node (t22) at (1*\xstep,-1*\ystep) {$\rd(x)$};

\draw[po] (t11) to (t12) to (t13);
\draw[po] (t21) to (t22);
\draw[rf,bend left=0] (t11) to node[above,pos=0.3, sloped]{$\rf$} (t22);
\draw[rf,bend left=0] (t13) to node[above,pos=0.3, sloped]{$\rf$} (t21);
\end{tikzpicture} 
}  
\caption{$\psomm$-consistent}
\label{subfig:pso} 
\end{subfigure}
\hfill
\begin{subfigure}[b]{0.3\textwidth}
\centering
\scalebox{0.9}{
\begin{tikzpicture}[yscale=1]
\def\ystep{1.8}
\def\xstep{2}
\def\ybias{0.6}
	\node (t10) at (0,0*\ystep+\ybias)  {\large $t_1$};
  \node (t11) at (0,0*\ystep)  {$\wt(x)$};
  \node (t12) at (0,-1*\ystep) {$\rd(x)$};
	\node (t20) at (1*\xstep,0*\ystep+\ybias)  {\large $t_2$};
  \node (t21) at (1*\xstep,0*\ystep) {$\wt(x)$};
  \node (t22) at (1*\xstep,-1*\ystep) {$\rd(x)$};
  \draw[po] (t11) to (t12);
  \draw[po] (t21) to (t22);
\draw[rf,bend left=0] (t11) to node[above,pos=0.9, sloped]{$\rf$} (t22);
\draw[rf,bend left=0] (t21) to node[above,pos=0.9, sloped]{$\rf$} (t12);
\end{tikzpicture} 
}  
\caption{$\tsomm$/$\psomm$-inconsistent}
\label{subfig:not_tso_pso}
\end{subfigure}
\caption{\label{fig:tso_pso}
(\subref{subfig:tso}) An execution consistent in $\tsomm$, as well as in $\psomm$ and all other memory models we have considered, but not $\scmm$.
(\subref{subfig:pso}) An execution consistent in $\psomm$, as well as in $\arlxmm$ but not in $\tsomm$ or even $\wramm$.
(\subref{subfig:not_tso_pso}) An execution inconsistent in $\tsomm$/$\psomm$, but consistent in $\ccmm$/$\wramm$.
}
\end{figure}

%% file: rlx.tex
\subsection{A Final Note on Relaxed}\label{SUBSEC:RLX}

Finally, we turn our attention to the $\rlxmm$ model.
The only two axioms of this model are \RlxWCoh{} and \RlxRCoh{}, which concern individual memory locations,
and guarantee per-location coherence, i.e., 
focusing on each location individually, the corresponding execution is $\scmm$-consistent.
Given an abstract execution $\expartial$, to decide whether $\expartial\models \rlxmm$, it suffices to check whether $\expartial_x\models \scmm$ for each location $x$,
where $\expartial_x$ occurs from $\expartial$ by considering only events accessing $x$.
As each consistency check $\expartial_x\models \scmm$ takes polynomial time~\cite{Agarwal2021},
and we clearly have polynomially many such checks, we arrive at \cref{cor:rlx}.

\corrlx*

%% file: bounded_values.tex

\section{Hardness with Bounded Values}\label{SEC:BOUNDED_VALUES}

For ease of presentation, our reductions in \cref{SEC:LOWER_THREADS_LOCATIONS_RLX} and \cref{SEC:LOWER_THREADS_LOCATIONS_WRA_RA_SRA} use a bounded number of threads and memory locations but an unbounded value domain.
Indeed, given the Boolean formula $\varphi=\{\Clause_i\}_{i\in[m]}$ on $m$ clauses and $n$ variables, the value domain of the abstract execution $\expartial$ has size $\Theta(n\cdot m)$.
In this section we outline how to modify those reductions so that $\expartial$ also uses a bounded value domain, thereby arriving at \cref{thm:lower_ra} and \cref{thm:lower_cc}.

\Paragraph{Intuition.}
Our two reductions are such that every read $\rd$ can read from at most three writes, and these appear in the same gadget as $\rd$.
However, the values of these events are specific to the gadget, and in particular, 
specific to the phase $i$ and the step $j$ of the events (i.e., events are of the form $\ReadXOne{t_3, v^i_j}$).
Our strategy for decreasing the size of the value domain (of both reductions in \cref{SEC:LOWER_THREADS_LOCATIONS_WRA_RA_SRA} and \cref{SEC:LOWER_THREADS_LOCATIONS_RLX}) is by using repeating values which are not parameterized by the superscript $i$ and subscript $j$ 
(i.e., the events in the executions constructed now look like $\ReadXOne{t_3, v}$ or $\WriteXOne{t_1, v}$).
This change does not affect completeness but threatens soundness, as now,
some read events may read from write events in other gadgets that were 
previously forbidden simply because their values were not matching.
To avoid this, we slightly modify our abstract executions $\expartial$ by inserting a bounded number of auxiliary write and read events between consecutive gadgets, which also access a bounded number of values.
The auxiliary write events write dummy values read by the auxiliary read events. 
The effect of these additional reads-from edges due to auxiliary events
is to create $(\po_x\cup\rf_x)^+$ paths that once again forbid the 
original (i.e., non-auxiliary) read events of a gadget to 
access write events from other gadgets (while obeying the desired consistency axioms).

\Paragraph{Construction.}
We now outline the construction.
The process is similar for both the abstract executions of \cref{SEC:LOWER_THREADS_LOCATIONS_RLX} and \cref{SEC:LOWER_THREADS_LOCATIONS_WRA_RA_SRA}.
For this reason, we describe it generically on an abstract execution $\expartial$.
Our transformation is carried out in two steps, $\expartial\LPath{~} \expartial_1 \LPath{~}\expartial_2$, where $\expartial_1$ and $\expartial_2$ have the same number of threads and locations as $\expartial$, and $\expartial_2$ additionally has a bounded value domain.

\SubParagraph{Step 1.}
We obtain $\expartial_1$ by inserting various events in $\expartial$ while keeping the threads and memory locations the same.
We start by fixing a total order $\sigma_1$ on locations, and a total order $\sigma_2$ on threads.
\begin{align*}
\sigma_1&=\VarXOne, \VarXTwo, \VarYOne,\dots, \VarYEight, \VarZOne,\dots, \VarZEight, \VarA{1},\dots, \VarA{3}, \VarB\\
\sigma_2&=t_1,\dots,t_6, f_1, \dots, f_6, g_1, \dots, g_6, h_{1},\dots, h_{3}, p, q
\end{align*}

Fix a phase $i$ and step $j$ and let $\zeta=(i*m+j)\mod 2$.
For a location $x$ of $\expartial$, different from $\VarA{1},\VarA{2}, \VarA{3}, \VarB$,
we introduce auxiliary write and read events on $x$ as follows:
(i)~if a thread $t$ writes on $x$,we insert a write $\wt(t,x,v^{t}_{\zeta})$ after all events of phase $i$ and step $j$ in $t$, and
(ii)~if a thread $t$ reads from $x$, we insert a sequence of read events $\rd(t,x,v^{t^1}_{\zeta}),\rd(t,x,v^{t^2}_{\zeta}),\dots$ before all events of phase $i$ and step $j+1$ (or phase $i+1$ and step $1$, if $j=n$), where $t^1, t^2,\dots$ is the subsequence of $\sigma_2$ of threads writing to $x$ values read by thread $t$.
We repeat this process for all locations $x \not\in \set{\VarA{1}, \VarA{2}, \VarA{3}, \VarB}$ 
in the order of appearance in the total order $\sigma_1$, placing the auxiliary writes before the auxiliary reads in each thread.
Observe that each $\rd(t,x,v^{t^{\ell}}_{\zeta})$ event is forced to read from the respective $\wt(t^{\ell},x,v^{t^{\ell}}_{\zeta})$.

Next, we turn our attention to the locations $\VarA{1}, \VarA{2}, \VarA{3}$ and $\VarB$.
The auxiliary events are positioned similarly, except for the detail about the step number $j$, 
because accesses to these locations span an entire phase (in the at-most-one-true and at-least-one-true gadgets).
In particular, we have $\zeta=i\mod 2$, while auxiliary write events are placed in each thread after all events of phase $i$, and read events are placed before events of phase $i+1$.

Observe that since the number of threads and locations is bounded in $\expartial$, the same holds for $\expartial_1$, while the additional values accessed by the auxiliary events in $\expartial_1$ are also bounded.

\SubParagraph{Step 2.}
In the second step, we transform $\expartial_1$ to $\expartial_2$ so that the latter only accesses a bounded number of values.
In particular, we make $\expartial_2$ identical to $\expartial_1$ with the difference that, for every event of $\expartial_1$ that also appears in $\expartial$ (i.e., non-auxiliary events), 
we remove from its value the superscript of the phase and the subscript of the step of that event.
For example, each write $\TWriteXOne{v^i_j}{t_1}$ in $\expartial_1$ becomes $\TWriteXOne{v}{t_1}$ in $\expartial_2$.
It is straightforward to verify that $\expartial_2$ has a bounded domain of threads, locations, and values.
Moreover, $\expartial_2$ is consistent in the respective memory model iff $\expartial$ is, by repeating the arguments in \cref{SEC:LOWER_THREADS_LOCATIONS_RLX} and \cref{SEC:LOWER_THREADS_LOCATIONS_WRA_RA_SRA}, this time also accounting for the auxiliary events.

%% file: conclusion.tex
\section{Conclusion}\label{sec:conclusion}
We have studied the standard problem of consistency-testing for various popular weak-memory models spanning across software, hardware, and distributed systems.
We have shown that even the \emph{bounded} version of consistency testing is $\NP$-complete in most of these models, i.e., when every natural input parameter is bounded.
This is a significant improvement over an abundance of prior hardness results which primarily stemmed from parameters such as the number of threads or memory locations being unbounded. 
Our results thus highlight the true intricacies of weak-memory testing.
In particular, our results imply that the problem provably admits no parameterization with respect to natural input parameters.
Interesting future work includes the possibility of extending our hardness to other memory models such as the one in ARM architectures, 
as well as recovering tractability by imposing further restrictions (such as context/view-switching).

%% file: app_lower_ra.tex
\section{Proofs from \cref{SEC:LOWER_THREADS_LOCATIONS_WRA_RA_SRA}}\label{SEC:APP_LOWER_THREADS_LOCATIONS_WRA_RA_SRA}
\lemmawrarasrasoundnesscopy*
\begin{proof}
We argue by induction that for every $i\in[m-1]$, we have $(\TWriteXOne{v^{i}_j}{t_1}, \TReadXOne{v^{i}_j}{t_3}) \in \rf$ iff $(\TWriteXOne{v^{i+1}_j}{t_1}, \TReadXOne{v^{i+1}_j}{t_3}) \in \rf$.
First, note that if $(\TWriteXOne{v^{i}_j}{t_1}, \TReadXOne{v^{i}_j}{t_3}) \in \rf$, then the copy gadget $\CopyGadget{i}{j}$ forces $(\TWriteXTwo{v^{i}_j}{t_4}, \TReadXTwo{v^{i}_j}{t_6}) \in \rf$. 
Indeed, we have the following forced sequence of $\rf$ edges (see \cref{subfig:copy_gadget_wra_ra_sraFalse}). 
We begin the sequence with the first edge $(\TWriteXOne{v^{i}_j}{t_1}, \TReadXOne{v^{i}_j}{t_3}) \in \rf$, marked \circled{1} in \cref{subfig:copy_gadget_wra_ra_sraFalse}.
\begin{compactenum}
\item We have $(\TWriteYOne{\ov{v}^i_j}{t_1},\TReadYOne{v^i_j}{t_3})\in \hb$ and thus $(\TWriteYOne{v^i_j}{t_1}, \TReadYOne{v^i_j}{t_3})\not \in \rf$ due to \WRCoh. 
Thus $\TReadYOne{v^i_j}{t_3})$ is forced to read from the only other available write, i.e., $(\TWriteYOne{v^i_j}{f_4}, \TReadYOne{v^i_j}{t_3}) \in \rf$. 
This is depicted as \circled{2} in \cref{subfig:copy_gadget_wra_ra_sraFalse}. 
\item We have $(\TWriteYSeven{v^i_j}{f_5}, \TReadYSeven{v^i_j}{f_4})\not \in \rf$, otherwise we would have $(\TWriteYFour{\ov{v}^i_j}{f_5}, \TReadYFour{v^i_j}{t_3})\in \hb$, and since $(\TWriteYFour{v^i_j}{f_5}, \TReadYFour{v^i_j}{t_3})\in \rf$, we would have a violation of \WRCoh.
Thus $\TReadYSeven{v^i_j}{f_4}$ is forced to read from the only other available write, i.e., $(\TWriteYSeven{v^i_j}{f_6}, \TReadYSeven{v^i_j}{f_4}) \in \rf$.
This is depicted as \circled{3} in \cref{subfig:copy_gadget_wra_ra_sraFalse}. 

\item We have $(\TWriteXTwo{v^i_j}{t_5}, \TReadXTwo{v^i_j}{t_6})\not \in \rf$, otherwise we would have $(\TWriteYEight{\ov{v}^i_j}{t_5}, \TReadYEight{v^i_j}{f_4})\in \hb$, and since $(\TWriteYEight{v^i_j}{t_5}, \TReadYEight{v^i_j}{f_4})\in \rf$, we would have a violation of weak-read coherence.
Thus $\TReadXTwo{v^i_j}{t_6}$ is forced to read from the only other available write, i.e., $(\TWriteXTwo{v^i_j}{t_4}, \TReadXTwo{v^i_j}{t_6}) \in \rf$. 
This is depicted as \circled{4} in \cref{subfig:copy_gadget_wra_ra_sraFalse}. 
\end{compactenum}

On the other hand, if $(\TWriteXOne{v^{i}_j}{t_2}, \TReadXOne{v^{i}_j}{t_3}) \in \rf$, then the copy gadget $\CopyGadget{i}{j}$ forces $(\TWriteXTwo{v^{i}_j}{t_5}, \TReadXTwo{v^{i}_j}{t_6}) \in \rf$, by an analysis very similar to the one above (see \cref{subfig:copy_gadget_wra_ra_sraTrue}).
Starting with $(\TWriteXOne{v^{i}_j}{t_2}, \TReadXOne{v^{i}_j}{t_3}) \in \rf$ which is marked \circled{1} 
in \cref{subfig:copy_gadget_wra_ra_sraTrue}, the sequence of forced $\rf$ edges in order are marked \circled{2}, \circled{3} leading 
to \circled{4} which is $(\TWriteXTwo{v^{i}_j}{t_5}, \TReadXTwo{v^{i}_j}{t_6}) \in \rf$.

Finally, a similar analysis on the copy-down gadget $\CopyGadgetDown{i}{j}$ (see \cref{fig:copy_gadget_down_wra_ra_sra}) concludes that $(\TWriteXOne{v^{i+1}_j}{t_1}, \TReadXOne{v^{i+1}_j}{t_3}) \in \rf$ iff $(\TWriteXTwo{v^{i}_j}{t_4}, \TReadXTwo{v^{i}_j}{t_6}) \in \rf$,
and hence we have $(\TWriteXOne{v^{i}_j}{t_1}, \TReadXOne{v^{i}_j}{t_3}) \in \rf$ iff $(\TWriteXOne{v^{i+1}_j}{t_1}, \TReadXOne{v^{i+1}_j}{t_3}) \in \rf$, as desired.
The sequence of inferred $\rf$-edges in order are highlighted \circled{1} through \circled{4} in \cref{subfig:copy_gadget_down_wra_ra_sraFalse}(a). 
Likewise, \cref{subfig:copy_gadget_down_wra_ra_sraTrue}(b) depicts through 
the sequence of $\rf$-edges highlighted  \circled{1} - \circled{4}
that, 
$(\TWriteXOne{v^{i+1}_j}{t_2}, \TReadXOne{v^{i+1}_j}{t_3}) \in \rf$ iff $(\TWriteXTwo{v^{i}_j}{t_5}, \TReadXTwo{v^{i}_j}{t_6}) \in \rf$,
and hence we have $(\TWriteXOne{v^{i}_j}{t_2}, \TReadXOne{v^{i}_j}{t_3}) \in \rf$ iff $(\TWriteXOne{v^{i+1}_j}{t_2}, \TReadXOne{v^{i+1}_j}{t_3}) \in \rf$, as desired.

\end{proof}

\lemwrarasrasoundnessatmostone*
\begin{proof}
The statement follows by analyzing the at-most-one-true gadget $\AtMostOneGadget{c}{i}{j}{k}$, where $c\in[3]$ is such that $(s_j, s_k)$ is the $c$-th pair of variables in $\Clause_i$ (see \cref{fig:at_most_one_gadget_wra_ra_sra}).

First, we argue that if $(\TWriteXOne{v^i_j}{t_2}, \TReadXOne{v^i_j}{t_3}) \in \rf$ then $(\TWriteXOne{v^i_k}{t_2}, \TReadXOne{v^i_k}{t_3})\not \in \rf$.
Indeed, we have the following forced sequence of $\rf$ edges (see \circledsmall{1}-\circledsmall{4} in \cref{subfig:at_most_one_gadget_wra_ra_sra_true_false}
starting with $(\TWriteXOne{v^i_j}{t_2}, \TReadXOne{v^i_j}{t_3}) \in \rf$ marked \circledsmall{1}).
\begin{compactenum}
\item We have $(\TWriteA{\ov{v}^i_j}{c}{t_2}, \TReadA{v^i_j}{c}{t_3})\in \hb$, and thus $(\TWriteA{v^i_j}{c}{t_2}, \TReadA{v^i_j}{c}{t_3})\not \in \rf$, otherwise we would have a violation of weak-read-coherence.
Hence $\TReadA{v^i_j}{c}{t_3})$ is forced to read from the only other available write, i.e., $(\TWriteA{v^i_j}{c}{h_{c}}, \TReadA{v^i_j}{c}{t_3}) \in \rf$.
This is marked \circledsmall{2} in \cref{subfig:at_most_one_gadget_wra_ra_sra_true_false}.  
\item We now have $(\TWriteA{v^i_j}{c}{h_{c}}, \TReadA{v^i_k}{c}{t_3})\in \hb$, and thus $(\TWriteA{v^i_k}{c}{h_{c}}, \TReadA{v^i_k}{c}{t_3})\not \in \rf$, otherwise we would have a violation of weak-read-coherence.
Thus $\TReadA{v^i_k}{c}{t_3}$ is forced to read from the only other available write, i.e., $(\TWriteA{v^i_k}{c}{t_2}, \TReadA{v^i_k}{c}{t_3})\in \rf$.
This is marked \circledsmall{3} in \cref{subfig:at_most_one_gadget_wra_ra_sra_true_false}.

\item In turn, we have $(\TWriteXOne{v^i_k}{t_2}, \TReadXOne{v^i_k}{t_3})\not \in \rf$, as otherwise we would have $(\TWriteA{\ov{v}^i_2}{c}{t_2}, \TReadA{v^i_k}{c}{t_3})\in \hb$, which would violate weak-read coherence.
Hence we have $(\TWriteXOne{v^i_k}{t_1}, \TReadXOne{v^i_k}{t_3}) \in \rf$,  marked \circledsmall{4} in \cref{subfig:at_most_one_gadget_wra_ra_sra_true_false}.

\end{compactenum}

Second, we argue that if $(\TWriteXOne{v^i_k}{t_2}, \TReadXOne{v^i_k}{t_3}) \in \rf$ then $(\TWriteXOne{v^i_j}{t_2}, \TReadXOne{v^i_j}{t_3})\not \in \rf$.
Indeed, we have the following forced sequence of $\rf$ edges (marked \circledsmall{1}-\circledsmall{4}, see \cref{subfig:at_most_one_gadget_wra_ra_sra_false_true}, starting with $(\TWriteXOne{v^i_k}{t_2}, \TReadXOne{v^i_k}{t_3}) \in \rf$ marked 
\circledsmall{1}).
\begin{compactenum}
\item We have $(\TWriteA{\ov{v}^i_k}{c}{t_2}, \TReadA{v^i_k}{c}{t_3})\in \hb$, and thus $(\TWriteA{v^i_k}{c}{t_2}, \TReadA{v^i_k}{c}{t_3})\not \in \rf$, otherwise we would have a violation of weak-read-coherence.
Thus $\TReadA{v^i_k}{c}{t_3})$ is forced to read from the only other available write, i.e., $(\TWriteA{v^i_k}{c}{h_{c}}, \TReadA{v^i_k}{c}{t_3}) \in \rf$. 
This is marked \circledsmall{2} in \cref{subfig:at_most_one_gadget_wra_ra_sra_false_true}. 

\item We have $(\TWriteA{v^i_j}{c}{h_{c}}, \TReadA{v^i_j}{c}{t_3})\not \in \rf$, as otherwise we would have $(\TWriteA{v^i_j}{c}{h_{c}}, \TReadA{v^i_k}{c}{t_3})\in \hb$ which would violate weak-read-coherence, given that 
$(\TWriteA{v^i_k}{c}{h_{c}}, \TReadA{v^i_k}{c}{t_3}) \in \rf$ and 
$(\TWriteA{v^i_k}{c}{h_{c}}, \TWriteA{v^i_j}{c}{h_\ell}) \in \po \subseteq \hb$. 
Thus $\TReadA{v^i_j}{c}{t_3}$ is forced to read from the only other available write, i.e., $(\TWriteA{v^i_j}{c}{t_2}, \TReadA{v^i_j}{c}{t_3}) \in \rf$, see \circledsmall{3} in  \cref{subfig:at_most_one_gadget_wra_ra_sra_false_true}.

\item In turn, we have $(\TWriteXOne{v^i_j}{t_2}, \TReadXOne{v^i_j}{t_3})\not \in \rf$, as otherwise we would have 
$(\TWriteA{\ov{v}^i_j}{c}{t_2}, \TReadA{v^i_j}{c}{t_3})\in \hb$, which would violate weak-read coherence. 
Thus, $(\TWriteXOne{v^i_j}{t_1}, \TReadXOne{v^i_j}{t_3}) \in \rf$, depicted by \circledsmall{4} in  \cref{subfig:at_most_one_gadget_wra_ra_sra_false_true}.
\end{compactenum}
\end{proof}

\lemwrarasrasoundnessatleastone*
\begin{proof}
The statement follows by analyzing the at-least-one-true gadget $\AtLeastOneGadget{i}{j}{k}{\ell}$ (see \cref{fig:at_least_one_gadget_wra_ra_sra}).
We assume wlog that $\Clause_i=(s_j, s_k, s_{\ell})$, i.e., the three variables appear in $\Clause_i$ in this order.

First, we argue that if $(\TWriteXOne{v^i_j}{t_2}, \TReadXOne{v^i_j}{t_3})\not \in \rf$ and $(\TWriteXOne{v^i_k}{t_2}, \TReadXOne{v^i_k}{t_3})\not \in \rf$ then $(\TWriteXOne{v^i_{\ell}}{t_2}, \TReadXOne{v^i_{\ell}}{t_3}) \in \rf$.
Indeed, we have the following forced sequence of $\rf$ edges (see \circledsmall{1}-\circledsmall{4}, \cref{subfig:at_least_one_gadget_wra_ra_sra_false_false_true}. We begin with 
$(\TWriteXOne{v^i_j}{t_1}, \TReadXOne{v^i_j}{t_3})\in \rf$ and $(\TWriteXOne{v^i_k}{t_1}, \TReadXOne{v^i_k}{t_3}) \in \rf$
both marked \circledsmall{1}).
\begin{compactenum}
\item We have $(\TWriteB{\ov{v}^i_j}{t_1}, \TReadB{v^i_j}{t_3})\in \hb$, and thus $(\TWriteB{v^i_j}{t_1}, \TReadB{v^i_j}{t_3})\not \in \rf$, otherwise we would have a violation of weak-read-coherence.
Thus $\TReadB{v^i_j}{t_3}$ is forced to read from the only other available write, i.e., $(\TWriteB{v^i_j}{p}, \TReadB{v^i_j}{t_3})\in \rf$.
 Similarly, we have $(\TWriteB{\ov{v}^i_k}{t_1}, \TReadB{v^i_k}{t_3})\in \hb$, and thus $(\TWriteB{v^i_k}{t_1}, \TReadB{v^i_k}{t_3})\not \in \rf$, otherwise we would have a violation of weak-read-coherence. 
Thus $\TReadB{v^i_k}{t_3}$ is forced to read from the only other available write, i.e., $(\TWriteB{v^i_k}{q}, \TReadB{v^i_k}{t_3})\in \rf$. 
These two edges are marked \circledsmall{2}.
\item We now have $(\TWriteB{v^i_j}{p}, \TReadB{v^i_{\ell}}{t_3})\in \hb$, and thus $(\TWriteB{v^i_{\ell}}{p}, \TReadB{v^i_{\ell}}{t_3})\not \in \rf$, otherwise we would have a violation of weak-read coherence. 
Similarly, we have $(\TWriteB{v^i_k}{q}, \TReadB{v^i_{\ell}}{t_3})\in \hb$, and thus $(\TWriteB{v^i_{\ell}}{q}, \TReadB{v^i_{\ell}}{t_3})\not \in \rf$, otherwise we would have a violation of weak-read-coherence.
Thus $\TReadB{v^i_{\ell}}{t_3}$ is forced to read from the only other available write, i.e., $(\TWriteB{v^i_{\ell}}{t_1}, \TReadB{v^i_{\ell}}{t_3})\in \rf$.  
This is marked as \circledsmall{3}.
\item In turn, we have $(\TWriteXOne{v^i_{\ell}}{t_1}, \TReadXOne{v^i_{\ell}}{t_3})\not \in \rf$, as otherwise we would have $(\TWriteB{\ov{v}^i_{\ell}}{t_1}, \TReadB{v^i_{\ell}}{t_3})\in \hb$, which would violate weak-read-coherence. 
Thus we have $(\TWriteXOne{v^i_{\ell}}{t_2}, \TReadXOne{v^i_{\ell}}{t_3}) \in \rf$, marked \circledsmall{4}.
\end{compactenum}
A similar analysis establishes that if $(\TWriteXOne{v^i_j}{t_2}, \TReadXOne{v^i_j}{t_3})\not \in \rf$ and $(\TWriteXOne{v^i_{\ell}}{t_2}, \TReadXOne{v^i_{\ell}}{t_3})\not \in \rf$ then $(\TWriteXOne{v^i_k}{t_2}, \TReadXOne{v^i_k}{t_3}) \in \rf$ (see \cref{subfig:at_least_one_gadget_wra_ra_sra_false_true_false}), as well as that if $(\TWriteXOne{v^i_k}{t_2}, \TReadXOne{v^i_k}{t_3})\not \in \rf$ and $(\TWriteXOne{v^i_{\ell}}{t_2}, \TReadXOne{v^i_{\ell}}{t_3})\not \in \rf$ then $(\TWriteXOne{v^i_j}{t_2}, \TReadXOne{v^i_j}{t_3}) \in \rf$ (see \cref{subfig:at_least_one_gadget_wra_ra_sra_true_false_false}).
The desired result follows.
\end{proof}

\corwrarasralowerthreadslocationssoundness*

\lemwrarasracompletenessignoreevents*
\begin{proof}
We prove each item separately.
\begin{compactenum}
\item Any cycle containing a write $\wt\in \InactiveWrites$ must contain a sequence of edges $\event_1\LTo{\po}{\wt}\LTo{\po}\event_2$, which can be replaced by the single edge $\event_1\LTo{\po}\event_2$.
\item Any unsafe triplet $(\wt, \rd, \wt')$ requires the existence of an $\hb$-path $P\colon \wt' \LPath{\hb} \rd$.
Since the threads $\{h_{c}\}_{c\in [3]}$, $p$ and $q$ contain only same-location writes,  if $\wt' \in \InactiveWrites$, 
then $P$ must be of the form $P\colon \wt' \LPath{\po} \wt''\LTo{\rf}\rd''\LPath{\hb} \rd$, where, $\wt''$ and $\rd''$ conflict with $\wt'$.
Note that $\rd'' \neq \rd$, otherwise, by definition of a triplet, we will have $\wt''=\wt$ since $(\wt,\rd) \in \rf$. This would 
imply that $\wt' ~\po ~\wt$ making $(\wt, \rd, \wt')$ safe. 

Since $(\wt, \rd, \wt')$ is unsafe, we have $(\wt', \wt)\not \in (\hb\cup\mopartial)^+$ and hence, $(\wt'', \wt)\not \in (\hb\cup\mopartial)^+$. 
Thus $(\wt, \rd, \wt'')$ is also unsafe, while clearly $\wt''\not\in \InactiveWrites$.
\end{compactenum}
\end{proof}

\lemwrarasracompletenessmonotonichbpaths*
\begin{proof}
Observe that every $\rf$-edge connects events of the same phase and step.
Recall that, due to \cref{lem:wra_ra_sra_completeness_inactive_writes}, we can assume that the $\hb$-path between any two events does not contain any of the inactive writes $\InactiveWrites$.
Thus, for every $\po$-edge $\event_1\LTo{\po} \event_2$ between events in the threads $h_{c}$, $p$ or $q$, we have $\Step(\event_1)< \Step(\event_2)$ and thus $\event_1\OrderedBefore \event_2$.
Hence, the only $\po$-edge  $\event_1\LTo{\po} \event_2$ with $\event_2\StrictOrderedBefore\event_1$ is either in thread $g_1$ (where $\event_2=\TReadZSix{v^i_j}{g_1}$) or in thread $g_4$ (where $\event_2=\TReadZEight{v^i_j}{g_4}$).

Now consider any $\hb$-path $P\colon \wt_1\LPath{\hb}\wt_2$.
If $P$ is not monotonic, then it must contain a subpath $\event_1\LTo{\po}\event_2\LTo{\po}\event_3$, where $\event_2=\TReadZSix{v^i_j}{g_1}$ or  $\event_2=\TReadZEight{v^i_j}{g_4}$, for some $i\in[m]$ and $j\in[n]$, due to the analysis in the previous paragraph.
But then we can remove $\event_2$, obtaining a shorter path $\wt\LPath{\hb}\rd$.
Repeating this process, we end up with a monotonic $\hb$-path $\wt_1\LPath{\hb}\wt_2$.
\end{proof}

\lemwrarasracompletenessnocyclehbmo*
\begin{proof}
Consider any $(\hb\cup\mopartial)$-cycle $C$ that traverses only write events connected by $\hb$ and $\mopartial$.
By construction, whenever $(\wt_1, \wt_2)\in\mopartial$ we also have $\wt_1\OrderedBefore\wt_2$.
This fact, in conjunction with \cref{lem:wra_ra_sra_completeness_monotonic_hb_paths} implies that $C$ is monotonic.
Since $C$ is a cycle, it follows that it contains events of the same phase and step.
A direct inspection of each instantiation of each gadget concludes that no such cycle $C$ can exist.

In particular, regarding the copy gadgets $\CopyGadget{i}{j}$ and copy-down gadgets $\CopyGadgetDown{i}{j}$, we have the following.
\begin{compactenum}
\item Observe that $C$ does not cross $t_3$ as this thread has no writes.
\item In turn, $C$ does not cross any of $f_1$, $f_4$, $g_1$ and $g_4$, as each of these threads can only enter $t_3$.
\item In turn, $C$ does not cross any of $f_2$ and $f_3$, as  each of these threads can only enter $t_3$ and $f_1$.
Similarly, $C$ does not cross any of $f_5$ and $f_6$, as each of these threads can only enter $t_3$ and $f_4$.
Likewise, $C$ does not cross any of $g_2$ and $g_3$, as each of these threads can only enter $t_3$ and $g_1$.
Similarly, $C$ does not cross any of $g_5$ and $g_6$, as each of these threads can only enter $t_3$ and $g_4$.
\item In turn, $C$ does not cross $t_6$, as this thread can only enter $f_3$, $f_6$, $g_3$ and $g_6$.
\item Finally, $C$ does not cross any of $t_4$ and $t_5$, as these threads can only enter $t_6$, the $f$ threads and the $g$ threads.
\end{compactenum}

Having excluded the above threads, $C$ must be contained in $\{t_{\ell}\}_{\ell\in[2]}$, $\{h_{c}\}_{c\in [3]}$, $q$ and $p$.
Note that this implies that $C$ only consists of $\po$/$\mopartial$-edges.
A straightforward inspection of the at-most-one-true and at-least-one-true gadgets (see \cref{fig:at_most_one_gadget_wra_ra_sra} and \cref{fig:at_least_one_gadget_wra_ra_sra}) shows that no such cycle exists.
Thus $(\hb\cup\mopartial)$ is acyclic, as desired.
\end{proof}

\lemwrarasrarfimmediate*
\begin{proof}
Note that $\wt\OrderedBefore\wt'$, and thus $\rd\OrderedBefore\wt'$, as $\Phase(\wt)=\Phase(\rd)$ and $\Step(\wt)=\Step(\rd)$ by construction.
We first argue that any path $P\colon\wt'\LPath{\hb}\rd$ contains events of the same phase and step.
Indeed, as no location is ever written and read by the same thread, $P$ has the general form $P\colon \wt'\LPath{\hb^?}\wt''\LTo{\rf}\rd''\LTo{\po^?}\rd$ for some write $\wt''$ and read $\rd''$.
Due to \cref{lem:wra_ra_sra_completeness_monotonic_hb_paths}, the subpath $\wt'\LPath{\hb^?}\wt''$ is, without loss of generality monotonic, while the last two edges of $P$ are also monotonic by construction ($\rf$-edges are always monotonic, while the only non-monotonic $\po$-edges go from writes to reads).
Hence $P$ is monotonic.

On the other hand,  we have $\wt\OrderedBefore\wt'$ (as $(\wt, \wt')\in \po$), while, by construction, $\rd\OrderedBefore\wt$.
Hence $\rd\OrderedBefore\wt'$, and since $P$ is monotonic, it must contain only events of the same phase and step.
In particular, $P$ must be contained in one of the gadgets in \cref{fig:copy_gadget_wra_ra_sra,fig:copy_gadget_down_wra_ra_sra,fig:at_most_one_gadget_wra_ra_sra,fig:at_least_one_gadget_wra_ra_sra}.

The absence of such paths $P$ can then be established by a careful inspection of the gadgets.
The statement is then followed by a case-by-case analysis of each write $\wt$ in each thread and each gadget.
In particular:
\begin{compactenum}
\item For the threads $t_3$ and $t_6$ the statement follows trivially, as these contain no writes.
\item For all writes on threads $\{f_{\ell}, g_{\ell}\}_{\ell\in [6]}$, $\{h_c\}_{c\in[3]}$, $p$ and $q$, the statement follows by simply observing each of the possible instantiations of each gadget.
\item For all writes on threads $t_4$ and $t_5$, it again suffices to observe each of the possible instantiations of each gadget separately, as 
(i)~the writes on $\VarYSix$/$\VarYEight$ are read by the $f$ threads, 
(ii)~the writes on $\VarZSix$/$\VarZEight$ are read by the $g$ threads, and
(iii)~there are no $\hb$-paths between the $f$ threads and the $g$ threads.
\item Finally, for the threads $t_1$ and $t_2$, we have to consider all writes of a given phase and step that appear in different gadgets (e.g., $\TWriteYOne{v^i_j}{t_1}$, $\TWriteZOne{v^i_j}{t_1}$ etc.).
The statement follows by the fact that all these writes appear in $t_1$ and $t_2$ following the predefined order $\sigma$ on memory locations, and they are read in the same order by $t_3$.
\end{compactenum}
\end{proof}

\lemwrarasracompletenesssafetriplets*
\begin{proof}
Observe that $\Phase(\wt)=\Phase(\rd)$ and $\Step(\wt)=\Step(\rd)$ for all such triplets.
If $\wt' \StrictOrderedBefore \wt$, then by construction we have $(\wt', \wt)\in (\hb\cup \mo)^+$, thus the triplet is safe.
On the other hand, if $\wt\StrictOrderedBefore \wt'$,  then also $\rd\StrictOrderedBefore\wt'$ (as $\rd\OrderedBefore\wt$ by construction).
We argue that in this case there is no $\hb$-path $\wt'\LPath{\hb}\rd$.
Assume towards contradiction that such a path $P$ exists.
Since no location is both written and read by the same thread, $P$ has the form $P\colon \wt'\LPath{\hb}\wt''\LTo{\rf}\rd'' \LTo{\po?}\rd$.
Due to \cref{lem:wra_ra_sra_completeness_monotonic_hb_paths}, we have $\wt'\OrderedBefore\wt''$ and thus also $\wt'\OrderedBefore\rd''$.
Since $\rd\StrictOrderedBefore \wt'$, we have $\rd \StrictOrderedBefore \rd''$.
Observe that this can only happen if $\rd=\TReadZSix{v^i_j}{g_1}$ or $\rd=\TReadZEight{v^i_j}{g_4}$, for some $i\in[m]$ and $j\in [n]$.
However, the only thread that writes to $\VarZSix$ is $t_4$, and the only thread that writes to $\VarZEight$ is $t_5$.
In both cases we have $(\wt, \wt')\in \po$, and by \cref{lem:wra_ra_sra_rf_immediate}, we have $(\wt', \rd)\not \in \hb$.

Finally, we consider triplets where $\Phase(\wt)=\Phase(\wt')=\Phase(\rd)$ and $\Step(\wt)=\Step(\wt')=\Step(\rd)$.
Note that all events of such a triplet must exist together in one of the gadgets in \cref{fig:copy_gadget_wra_ra_sra,fig:copy_gadget_down_wra_ra_sra,fig:at_most_one_gadget_wra_ra_sra,fig:at_least_one_gadget_wra_ra_sra}.
Moreover, again due to \cref{lem:wra_ra_sra_completeness_monotonic_hb_paths}, it suffices to only consider $\hb$-paths $\wt'\LTo{\hb}\rd$ that traverse events in the same phase and step.
The lemma then follows by a laborious but straightforward case-by-base analysis of all triplets appearing in each of the gadgets.
We outline this analysis now.
\begin{compactenum}
\item For every triplet $(\wt, \rd, \wt')$ accessing a location that is written by a single thread (i.e., $\VarYThree$, $\VarYFour$, $\VarYSix$, $\VarYEight$, $\VarZThree$, $\VarZFour$, $\VarZSix$, $\VarZEight$), we have $(\wt, \wt')\in \po$ and thus $(\wt', \rd)\not \in \hb$ by \cref{lem:wra_ra_sra_rf_immediate}.
\item Consider the read $\TReadXOne{v^i_j}{t_3}$ (see \cref{fig:copy_gadget_wra_ra_sra}).
Regardless of which thread it reads from (i.e., thread $t_1$ if $s_j=\bfalse$ or thread $t_2$ if $s_j=\btrue$), the write in the opposite thread does not have an $\hb$-path to $\TReadXOne{v^i_j}{t_3}$ that is contained in phase $i$ and step $j$.
Similarly for the read $\TReadXTwo{v^i_j}{t_6}$.
\item Consider the read $\TReadYFive{v^i_j}{f_1}$ in $\CopyGadget{i}{j}$ (see \cref{fig:copy_gadget_wra_ra_sra}).
Regardless of which thread it reads from (i.e., thread $f_2$ if $s_j=\bfalse$ or thread $f_3$ if $s_j=\btrue$), the write in the opposite thread does not have an $\hb$-path to $\TReadYFive{v^i_j}{f_1}$ that is contained in phase $i$ and step $j$.
Similarly for the read $\TReadYSeven{v^i_j}{f_4}$, as well as for the reads $\TReadZFive{v^i_j}{g_2}$ and $\TReadZSeven{v^i_j}{g_4}$ in the copy-down gadget $\CopyGadgetDown{i}{j}$ (see \cref{fig:copy_gadget_down_wra_ra_sra}).
\item Consider the read $\TReadYOne{v^i_j}{t_3}$ in the copy gadget $\CopyGadget{i}{j}$ (see \cref{fig:copy_gadget_wra_ra_sra}).
If $s_j=\bfalse$ (see \cref{subfig:copy_gadget_wra_ra_sraFalse}) we have $(\TWriteYOne{v^i_j}{f_4}, \TReadYOne{v^i_j}{t_3})\in \rf$ and $(\TWriteYOne{\ov{v}^i_j}{t_1}, \TWriteYOne{v^i_j}{f_4})\in \mopartial$, which also implies $(\TWriteYOne{v^i_j}{t_1}, \TWriteYOne{v^i_j}{f_4})\in (\hb\cup \mopartial)^+$.
On the other hand, if $s_j=\btrue$ (see \cref{subfig:copy_gadget_wra_ra_sraTrue}), we have $(\TWriteYOne{v^i_j}{t_1}, \TReadYOne{v^i_j}{t_3})\in \rf$.
Observe that there is no $\hb$-path $\TWriteYOne{v^i_j}{f_4}\LPath{\hb}\TReadYOne{v^i_j}{t_3}$ or $\TWriteYOne{\ov{v}^i_j}{t_1}\LPath{\hb}\TReadYOne{v^i_j}{t_3}$ contained in this gadget.
The former case is straightforward by just observing \cref{subfig:copy_gadget_wra_ra_sraTrue}, while the latter case follows by \cref{lem:wra_ra_sra_rf_immediate}.
A similar analysis holds for the read $\TReadYTwo{v^i_j}{t_3}$, as well as for the reads $\TReadZOne{v^i_j}{t_3}$ and $\TReadZTwo{v^i_j}{t_3}$ in the copy-down gadget $\CopyGadgetDown{i}{j}$ (see \cref{fig:copy_gadget_down_wra_ra_sra}).
\item Consider the read $\TReadA{v^i_j}{c}{t_3}$ in the at-most-one-true gadget $\AtMostOneGadget{c}{i}{j}{k}$ (see \cref{fig:at_most_one_gadget_wra_ra_sra}).
If $s_j=\true$ (see \cref{subfig:at_most_one_gadget_wra_ra_sra_true_false}), we have $(\TWriteA{v^i_j}{c}{h_{c}},\TReadA{v^i_j}{c}{t_3})\in \rf$ and $(\TWriteA{\ov{v}^i_j}{c}{t_2}, \TWriteA{v^i_j}{c}{h_{c}})\in \mopartial$, which also implies  $(\TWriteA{v^i_j}{c}{t_2}, \TWriteA{v^i_j}{c}{h_{c}})\in (\hb\cup \mopartial)^+$.
Moreover, all other writes to $\VarA{c}$ are of a different step.
On the other hand, if $s_j=\bfalse$ (see \cref{subfig:at_most_one_gadget_wra_ra_sra_false_true,subfig:at_most_one_gadget_wra_ra_sra_false_false}), we have $(\TWriteA{v^i_j}{c}{t_2}, \TReadA{v^i_j}{c}{t_3})\in \rf$.
Observe that in this case, all other writes on $\VarA{c}$ are of a different step, 
except $\TWriteA{\ov{v}^i_j}{c}{t_2}$.
Due to \cref{lem:wra_ra_sra_rf_immediate}, we have $(\TWriteA{\ov{v}^i_j}{c}{t_2}, \TReadA{v^i_j}{c}{t_3})\not \in \hb$.
A similar analysis holds for  $\TReadA{v^i_k}{c}{t_3}$, as well as the reads $\TReadB{v^i_j}{t_3}$, $\TReadB{v^i_k}{t_3}$ and $\TReadB{v^i_{\ell}}{t_3}$ in the at-least-one-gadget $\AtLeastOneGadget{i}{j}{k}{\ell}$ (see \cref{fig:at_least_one_gadget_wra_ra_sra}).
\end{compactenum}
\end{proof}

%% file: app_lower_other_models.tex
\section{Proofs from \cref{SEC:OTHER_MODELS}}\label{SEC:APP_OTHER_MODELS}

\subsection{Proofs from \cref{SUBSEC:CAUSAL_MEMORY}}\label{SUBSEC:APP_CAUSAL_MEMORY}

\lemcmrfimmediate*
\begin{proof}
Assume towards contradiction that there exists an $\hb$-path $P\colon \wt\LPath{\hb}\event$.
Note that in our construction, every write is read at most once, and always by a read in a different thread.
Thus $P$ must have the form $P\colon \wt \LTo{\po}\wt'\LPath{\hb}\event$, implying that $(\wt', \rd)\in \hb$.
However, the latter is forbidden by \cref{lem:wra_ra_sra_rf_immediate}. 
\end{proof}

\lemcmob*
\begin{proof}
It suffices to show that for every event $\event$ of thread $t$ and conflicting triplet $(\wt, \rd, \wt')$ such that 
(i)~$(\wt',\rd)\in \obOne{\event}$ and
(ii)~$(\rd, \event)\in \po^?$,
we have $(\wt', \wt)\in \obOne{\event}$.
Indeed, this implies that $\obOne{\event}$ is the fixpoint of the definition of $\ob{\event}$.
Given a pair of writes ($\wt_1$ $\wt_2$) with $(\wt_1,\wt_2)\in\obOne{\event}$, we say that $\wt_1\LTo{\obOne{\event}}\wt_2$ is a \emph{true edge} if $(\wt_1, \wt_2)\not \in \hb$.
Let $P\colon \wt'\LPath{\obOne{\event}}\rd$ be a $\obOne{\event}$-path that contains the smallest number of true edges, and we argue that it contains $0$ true edges (i.e., it is an $\hb$-path).
In turn, this implies that $(\wt', \wt)\in \obOne{\event}$.

Assume towards contradiction otherwise, and let $\wt_1\LTo{\obOne{\event}}\wt_2$ be the last true edge in $P$.
Thus there is a triplet $(\wt_2, \rd_2, \wt_1)$ such that 
(i)~$\tid(\rd_2)=t$, and
(ii)~there is an $\hb$-path $P'\colon \wt_1\LPath{\hb}\rd_2$.
Observe that by construction $(\wt_2, \rd_2)\not\in \po$, as no write is read by a read in the same thread.
Let $\event'$ be the first event of thread $t$ in $P$ after $\wt_2$. Clearly such an $\event'$ exists 
since $w_2$ is an event in $P$ and $P$ culminates in $\rd$ with $\tid(\rd)=t$. 
Then it must be that   $(\event', \rd)\in \po$ : if not, we have $(\rd, \event')\in \po$, and this results in the   $\hb$-cycle
$\event'\LPath{\hb}\rd\LTo{\po} \event'$.
Due to \cref{lem:cm_rf_immediate}, we have that $\event'=\rd_2$.
Then, we can replace the subpath of $P$ between $\wt_1$ and $\rd_2$ with $P'$ thereby obtaining a $\obOne{\event}$-path $\wt'\LPath{\obOne{\event}}\rd$ traversing a smaller number of true edges, a contradiction.
Thus, $(\wt',\rd) \in \hb$, and hence $(\wt', \wt)\in \obOne{\event}$, as desired.
\end{proof}

\subsection{Proofs from \cref{SUBSEC:IMPLICATIONS_TSO_PSO}}\label{SUBSEC:APP_IMPLICATIONS_TSO_PSO}

For the consistency problem on an abstract execution $\expartial=(\E, \po)$,
the semantics of $\tsomm$ can be represented operationally as a labeled transition system $\LTS_{\tsomm}$.
A state in $\LTS_{\tsomm}$ is a triplet $\langle \Executed, \Buffers, \Memory\rangle$:
\begin{compactenum}
\item $\Executed\subseteq \E$ is the set of events that have been executed so far.
\item $\Buffers\colon \Threads\to (\W)^*$ maps every thread $t$ to a sequence of writes $\wt(t,x_1),\wt(t,x_2),\dots, \wt(t,x_i)$, which represents the state of the buffer of thread $t$.
\item $\Memory\colon \Locations\to\W$ maps every memory location of the shared memory to the most recent write to it.
\end{compactenum}
Given a set $\Executed\subseteq \E$, an event $\event$ is \emph{enabled} in $\Executed$, written $\Enabled(\Executed, \event)$, if $\event(t,x)\not\in \Executed$  and every $\po$-predecessor $\event'$ of $\event$ (i.e., $(\event', \event)\in \po$) is such that $\event'\in \Executed$.
\cref{fig:tso_semantics} lists the transitions in $\LTS_{\tsomm}$, which are of the following types.
\begin{compactenum}
\item ${\small [\textsc{Buffer Write}]}$ defines the execution of a write $\wt(t,x,v)$, which modifies the buffer of $t$ .
\item ${\small [\textsc{Memory Write}]}$ defines the flushing of a buffer write $\wt(t,x,v)$ to the shared memory.
\item ${\small [\textsc{Buffer Read}]}$ defines the execution of a read $\rd(t,x,v)$, reading from the buffer of $t$.
\item ${\small [\textsc{Memory Read}]}$ defines the execution of a read $\rd(t,x,v)$, reading from the shared memory.
\end{compactenum}
Naturally, $\expartial\models \tsomm$ iff a state $\langle\E, \Buffers,\Memory \rangle$ is reachable from the initial state $\langle \emptyset, \lambda t. \epsilon, \lambda x. 0 \rangle$ in $\LTS_{\tsomm}$.
\input{figures/tso_algo}

\thmuppertsopso*
\begin{proof}
We only argue about $\tsomm$, as the case of $\psomm$ is similar.
Our proof is by bounding the size of $\LTS_{\tsomm}$ by a polynomial (in the size of $\expartial$) bound.
Indeed:
\begin{compactenum}
\item Since $\Executed$ is downward closed for $\po$, there are $\leq \NumEvents^{\NumThreads}$ different instantiations of $\Executed$.
\item For every thread $t$, $\Buffers(t)$ is a contiguous sequence of writes of thread $t$.
Hence there are $\leq \NumEvents^{2}$ different instantiations of $\Buffers(t)$, leading to $\leq \NumEvents^{2\cdot \NumThreads}$ different instantiations of $\Buffers$.
\item For every memory location $x$ and given an instantiation of $\Executed$ and $\Buffers$, $\Memory(x)$ can hold $\leq \NumThreads$ different writes, one from each thread $t$, which is uniquely determined as the latest (wrt $\po$) write of thread $t$ on $x$ that is in $\Executed$ but not in $\Buffers(t)$.
Hence, given an instantiation of $\Executed$ and $\Buffers$, there are $\leq \NumThreads^{\NumLocations}$ instantiations of $\Memory$.
\end{compactenum}
It follows from the above that the number of states of $\LTS_{\tsomm}$ is bounded by $\NumThreads^{\NumLocations}\cdot \NumEvents^{O(\NumThreads^{2})}$, and thus polynomial when $\NumThreads, \NumLocations=O(1)$.
\end{proof}

Finally, we remark that \cref{thm:upper_tso_pso} also holds in the presence of fences and atomic read-modify-write (RMW) events, by a very similar argument.

%% file: figures/tso_algo.tex
\begin{figure}
\small
\begin{subfigure}[b]{0.45\textwidth}
\begin{align*}
&{\small [\textsc{Buffer Write}]}\\
&\frac{
\begin{gathered}
\Enabled(\Executed, \wt(t,x,v))\\ 
\Executed'=\Executed\cup \{\wt(t,x,v)\}\\
\Buffers'=\Buffers[t\mapsto \Buffers(t)\cdot \wt(t,x,v)]
\end{gathered}
}
{
\langle \Executed, \Buffers, \Memory \rangle \LTo{\wt(t,x,v)} \langle \Executed' , \Buffers',  \Memory\rangle
}
\end{align*}
\end{subfigure}
\begin{subfigure}[b]{0.45\textwidth}
\begin{align*}
&{\small [\textsc{Memory Write}]}\\
&\frac{
\begin{gathered}
\Buffers(t)=\wt_1, \dots, \wt_{i}\\
\Buffers'=\Buffers[t\mapsto \wt_2,\dots, \wt_i]\\
\Memory'=\Memory[x\mapsto \wt_1]
\end{gathered}
}
{
\langle \Executed, \Buffers, \Memory \rangle \LTo{b(t)} \langle \Executed , \Buffers',  \Memory'\rangle
}
\end{align*}
\end{subfigure}
\\[2em]
\begin{subfigure}[b]{0.45\textwidth}
\begin{align*}
&{\small [\textsc{Buffer Read}]}\\
&\frac{
\begin{gathered}
\Enabled(\Executed, \rd(t,x,v))\\ 
\Executed'=\Executed\cup \{\rd(t,x,v)\}\\
\Buffers(t)=\wt_1,\dots, \wt_i\\
\exists j\in[i]. \wt_j=\wt(t,x,v) \text{ and } \forall \ell> j, \wt_{\ell} \neq \wt(t,x,u)
\end{gathered}
}
{
\langle \Executed, \Buffers, \Memory \rangle \LTo{\rd(t,x,v)} \langle \Executed' , \Buffers,  \Memory\rangle
}
\end{align*}
\end{subfigure}
\begin{subfigure}[b]{0.45\textwidth}
\begin{align*}
&{\small [\textsc{Memory Read}]}\\
&\frac{
\begin{gathered}
\Enabled(\Executed, \rd(t,x,v))\\ 
\Executed'=\Executed\cup \{\rd(t,x,v)\}\\
\Buffers(t)=\wt_1,\dots, \wt_i\\
\not\exists j\in[i]. \wt_j=\wt(t,x,\cdot)\\
\Memory(x)=\wt(f,x,v)
\end{gathered}
}
{
\langle \Executed, \Buffers, \Memory \rangle \LTo{\rd(t,x,v)} \langle \Executed' , \Buffers,  \Memory\rangle
}
\end{align*}
\end{subfigure}
\caption{\label{fig:tso_semantics}
The operational semantics of TSO.
}
\end{figure}